\documentclass[11pt]{article}
\usepackage[usenames,dvipsnames,svgnames,table]{xcolor} % usenames : 16 basic colors,  dvipsnames : other 68 colors, svgnames : another 150 colors. It is used instead of \usepackage{color}. Can find descriptions for color in xcolor package. 

\usepackage{enumitem} % for \begin{description}
\usepackage{rotfloat} % For [H] option in sidewaystable
\usepackage{graphicx}
\usepackage{epsfig}
\usepackage{amssymb, amsmath}
\usepackage{verbatim}
\usepackage{natbib}
\usepackage{authblk}
\usepackage{kotex}
\usepackage{multirow}
\usepackage[colorlinks]{hyperref} % it is needed for todonotes package. 
\usepackage[colorinlistoftodos, textsize=scriptsize]{todonotes} % For making notes. disable option removes all notes. 
\usepackage{subcaption} % for \begin{subtable}

\newtheorem{theorem}{Theorem}[section]
\newtheorem{lemma}[theorem]{Lemma}

\newenvironment{proof}[1][Proof]{\begin{trivlist}
\item[\hskip \labelsep {\bfseries #1}]}{\end{trivlist}}

\newcommand{\bea}{\begin{eqnarray*}}
\newcommand{\eea}{\end{eqnarray*}}
\newcommand{\bean}{\begin{eqnarray}}
\newcommand{\eean}{\end{eqnarray}}

\newcommand{\bfX}{{\bf X}}
\newcommand{\bfY}{{\bf Y}}
\newcommand{\bfZ}{{\bf Z}}

\newcommand{\V}{{\rm Var}}

\newcommand{\sg}{\Sigma}
\newcommand{\what}{\widehat}

\newcommand{\lra}{\longrightarrow}

\newcommand{\calT}{\mathcal{T}}
\newcommand{\calU}{\mathcal{U}}

\newcommand{\bbP}{\mathbb{P}} 
\newcommand{\bbR}{\mathbb{R}}
\newcommand{\bbE}{\mathbb{E}}

\begin{document}

\title{Bayesian Optimal Two-sample Tests in High-dimension}
\author[1]{Kyoungjae Lee}
\author[2]{Kisung You}
\author[3]{Lizhen Lin}
\affil[1]{Department of Statistics, Sungkyunkwan University}
\affil[2]{Department of Genetics, Yale University}
\affil[3]{Department of Applied and Computational Mathematics and Statistics, The University of Notre Dame}

\maketitle
\begin{abstract}
	We propose optimal Bayesian two-sample tests for testing equality of high-dimensional mean vectors and covariance matrices between two populations.
	In many applications including genomics and medical imaging,  it is natural  to assume that only a few entries of two mean vectors or covariance matrices are different. Many existing tests that rely on aggregating the difference between empirical means or covariance matrices are not optimal or  yield low power under such  setups.
	Motivated by this, we develop Bayesian two-sample tests employing a divide-and-conquer idea, which is powerful especially when the difference between two populations is sparse but large.
	The proposed two-sample tests manifest closed forms of Bayes factors and allow scalable computations even in high-dimensions.
	We prove that the proposed tests are consistent under relatively mild conditions compared to existing tests in the literature.
	Furthermore,  the testable regions from the proposed tests turn out to be optimal in terms of rates.
	Simulation studies show clear advantages of the proposed tests over other state-of-the-art methods in various scenarios. Our tests are also applied to the analysis of the gene expression data of two cancer data sets. 
\end{abstract}

Key words: Bayesian hypothesis test; 
Bayes factor consistency; 
high-dimensional covariance matrix; 
optimal high-dimensional tests.

%%%%%%%%%%%%%%%%%%%%%%%%%%%%%%%%%%%%%%%%%%%%%
\section{Introduction}\label{sec:intro}

Consider two samples of observations from high-dimensional normal models 
\begin{equation}
	\begin{split}\label{model_mean_cov}
		X_1,\ldots, X_{n_1} \mid \mu_1, \sg_1 &\,\,\overset{i.i.d.}{\sim}\,\, N_p(\mu_1, \sg_1) , \\
		Y_1,\ldots, Y_{n_2} \mid \mu_2, \sg_2 &\,\,\overset{i.i.d.}{\sim}\,\, N_p(\mu_2,  \sg_2), 
	\end{split}
\end{equation}
where $N_p(\mu,\sg)$ is the $p$-dimensional normal distribution with mean vector $\mu\in\bbR^p$ and covariance matrix $\sg \in \bbR^{p\times p}$, and the number of variables $p$ can increase to infinity as the sample sizes ($n_1$ and $n_2$) grow.  Given two samples of such observations, there is abundant interest in testing the homogeneity between two populations through testing  the equality of high-dimensional mean vectors or covariance matrices with applications in medical imaging, genetics and biology \citep{tsai2009multivariate,shen2011shrinkage}. Although there is an emerging literature on high-dimensional hypothesis testing, most of the literature has focused on proposing frequentist testing statistics with relatively little work on developing Bayesian hypothesis tests in particular for high-dimensional problems.  Bayesian tests, which typically are based on Bayes factor with appropriate design of prior distributions for the model under the null and the alternative operate differently from their frequentist counterparts, and there is independent interest in developing Bayesian testing approaches.   We add to the limited literature by developing powerful and scalable Bayesian high-dimensional tests for testing the equality of means and covariance matrices between two populations.

Our initial focus is on the two-sample mean test, where we assume $\sg_1= \sg_2$ and test whether $\mu_1=\mu_2$ in model \eqref{model_mean_cov}.
When $\mu_1 \neq \mu_2$, we call the nonzero elements in the  mean difference vector  $\mu_1-\mu_2 \in \bbR^p$ the {\it signals}.
It is well known that the types of tests with good power are different depending on the number and magnitude of the signals. 
%It is well-known that the signal pattern such as the number and magnitude the signals play an important role  in determining what a powerful test is for each pattern.
From a frequentist perspective, \cite{bai1996effect} and \cite{srivastava2008test} proposed high-dimensional two-sample mean tests based on estimators of $\| A (\mu_1 - \mu_2)\|_2^2$ for some positive definite matrix $A \in \bbR^{p\times p}$.
We call these tests  $L_2$-type tests because their test statistics involve the $L_2$-norm.
It is known that $L_2$-type tests tend to have good power when there are dense signals, i.e., when  a large portion of $\mu_{1}-\mu_{2}$ is nonzero.
When there are many but small signals, $L_2$-type tests  tend to show better performance over other types of tests.

In many applications, however, it is more natural to assume rare signals, where only few entries of $\mu_{1}-\mu_{2} \in \bbR^p$ are nonzero.
Under the presence of rare but significant signals, it is well known that maximum-type tests tend to outperform $L_2$-type tests.
Here, a maximum-type test refers to a class of tests whose test statistic involves the maximum-norm.
\cite{cai2014two} proposed a consistent maximum-type test for high-dimensional two-sample mean test.
They standardized the difference between sample mean vectors using an estimated precision matrix based on either the constrained $\ell_1$-minimization for inverse matrix estimation (CLIME)  \citep{cai2011constrained} or the inverse of the adaptive thresholding estimator for a covariance matrix \citep{cai2011adaptive}.
Because their test statistics depend on an estimated precision matrix,  practical performance of the tests could be impacted by performance of the estimated precision matrix.

Besides the aforementioned papers, many  other interesting studies have been conducted for the two-sample testing setup. 
\cite{gregory2015two} proposed a two-sample mean test which bypasses the needs of the estimation of precision matrix and is  robust to highly unequal covariance matrices between two populations.
% gregory2015two: R package "highD2pop"
\cite{xu2016adaptive} proposed an adaptive two-sample mean test that retains high power against a wide range of alternatives.
\cite{cao2018two} developed a test for compositional data based on the centered log-ratio transformation.
Recently, \cite{wang2019robust} suggested a robust version of the maximum-type test for contaminated data. 
%which contains the test in  \cite{cai2014two} as a special case.

Our second focus is the two-sample covariance test of whether $\sg_1=\sg_2$ or not in model \eqref{model_mean_cov} under the assumption $\mu_1=\mu_2=0$.
In this case, we call the nonzero entries in $\sg_1 - \sg_2 \in \bbR^{p\times p}$ the signals.
Some frequentist tests have also been suggested in the literature for two-sample covariance in high-dimensional settings.
\cite{schott2007test} and \cite{li2012two} proposed to test equality of covariance matrices based on an estimator of $\| \sg_{1} - \sg_{2}\|_F^2$, whereas \cite{srivastava2010testing} suggested a test based on a consistent estimator of $tr(\sg_{1}^2)/\{tr(\sg_{1})\}^2 - tr(\sg_{2}^2)/\{tr(\sg_{2})\}^2$.
These tests can be categorized as $L_2$-type tests.
%\cite{li2012two} suggested a similar test statistic using an unbiased estimator of $\|\sg_{01}-\sg_{02}\|_F^2$ and showed consistency of the proposed test.
%They assumed a bit complicated condition, $tr(\sg_{0k} \sg_{0l}) \to \infty$ and $tr\{(\sg_{0i}\sg_{0j})(\sg_{0k}\sg_{0l}) \} = o ( tr(\sg_{0i}\sg_{0j}) tr(\sg_{0k}\sg_{0l}) ) $  for any $i,j,k,l \in\{1,2\}$.
%Furthermore, to have power tending to one, it is assumed that $\|\sg_{01}-\sg_{02}\|_F^2 \gg p/n$ under the true alternative. 
%\cite{yang2017weighted} 
A two-sample covariance test based on super-diagonals was proposed by \cite{he2018high} whose test turned out to be more powerful than other existing tests when $\sg_{1}$ and $\sg_{2}$ have bandable structures.
However, the aforementioned tests target dense signals, where most of components of $\sg_1-\sg_2$ are nonzero.
Thus, they might be less powerful under the rare signals setting, where only a few entries in $\sg_{1}-\sg_{2}$ are nonzero.
To obtain good power when there are rare signals, \cite{cai2013two} proposed a maximum-type test for two-sample covariance test.
Similar to two-sample mean test in \cite{cai2014two}, \cite{cai2013two} standardized the difference between sample covariances and took the maximum over the standardized sample covariances.
%However, some strong conditions on $p$ and the true covariance matrices were required.
Recently, \cite{zheng2017testing} combined the two tests in \cite{li2012two} and \cite{cai2013two} by taking weighted average to handle both dense and sparse alternatives.
%However, they assumed a moderate high-dimensional setting, where $p/n \lra c \in (0, \infty)$ holds.
%Our goal is to construct a consistent and scalable Bayesian test under much weaker assumptions than used in \cite{cai2013two}.

Bayesian hypothesis testing has very different characteristics from those of its frequentist counterpart,  thus it is important and of an independent interest to develop Bayesian tests for the above hypothesis testing problems. However, up to our knowledge, no theoretically supported Bayesian method has been proposed for high-dimensional two-sample tests, except a recent work of \cite{zoh2018powerful}.
They proposed a Bayesian test for high-dimensional two-sample mean test by reducing the dimension of data via random projections.
They proved  consistency of the proposed Bayesian test under the joint distribution of {\it data and prior}, where the true mean vector is a random variable from the prior distribution.
%However, in most cases, it may be hard to ensure that the unknown parameter follows a certain prior that we chose.

In this paper, we develop scalable Bayesian two-sample tests supported by theoretical guarantees. 
Since rare signals can be more realistic in many applications, our goal is to develop a consistent Bayesian test achieving good power when there are rare signals.
To this end, we apply the maximum pairwise Bayes factor approach suggested by \cite{lee2018maximum}, which is essentially a divide-and-conquer idea.
Rather than comparing the whole mean vectors or covariance matrices at once, we divide them into smaller pieces and reformulate the original testing problem into a multiple testing problem. 
The proposed Bayesian tests turn out to be consistent under both null and alternative hypotheses, and especially, attain good power when there are rare but significant signals.
We prove consistency of the Bayesian tests in different context where the true parameter is fixed unknown quantity, which differentiates our results from those in \cite{zoh2018powerful}.
Furthermore, the proposed tests achieve theoretical and practical improvements compared to those in the existing tests, which will be stated later in more detail.

Although we employ the general idea of modularization by \cite{lee2018maximum}, the former work however only focuses on one-sample testing of the structure of covariance matrices.     
Substantial new developments have been made in this work which differs in terms of problem setup, prior choice, theory development as well as computational approaches.
The main contributions of this paper can be summarized as follows.
% 1
The proposed Bayesian tests are scalable with simple implementations that can be readily used by practitioners.
%We carefully apply the modularization idea of \cite{lee2018maximum} to two-sample testing problems.
It accelerates the computation speed by circumventing computational issues such as inversion of a large matrix.
% 2
Furthermore, up to our knowledge, these are the first results on Bayes factor consistency in high-dimensional two-sample testings.
We prove that the proposed Bayesian tests are consistent under both null and alternative under mild conditions (Theorems \ref{thm:two_mean} and \ref{thm:two_cov}).
The proposed tests have the desired property of being  much more powerful than $L_2$-type tests under rare signals settings.
%, as wanted.
% 3
Besides the development of new Bayesian methods, our proposal also improves state-of-the-art methods theoretically and empirically. We show that the derived testable regions from the proposed tests are optimal in terms of rates (Theorem \ref{thm:two_cov_lowerbound}), and the required conditions for achieving the theoretical results are much weaker than those used in existing literature.
Furthermore, although there are existing frequentist maximum-type tests \citep{cai2014two,cai2013two}, the proposed tests in this paper  outperform the contenders in various settings.
%The proposed two-sample mean test attains comparable performances to the test suggested by \cite{cai2014two}, while the latter heavily depends on the CLIME estimator and the optimal procedure is based on cross-validation incurring extensive computational issue. 
%The proposed two-sample covariance test outperforms the test suggested by \cite{cai2013two} in various simulation settings.

The rest of paper is organized as follows.
Sections \ref{sec:mean} and \ref{sec:cov} present the proposed Bayesian two-sample tests for mean vectors and covariance matrices, respectively.
In Section \ref{sec:numerical}, the practical performance of the proposed methods is evaluated based on numerical study.
Concluding remarks are given in Section \ref{sec:disc}, and proofs of the main results are included in the supplementary material.

%%%%%%%%%%%%%%%%%%%%%%%%%%%%%%%%%%%%%%%%%%%%%
\section{Two-sample mean test}\label{sec:mean}

\subsection{Notation}\label{subsec:notation}

For any given constants $a$ and $b$, we denote the maximum and minimum between the two by $a\vee b$ and $a\wedge b$.
For a vector $x=(x_1,\ldots, x_p)^T$ and a positive integer $q$, we denote the vector $\ell_q$-norm as $\|x\|_q = \big(\sum_{j=1}^p x_j^q \big)^{1/q}$.
For any positive sequences $a_n$ and $b_n$, $a_n\ll b_n$, or equivalently $a_n = o(b_n)$, means that $a_n /b_n \lra 0$ as $n\to\infty$.
We denote $a_n= O(b_n)$ if there exists a constant $C>0$ such that $a_n/b_n \le C$ for all large $n$, and $a_n\asymp b_n$ means that $a_n=O(b_n)$ and $b_n = O(a_n)$.
For a given matrix $A\in \bbR^{p\times p}$, we denote the Frobenius norm $\|A\|_F = \big( \sum_{i=1}^p \sum_{j=1}^p a_{ij}^2 \big)^{1/2}$, the matrix $\ell_1$-norm  $\|A\|_1 = \sup_{x\in \bbR^p , \|x\|_1=1 } \|Ax\|_1$, the spectral norm $\|A\| = \sup_{x\in \bbR^p , \|x\|_2=1 }\|Ax\|_2$, and the matrix maximum norm $\|A\|_{\max} = \max_{1\le i\le j\le p}|a_{ij}| $.
The maximum and minimum eigenvalues of a matrix $A$ are denoted by $\lambda_{\max}(A)$ and $\lambda_{\min}(A)$, respectively.
For given positive numbers $a$ and $b$, $IG(a,b)$ denotes the inverse-gamma distribution with shape parameter $a$ and rate parameter $b$.

Throughout the paper, we assume that $0<\epsilon_0 <1$, $C_1 = 1+\epsilon, C_2 = 2 + \epsilon$ and $C_3= 3 + \epsilon$ for some small constant $\epsilon>0$ in that $C_1, C_2$ and $C_3$ are arbitrarily close to $1,2$ and $3$, respectively.

\subsection{Maximum pairwise Bayes factor for two-sample mean test}\label{subsec:mxPBF_mean}

Suppose that we observe the data from two populations
\begin{equation}
	\begin{split}\label{two_mean}
		X_1,\ldots, X_{n_1} \mid \mu_1, \sg &\,\,\overset{i.i.d.}{\sim}\,\, N_p (\mu_1, \sg) , \\
		Y_1,\ldots, Y_{n_2} \mid \mu_2, \sg &\,\,\overset{i.i.d.}{\sim}\,\, N_p (\mu_2, \sg),
	\end{split}
\end{equation}
where $\mu_1,\mu_2 \in \bbR^p$ and $\sg$ is a $p\times p$ covariance matrix.
Let $\bfX_{n_1} = (X_1,\ldots,X_{n_1})^T \in \bbR^{n_1 \times p}$ and $\bfY_{n_2}=(Y_1,\ldots,Y_{n_2})^T\in \bbR^{n_2 \times p}$ be the data matrices for each population.
We are interested in the testing problem
\bean\label{two_mean_test}
H_0: \mu_1 = \mu_2 \quad\text{ versus }\quad H_1: \mu_1 \neq \mu_2 .
\eean
The goal is to test the homogeneity between two populations based on underlying trends.

Bayesian hypothesis tests are typically based on Bayes factors. To construct a Bayes factor for two-sample mean test, marginal likelihoods should be calculated based on priors for each hypothesis.
Using normal priors for mean vectors and the Jeffreys' prior for a covariance matrix, which corresponds to a default choice, the resulting Bayes factor can be calculated in a closed form when $1<p<n-2$.
See \cite{zoh2018powerful} for the details.
However, the Bayes factor under such priors involves the inverse of a pooled sample covariance matrix, which prevents one from using when $p \ge n-2$.
\cite{zoh2018powerful} suggested projecting the data to a lower-dimensional subspace to reduce the dimensionality.

In this paper, we apply the maximum pairwise Bayes factor (mxPBF) approach suggested by \cite{lee2018maximum}.
Specifically, we compare two mean vectors by comparing them element-by-element. For a given integer $1\le j \le p$, let $\tilde{X}_j = (X_{1j},\ldots, X_{n_1 j})^T$ and $\tilde{Y}_j = (Y_{1j},\ldots, Y_{n_2 j})^T$ be the $j$th columns of $\bfX_{n_1}$ and $\bfY_{n_2}$, respectively.
From model \eqref{two_mean}, we have the following marginal models
\bea
\tilde{X}_j \mid \mu_{1j}, \sigma_{jj} &\sim& N_{n_1}(\mu_{1j}1_{n_1}, \sigma_{jj}I_{n_1} ) , \\
\tilde{Y}_j \mid \mu_{2j}, \sigma_{jj} &\sim& N_{n_2}(\mu_{2j}1_{n_2}, \sigma_{jj}I_{n_2} ) ,
\eea
where $\mu_k= (\mu_{k1},\ldots,\mu_{kp})^T$ for $k=1,2$, $\sg = (\sigma_{ij})$ and $1_q = (1,\ldots, 1)^T \in \bbR^q$.
The hypothesis testing problem \eqref{two_mean_test} can be reformulated as
\bea
H_{0j}: \mu_{1j} = \mu_{2j} \quad\text{ versus }\quad H_{1j}: \mu_{1j} \neq \mu_{2j} ,
\eea
in the sense that $H_0$ is true if and only if $H_{0j}$ is true for all $j=1,\ldots,p$.
Thus, we will first construct Bayesian tests for each testing problem $H_{0j}$ versus $H_{1j}$ and calculate {\it pairwise Bayes factors} (PBFs) based on $(\tilde{X}_j, \tilde{Y}_j)$ for $j=1,\ldots, p$.
For a given $1\le j \le p$, we suggest the following prior $\pi_{0j}(\mu_j , \sigma_{jj})$ under $H_{0j}$,
\bea
\mu_j = \mu_{1j} = \mu_{2j} \mid \sigma_{jj} &\sim& N \Big( \bar{Z}_j, \, \frac{\sigma_{jj}}{n\gamma} \,\Big) , \\
\pi(\sigma_{jj}) &\propto& \sigma_{jj}^{-1} ,
\eea
and the following prior $\pi_{1j}(\mu_{1j}, \mu_{2j} , \sigma_{jj})$ under $H_{1j}$,
\bea
\mu_{1j} \mid \sigma_{jj} &\sim& N \Big( \bar{X}_j, \, \frac{\sigma_{jj}}{n_1\gamma} \,\Big) , \\
\mu_{2j} \mid \sigma_{jj} &\sim& N \Big( \bar{Y}_j, \, \frac{\sigma_{jj}}{n_2\gamma} \,\Big) , \\
\pi(\sigma_{jj}) &\propto& \sigma_{jj}^{-1},
\eea
where $\tilde{Z}_j = (\tilde{X}_j^T, \tilde{Y}_j^T)^T$, $\bar{Z}_j = n^{-1}\sum_{i=1}^n Z_{ij}$, $\bar{X}_j = {n_1}^{-1}\sum_{i=1}^{n_1} X_{ij}$, $\bar{Y}_j = {n_2}^{-1}\sum_{i=1}^{n_2} Y_{ij}$, $n=n_1+n_2$ and $\gamma =  (n\vee p)^{-\alpha}$.
Throughout this paper, we consider $\alpha$ as a fixed positive constant.

For any vector $v$, define the projection matrix $H_v = v (v^T v)^{-1} v^T$.
Let $\what{\sigma}^2_{Z_j} = n^{-1} \tilde{Z}_j^T (I_n -  H_{1_n}) \tilde{Z}_j$, $\what{\sigma}^2_{X_j} = {n_1}^{-1} \tilde{X}_j^T (I_{n_1} -  H_{1_{n_1}}) \tilde{X}_j$ and $\what{\sigma}^2_{Y_j} = {n_2}^{-1} \tilde{Y}_j^T (I_{n_2} -  H_{1_{n_2}}) \tilde{Y}_j$.
Then, the resulting log PBF is
\begin{equation}
	\begin{split}\label{twomean_PBF}
		\log B_{10} (\tilde{X}_j, \tilde{Y}_j) 
		\,\,&:=\,\,  \log  \frac{p(\tilde{X}_j, \tilde{Y}_j \mid H_{1j})}{p(\tilde{X}_j, \tilde{Y}_j \mid H_{0j})}    \\
		\,\,&=\,\, \frac{1}{2}\log \Big(\frac{\gamma}{1+\gamma} \Big) + \frac{n}{2} \log \Bigg( \frac{ n\what{\sigma}^2_{Z_j} }{ n_1 \what{\sigma}^2_{X_j} + n_2 \what{\sigma}^2_{Y_j} }  \Bigg),
	\end{split}	
\end{equation}
where 
\begin{align*}
	p(\tilde{X}_j, \tilde{Y}_j \mid H_{0j}) &= \iint p(\tilde{X}_j \mid \mu_{j},\sigma_{jj}, H_{0j}) p(\tilde{Y}_j \mid \mu_{j},\sigma_{jj}, H_{0j}) \pi_{0j}(\mu_{j},\sigma_{jj}) d \mu_{j} d\sigma_{jj}   , \\
	p(\tilde{X}_j, \tilde{Y}_j \mid H_{1j}) &= \iiint p(\tilde{X}_j \mid \mu_{1j},\sigma_{jj}, H_{1j}) p(\tilde{Y}_j \mid \mu_{2j},\sigma_{jj}, H_{1j}) \pi_{1j}(\mu_{1j},\mu_{2j},\sigma_{jj}) d \mu_{1j} d\mu_{2j} d\sigma_{jj} .	
\end{align*}
To aggregate PBFs for all $j=1,\ldots, p$, we define the mxPBF as
\bean\label{mxPBF_mean}
B_{\max,10}^{\mu}(\bfX_{n_1}, \bfY_{n_2}) &:=& \max_{1\le j\le p} B_{10}(\tilde{X}_j, \tilde{Y}_j) .
\eean
Then one can conduct a Bayesian test by considering the mxPBF as a usual Bayes factor: for a given threshold $C_{\rm th}>0$, we support $H_1:\mu_1 \neq \mu_2$ if $B_{\max,10}^{\mu}(\bfX_{n_1}, \bfY_{n_2})  > C_{\rm th}$.
It is easy to see that $B_{\max,10}^{\mu}(\bfX_{n_1}, \bfY_{n_2})  > C_{\rm th}$ if and only if $B_{10}(\tilde{X}_j, \tilde{Y}_j)  > C_{\rm th}$ for some $1\le j\le p$.
Thus, a Bayesian test based on the mxPBF supports $H_1:\mu_1\neq \mu_2$ if and only if there is at least one strong evidence in favor of  $H_{1j}:\mu_{1j}\neq \mu_{2j}$.

\subsection{Bayes factor consistency}\label{subsec:mxPBF_consistency_mean}

A mxPBF is said to be consistent if it (i) converges to zero under $H_0$ and (ii) diverges to infinity under $H_1$ in probability.
%under $H_0$ if it goes to zero in probability when $H_0$ is true.
%Furthermore, we say that the mxPBF is consistent under $H_1$ if it diverges to infinity in probability when $H_1$ is true.
%If the mxPBF is consistent under both $H_0$ and $H_1$, we say that the mxPBF is consistent.
Let $\mu_{0,1} \in \bbR^p$ and $\mu_{0,2} \in \bbR^p$ be true mean vectors for each population, respectively, and $\Sigma_0 = (\sigma_{0,ij}) \in \bbR^{p\times p}$ be the true covariance matrix.
Theorem \ref{thm:two_mean} shows that the mxPBF is consistent under mild conditions.

\begin{theorem}\label{thm:two_mean}
	Consider model \eqref{two_mean} and the two-sample mean test $H_0: \mu_1=\mu_2$ versus $H_1:\mu_1\neq \mu_2$.
	Assume that $\log p \le n \epsilon_0$ and
	\bean
	\alpha &>&   \frac{2(1 + \epsilon_0)}{1 - 3 \sqrt{C_1 \epsilon_0}}   \label{alpha_mean}
	\eean
	for some $C_1>1$ and $0<\epsilon_0<1$.
	Then, the mxPBF \eqref{twomean_PBF} is consistent under $H_0$: for some constant $c>0$,
	\bea
	B_{\max,10}^{\mu}(\bfX_{n_1}, \bfY_{n_2})  &=&O_p\{(n\vee p)^{-c}\}  \quad \text{ under $H_0$} .
	\eea
	When $H_1$ is true, assume that there is at least one of indices $1\le j\le p$ satisfying
	\bean
	\frac{n_1 n_2(\mu_{0,1j}-\mu_{0,2j})^2}{n^2 \sigma_{0,jj}} \,\ge\,  \Big[ \sqrt{2C_1} + \sqrt{2C_1 + \alpha C_1 \{ 1+ (1+ 8 C_1)\epsilon_0 \} } \Big]^2 \frac{\log(n\vee p)}{n} . \label{betamin_mean}
	\eean
	Then, the mxPBF is also consistent under $H_1$: for some constant $c'>0$,
	\bea
	\big\{B_{\max,10}^{\mu}(\bfX_{n_1}, \bfY_{n_2})\big\}^{-1}  &=&O_p\{(n\vee p)^{-c'}\}     \quad \text{ under $H_1$}.
	\eea
\end{theorem}

It is worthwhile to compare our result to those of the existing literature.
As mentioned earlier, the test statistic of \cite{cai2014two} depends on an estimated precision matrix that some conditions for consistent estimation of the precision matrix are required.
For example, it was assumed that $\Omega_0$ has bounded eigenvalues and absolute correlations of $X_i$, $Y_i$, $\Omega_0 X_i$ and $\Omega_0 Y_i$ are bounded away from 1, where $\Omega_0$ is the true precision matrix.
Furthermore, $\Omega_0$ is assumed to satisfy $\|\Omega_0\|_1^2 = o( \sqrt{n/(\log p)^3} )$ or stronger sparsity assumption, which essentially means that a large amount of entries in $\Omega_0$ is sufficiently small.
They also assumed that $\mu_{0,1}-\mu_{0,2}$ has at most $p^r$ nonzero entries, where $r \in [0,1/4)$.

On the other hand, theoretical results in Theorem \ref{thm:two_mean} do not require any condition on the true precision matrix and allow the number of nonzero entries in $\mu_{0,1}-\mu_{0,2}$ to have the same order with $p$.
Therefore, we suspect that the mxPBF would perform better than the maximum-type test in \cite{cai2014two} when these conditions are violated.
Indeed, we find empirical evidences for this conjecture in our simulation study in Section \ref{subsec:sim_mean}.

Recently, \cite{zoh2018powerful} proposed a Bayesian two-sample mean test and proved consistency of the Bayes factor in high-dimensional settings.
They used random projections to reduce the dimensionality of the data and assumed that the reduced dimension has the same order with $(n_1\wedge n_2)$.
To conduct a Bayesian test, a single random projection matrix was considered, which can lead to different results depending on the generated projection matrix.
Furthermore, under the true alternative,  no lower bound condition of $\mu_{0,1}- \mu_{0,2}$ was provided to ensure consistency, like condition \eqref{betamin_mean}. 
They assumed that $\mu_1 - \mu_2$ is a random vector under $H_1:\mu_1 \neq \mu_2$ rather than considering a fixed true value  $\mu_{0,1}- \mu_{0,2}$, which differentiates our results from those in \cite{zoh2018powerful}.

We note here that, by Theorem 3  in \cite{cai2014two}, condition \eqref{betamin_mean} is rate-optimal to guarantee the existence of a consistent test when $\log n = O( \log p )$.
Thus, the proposed mxPBF-based test provides an optimal testable region with respect to the maximum norm.
\cite{cai2014two} assumed the condition $\max_j (\mu_{0,1j}-\mu_{0,2j})^2 \ge C \log p/n$ for some constant $C>0$, which is similar to \eqref{betamin_mean}.

%%%%%%%%%%%%%%%%%%%%%%%%%%%%%%%%%%%%%%%%%%%%%
\section{Two-sample covariance test}\label{sec:cov}
In this section, we propose  Bayesian  two-sample tests for testing the equity of high-dimensional covariance matrices and consider their theoretical properties in terms of Bayes factor consistency and optimality of the testing regions. 

\subsection{Maximum pairwise Bayes factor for two-sample covariance test}\label{subsec:mxPBF_cov}

Suppose that we observe the data from two populations
\begin{equation}
	\begin{split}\label{two_cov}
		X_1,\ldots, X_{n_1} \mid \sg_1 &\,\,\overset{i.i.d.}{\sim}\,\, N_p (0, \sg_1) , \\
		Y_1,\ldots, Y_{n_2} \mid \sg_2 &\,\,\overset{i.i.d.}{\sim}\,\, N_p (0, \sg_2),
	\end{split}
\end{equation}
where $\sg_1 = (\sigma_{1,ij})$ and $\sg_2 = (\sigma_{2,ij})$ are $p\times p$ covariance matrices.
In this section, we consider the  testing problem
\bean\label{two_cov_test}
H_0:\sg_1=\sg_2 \quad\text{ versus }\quad H_1:\sg_1\neq \sg_2.
\eean
To apply the mxPBF approach, we need to divide the comparison of two covariance matrices into smaller problems. Among various options for that, we use the reparametrization trick used in \cite{lee2018maximum}.
Specifically, for a given pair $(i,j)$ with $1\le i\neq j \le p$, \eqref{two_cov} induces the conditional distributions
\begin{equation}
	\begin{split}\label{cond_dist_XY}
		\tilde{X}_i \mid \tilde{X}_j, a_{1,ij}, \tau_{1,ij} \,\,&\sim\,\, N_{n_1}\big( a_{1,ij}\tilde{X}_j, \tau_{1,ij}I_{n_1}  \big), \\
		\tilde{Y}_i \mid \tilde{Y}_j, a_{2,ij}, \tau_{2,ij} \,\,&\sim\,\, N_{n_2}\big( a_{2,ij}\tilde{Y}_j, \tau_{2,ij}I_{n_2}  \big),
	\end{split}	
\end{equation}
where $a_{k,ij}= \sigma_{k,ij}/\sigma_{k,jj}$,  $\tau_{k,ij} = \sigma_{k,ii}(1- \rho_{k,ij}^2)$ and $\rho_{k,ij} = \sigma_{k,ij}/ (\sigma_{k,ii} \sigma_{k,jj} )^{1/2}$  for $k=1,2$.
The hypothesis testing problem \eqref{two_cov_test} can be reformulated as 
\bean\label{two_cov_ij_pair}
H_{0,ij}: a_{1,ij}= a_{2,ij} \text{ and } \tau_{1,ij}=\tau_{2,ij} \quad\text{ versus }\quad H_{1,ij}: \text{ not } H_{0,ij},
\eean
in the sense that $H_0$ is true if and only if $H_{0,ij}$ is true for all pairs $(i,j)$, $1\le i\neq j\le p$.

To construct a Bayesian test for testing \eqref{two_cov_ij_pair}, we suggest the following prior distribution $\pi_{0,ij}(a_{ij}, \tau_{ij})$ under $H_{0,ij}$,
\bea
a_{ij}=a_{1,ij}=a_{2,ij}\mid \tau_{ij} &\sim& N\Big( \what{a}_{ij}, \frac{\tau_{ij}}{\gamma \|\tilde{Z}_j\|_2^2 } \Big), \\
\tau_{ij} = \tau_{1,ij} = \tau_{2,ij} &\sim& IG(a_0, b_{0,ij}) ,
\eea
and the prior $\pi_{1,ij}(a_{1,ij}, a_{2,ij},\tau_{1,ij}, \tau_{2,ij})$ under $H_{1,ij}$,
\bea
a_{1,ij}\mid \tau_{1,ij} &\sim& N\Big(\what{a}_{1,ij}, \frac{\tau_{1,ij}}{\gamma\|\tilde{X}_j\|_2^2} \Big), \,\quad
a_{2,ij}\mid \tau_{2,ij} \,\,\sim\,\, N\Big(\what{a}_{2,ij} , \frac{\tau_{2,ij}}{\gamma\|\tilde{Y}_j\|_2^2} \Big), \\
\tau_{1,ij} &\sim& IG(a_0, b_{01,ij}), \,\,\quad \tau_{2,ij} \,\,\sim\,\, IG(a_0, b_{02,ij}) ,
\eea
where $a_0, b_{0,ij}, b_{01,ij}$ and $b_{02,ij}$ are positive constants, $\gamma = (n\vee p)^{-\alpha}$, $\what{a}_{ij} = \tilde{Z}_i^T \tilde{Z}_j / \|\tilde{Z}_j\|_2^2$, $\what{a}_{1,ij} = \tilde{X}_i^T \tilde{X}_j / \|\tilde{X}_j\|_2^2$ and $\what{a}_{2,ij} = \tilde{Y}_i^T \tilde{Y}_j / \|\tilde{Y}_j\|_2^2$.
Let $\what{\tau}_{ij} = n^{-1}\tilde{Z}_i^T(I_n - H_{\tilde{Z}_j} )\tilde{Z}_i$, $\what{\tau}_{1,ij} = n_1^{-1}\tilde{X}_i^T(I_{n_1} -  H_{\tilde{X}_j} )\tilde{X}_i$ and  $\what{\tau}_{2,ij} = n_2^{-1}\tilde{Y}_i^T(I_{n_2} - H_{\tilde{Y}_j} )\tilde{Y}_i$.
The resulting log PBF is given by
\bean
&& \log B_{10}(\tilde{X}_i, \tilde{Y}_i, \tilde{X}_j,\tilde{Y}_j) \label{B10_two_cov} \\
&:=& \log \frac{p(\tilde{X}_i, \tilde{Y}_i \mid \tilde{X}_j, \tilde{Y}_j, H_{1,ij})}{p(\tilde{X}_i, \tilde{Y}_i \mid \tilde{X}_j, \tilde{Y}_j, H_{0,ij})} \nonumber \\
&=& \frac{1}{2}\log \Big(\frac{\gamma}{1+\gamma} \Big) + \log \Gamma\Big(\frac{n_1}{2}+a_0 \Big) + \log \Gamma\Big(\frac{n_2}{2}+a_0 \Big) - \log \Gamma\Big(\frac{n}{2}+a_0 \Big) + \log \Big(\frac{b_{01,ij}^{a_0} b_{02,ij}^{a_0}}{b_{0,ij}^{a_0}\Gamma(a_0)} \Big) \nonumber \\
&-& \Big(\frac{n_1}{2}+a_0 \Big) \log \big(b_{01,ij} +\frac{n_1}{2}\what{\tau}_{1,ij} \big) - \Big(\frac{n_2}{2}+a_0 \Big) \log \big(b_{02,ij} +\frac{n_2}{2}\what{\tau}_{2,ij} \big) \label{leading1_hat} \\
&+& \Big(\frac{n}{2}+a_0 \Big) \log \big(b_{0,ij} + \frac{n}{2}\what{\tau}_{ij} \big), \label{leading2_hat}
\eean
where 
\begin{align*}
	p(\tilde{X}_i, \tilde{Y}_i \mid \tilde{X}_j, \tilde{Y}_j, H_{0,ij}) &= \iint p(\tilde{X}_i \mid \tilde{X}_j, a_{ij},\tau_{ij}, H_{0,ij}) p(\tilde{Y}_i \mid \tilde{Y}_j, a_{ij},\tau_{ij}, H_{0,ij}) \pi_{0,ij}(a_{ij},\tau_{ij}) d a_{ij} d\tau_{ij}   , \\
	p(\tilde{X}_i, \tilde{Y}_i \mid \tilde{X}_j, \tilde{Y}_j, H_{1,ij}) &= \iiiint p(\tilde{X}_i \mid \tilde{X}_j, a_{1,ij},\tau_{1,ij}, H_{1,ij}) p(\tilde{Y}_i \mid \tilde{Y}_j, a_{2,ij},\tau_{2,ij}, H_{1,ij}) \\
	& \quad\quad\quad\quad\quad\quad \times\,\, \pi_{1,ij}(a_{1,ij},a_{2,ij},\tau_{1,ij},\tau_{2,ij}) d a_{1,ij} da_{2,ij}d\tau_{1,ij}d \tau_{2,ij} .
\end{align*}
Then, the mxPBF for two-sample covariance test is given by
\bean\label{mxPBF_two_cov}
B_{\max,10}^{\sg}(\bfX_{n_1}, \bfY_{n_2}) &:=& \max_{i\neq j} B_{10}(\tilde{X}_i, \tilde{Y}_i, \tilde{X}_j,\tilde{Y}_j).
\eean
Similar to the two-sample mean test, one can conduct a Bayesian test by supporting $H_1:\sg_1\neq \sg_2$ if $B_{\max,10}^{\sg}(\bfX_{n_1}, \bfY_{n_2}) > C_{\rm th}$ for a given threshold $C_{\rm th}>0$.

\subsection{Bayes factor consistency}\label{subsec:mxPBF_consistency_cov}

In this section, we show that the mxPBF in \eqref{mxPBF_two_cov} is consistent for high-dimensional two-sample covariance test.
We first introduce sufficient conditions that guarantee consistency of the mxPBF.
The first condition, (A1), roughly means that $p = O (\exp(n^c))$ for some $0<c<1$.
%This condition is much weaker than the condition, $(\log p)^5 = o(n_1 \vee n_2)$, used in \cite{cai2013two}.
\begin{itemize}
	\item[(A1)] $\epsilon_{0k} := \log(n\vee p)/n_k = o(1)$ for $k=1,2$.
\end{itemize}

When $H_0: \sg_1= \sg_2$ is true, we denote $\sg_0$ as the true covariance matrix. 
Furthermore, we define $a_{0,ij} = \sigma_{0,ij}/\sigma_{0,jj}, \tau_{0,ij}= \sigma_{0,ii}(1- \rho_{0,ij}^2)$, and $R_0=(\rho_{0,ij})$ is a correlation matrix.
Condition (A2) is a sufficient condition for consistency under the null $H_0: \sg_1= \sg_2$.
\begin{itemize}
	\item[(A2)] $\min_{i\neq j}\tau_{0,ij} \gg \{\log (n\vee p)\}^{-1}$.
\end{itemize}
Condition (A2) is satisfied if $\min_{1\le i\le p} \sigma_{0,ii} > \epsilon$ and $\max_{i\neq j} \rho_{0,ij}^2 < 1-\epsilon$ for some small constant $\epsilon>0$.
However, in fact, condition (A2) allows more general cases where possibly $\sigma_{0,ii} \to 0$ and $\rho_{0,ij}^2 \to 1$ as $p\to \infty$ at certain rates.

When $H_1:\sg_1\neq \sg_2$ is true, we denote $\sg_{01}=(\sigma_{01,ij})$ and $\sg_{02}=(\sigma_{02,ij})$ as the true covariance matrices for each population.
Furthermore, we define $a_{0k,ij} = \sigma_{0k,ij}/\sigma_{0k,jj}, \tau_{0k,ij}= \sigma_{0k,ii}(1- \rho_{0k,ij}^2)$ and $R_{0k}=(\rho_{0k,ij})$ is a correlation matrix for $k=1,2$.
Under the alternative $H_1:\sg_1\neq \sg_2$, we assume that $(\sg_{01}, \sg_{02})$ satisfies condition (A3) or (A3$^\star$):
\begin{itemize}
	\item[(A3)] 
	There exists a pair $(i,j)$ with $i\neq j$ such that
	\bea
	\{\log (n\vee p)\}^{-1} \ll \tau_{01,ij}\wedge \tau_{02,ij} &\le& \tau_{01,ij}\vee \tau_{02,ij} \ll (n\vee p),
	\eea
	satisfying either
	\bea
	\frac{\tau_{01,ij}}{\tau_{02,ij}} &>&  \frac{1+ C_{\rm bm}\sqrt{\epsilon_{01}}}{1- 4 \sqrt{C_1(\epsilon_{01}\vee \epsilon_{02})}},
	\eea
	or 
	\bea
	\frac{\tau_{02,ij}}{\tau_{01,ij}} &>&  \frac{1+ C_{\rm bm}\sqrt{\epsilon_{02}} }{1- 4 \sqrt{C_1(\epsilon_{01}\vee \epsilon_{02})}},
	\eea 
	for some constants $C_{\rm bm}^2 > 8 (\alpha+1)$ and $C_1>1$.
	
	\item[(A3$^\star$)] There exists a pair $(i,j)$ with $i\neq j$ such that $\sigma_{01,ii} \vee \sigma_{02,ii} \ll (n\vee p)$ and
	\bean
	(a_{01,ij} - a_{02,ij})^2 &\ge& \frac{25}{2} C_1 \sum_{k=1}^2 \Big\{\frac{\tau_{0k,ij}\epsilon_{0k} }{\sigma_{0k,jj}(1-2\sqrt{C_1 \epsilon_{0k}}) }  \Big\} ,  \label{diff_aij1} \\
	(a_{01,ij} - a_{02,ij})^2 &\ge& \frac{10n}{n+2a_0} \sum_{k=1}^2 \Big\{ \frac{\epsilon_{0k}}{ \sigma_{0k,jj}  (1-2\sqrt{C_1\epsilon_{0k}}) }  \Big\} \label{diff_aij2} \\
	&&\times\, \Big[ \frac{b_{0,ij}}{\log (n\vee p)} + \Big\{ \sum_{k=1}^2 \sigma_{0k,ii}(1+4\sqrt{C_1\epsilon_{0k}}) + \frac{2 b_{0,ij}}{n} \Big\} C_{{\rm bm}, a}   \Big] \nonumber
	\eean
	for some constants $C_{{\rm bm}, a} > \alpha + a_0 +1$ and $C_1>1$.
\end{itemize}
Conditions (A3) and (A3$^\star$) may seem complicated at first glance, but it can be transformed into simpler conditions.
For given positive constants $\alpha, C_{\rm bm}$ and $C_{{\rm bm},a}$ such that $C_{\rm bm}^2>8(\alpha+1)$ and $C_{{\rm bm},a}>\alpha+1$, define a class of two covariance matrices
\bea
{H}_1(C_{\rm bm}, C_{{\rm bm},a})   &:=& \Big\{ (\sg_1,\sg_2) : (\sg_1,\sg_2) \text{ satisfies condition (A3) or (A3$^\star$)}  \Big\} .
\eea
Conditions (A3) and (A3$^\star$) specify the minimum difference condition between $\sg_{01}$ and $\sg_{02}$ to consistently detect the alternative $H_1:\sg_1\neq \sg_2$  under the reparametrization using $\{a_{0k,ij},\tau_{0k,ij}: k=1,2 \text{ and } 1\le i\neq j \le p \}$.
%Lemma \ref{lem:testable_two_cov} in Section \ref{sec:auxil} shows how to convert these conditions into conditions on $(\sg_{01}, \sg_{02})$.
Suppose that
\bean\label{cond_delta}
\begin{split}
	\max_{1\le k \le 2}\max_{1\le i\neq j \le p}\rho_{0k,ij}^2 &\le 1 - c_0, \\
	\{\log(n\vee p) \}^{-1} \ll \min_{1\le k \le 2}\min_{1\le i \le p} \sigma_{0k,ii} &\le \max_{1\le k \le 2}\max_{1\le i \le p} \sigma_{0k,ii} \ll (n\vee p),
\end{split}
\eean
for some small constant $c_0>$.
If $\alpha>1$, $n_1\asymp n_2$ and
%our {\it testable region}, where the mxPBF consistently detects the true alternative $H_1: \sg_1\neq \sg_2$, includes 
\bean\label{two_cov_testable}
\begin{split}
	\widetilde{H}_1(C_\star ,c_0) &:= \Big\{ (\sg_{1}, \sg_{2}): \max_{1\le i\le j\le p} \frac{ (\sigma_{1,ij} - \sigma_{1,ij})^2 }{ \sigma_{1,ii}\sigma_{1,jj} + \sigma_{2,ii}\sigma_{2,jj} } \ge C_\star  \, \frac{\log (n\vee p)}{n}, \\
	&\quad\quad\quad\quad\quad\quad\quad\,\, (\sg_1,\sg_2) \text{ satisfies conditions in } \eqref{cond_delta} \text{ with } c_0\,  \Big\} ,
\end{split}
\eean
then $\widetilde{H}_1(C_\star ,c_0) \subset {H}_1(C_{\rm bm}, C_{{\rm bm},a})$ for some large constant $C_\star>0$ by Lemma 3.1 in the supplementary material. 
Condition \eqref{two_cov_testable} characterizes the difference between $\sg_{01}$ and $\sg_{02}$ using the squared maximum {\it standardized difference}. 
Hence, conditions (A3) and (A3$^\star$) can essentially be understood as the squared maximum standardized difference condition given at \eqref{two_cov_testable}.
\cite{cai2013two} also used a similar difference measure between $\sg_{01}$ and $\sg_{02}$.

The following theorem shows consistency of the mxPBF \eqref{mxPBF_two_cov}. We note that the condition $\lim_{(n_1\wedge n_2)\to\infty} n_1/n = 1/2$ in Theorem \ref{thm:two_cov} can be relaxed to  $n_1\asymp n_2$, although constants in conditions (A2), (A3) and (A3$^\star$) should be changed accordingly.
\begin{theorem}\label{thm:two_cov}
	Consider model \eqref{two_cov} and the two-sample covariance test $H_0:\sg_1=\sg_2$ versus $H_1:\sg_1\neq \sg_2$.
	Assume that $\lim_{(n_1\wedge n_2)\to\infty} n_1/n = 1/2$ and condition (A1) holds.
	Then, under $H_0$, if $\alpha > 6C_2 $ and condition (A2) holds, for some constant $c>0$,
	\bea
	B_{\max,10}^{\sg}(\bfX_{n_1}, \bfY_{n_2})  &=&O_p\{(n\vee p)^{-c}\} .
	\eea
	Under $H_1$, if  $(\sg_{01}, \sg_{02}) \in H_1( C_{\rm bm} , C_{{\rm bm},a} )$, for some constant $c'>0$,
	\bea
	\{B_{\max,10}^{\sg}(\bfX_{n_1}, \bfY_{n_2})\}^{-1}  &=&O_p\{(n\vee p)^{-c}\} .
	\eea
\end{theorem}

\cite{cai2013two} considered a high-dimensional setting, $(\log p)^5 = o(n_k)$, while we assume a weaker condition, $\log p = o(n_k)$ for $k=1,2$ (condition (A1)).
For given constants $C>0$ and $0<r<1$, define $s_j(C) = \text{card}\{i: |\rho_{01,ij}| \ge (\log p )^{-1-C} \text{ or } |\rho_{02,ij}| \ge (\log p )^{-1-C}  \}$ and $\Lambda(r) = \{i : |\rho_{01,ij}|  >r \text{ or } |\rho_{02,ij}|>r \text{ for some } j \neq i    \}$, where $\text{card}(A)$ means the cardinality of the set $A$.
\cite{cai2013two} assumed that there exist $\Gamma \subset \{1,\ldots, p \}$, $C>0$ and $0<r<1$ such that $\text{card}(\Gamma) =o(p)$, $\max_{j \neq \Gamma} s_j(C) = o(p^{c})$ for some constant $c>0$, and $\text{card(}\Lambda(r)) = o(p)$.
%Although seeming a bit complicated, 
These conditions essentially restrict the number of highly correlated variables.
They are satisfied if $\lambda_{\max}(R_{01}) \vee \lambda_{\max}(R_{02}) \le C'$ for some constant $C'>0$ and $\|R_{01}\|_{\max} \vee \|R_{02}\|_{\max} \le r <1$. 
%which are stronger than our conditions, (A2), (A3) and (A3$^\star$).
The power of their test tends to one if 
\bea
\max_{1\le i\le j\le p} \frac{ (\sigma_{01,ij} - \sigma_{01,ij})^2 }{ n_1^{-1}\theta_{01,ij} + n_2^{-1} \theta_{02,ij} },
&\ge& C  \log p  
\eea
for $C \ge 4$, where $\theta_{01,ij} = \V( X_{1i}X_{1j} )$ and $\theta_{02,ij} = \V( Y_{1i}Y_{1j} )$.
This condition is equivalent to condition \eqref{two_cov_testable} in terms of the rate.
Thus, compared with those used in \cite{cai2013two}, we obtain consistency of the mxPBF under weaker conditions for $(n,p)$ and similar conditions for true covariance matrices.

One of the interesting findings from Theorem \ref{thm:two_cov} is that the mxPBF does not require any standardization step.
\cite{cai2013two} mentioned that the standardization of the test statistic is necessary to deal with a wide range of variability and heteroscedasticity of sample covariances.

However, the mxPBF \eqref{mxPBF_two_cov} still enjoys consistency for the similar parameter space without standardization.
Although we did not mention earlier,  a similar phenomenon is observed for the two-sample mean test: the proposed mxPBF \eqref{mxPBF_mean} does not require a standardization step while having similar properties with a standardized test.

Another important finding is that condition (A3) (or (A3$^\star$)) is rate-optimal to guarantee consistency under $H_0:\sg_1=\sg_2$ as well as $H_1:\sg_1\neq \sg_2$.
Theorem \ref{thm:two_cov_lowerbound} shows that, for some small constants $C_{\rm bm}$ and $C_{{\rm bm},a} >0$, there is no consistent test having power tending to one for any true alternative satisfying $(\sg_{01},\sg_{02}) \in {H}_1(C_{\rm bm}, C_{{\rm bm},a})$.

\begin{theorem}\label{thm:two_cov_lowerbound}
	Let $\bbE_{\sg_{01},\sg_{02}}$ be the expectation corresponding to model \eqref{two_cov} with $(\sg_{01},\sg_{02})$.
	Suppose that $n_1\asymp n_2$ and $p \ge n^c$ for some constant $c>0$.
	Then, there exists small constants $C_{\rm bm}$ and $C_{{\rm bm},a} >0$ such that for any $0<\alpha_0<1$ and all large $n$ and $p$,
	\bea
	\inf_{(\sg_{01},\sg_{02}) \in {H}_1(C_{\rm bm}, C_{{\rm bm},a})} \sup_{\phi \in \calT} \bbE_{\sg_{01},\sg_{02}} \phi &\le& {\alpha_0} + o(1),
	\eea
	where $\calT$ is the set of consistent tests over the multivariate normal distributions such that $\bbE_{0} \phi \lra 0$ as $n\to\infty$ for any $\phi \in \calT$, and $\bbE_0$ is the expectation corresponding to model \eqref{two_cov} under $H_0:\sg_1=\sg_2$.
	%	Thus, conditions (A3) and (A3$^\star$) are rate-optimal to guarantee consistency under $H_0:\sg_1=\sg_2$ as well as $H_1:\sg_1\neq \sg_2$.
\end{theorem}

%%%%%%%%%%%%%%%%%%%%%%%%%%%%%%%%%%%%%%%%%%%%%
\section{Numerical results}\label{sec:numerical}
\subsection{Simulation study: two-sample mean test}\label{subsec:sim_mean}

In this section, we illustrate performance of the mxPBF for two-sample mean test through simulation studies. 
We generate the data as follows: $X_1,\ldots, X_{n} \overset{i.i.d.}{\sim} N_p( \mu_{01} , \sg_0)$ and $Y_1,\ldots, Y_{n} \overset{i.i.d.}{\sim} N_p( \mu_{02} , \sg_0)$ with $n = 100$ and $p \in \{100, 300\}$.
Under the null hypothesis, $H_0: \mu_{01} = \mu_{02}$, we set $\mu_{01} = \mu_{02} = 0 \in \bbR^p$.
Under the alternative hypothesis, $H_1: \mu_{01} \neq \mu_{01}$, we set $\mu_{01}=0 \in \bbR^p$ and randomly choose $n_0$ entries in $\mu_{02}$, say $\{ \mu_{02,j} : 1\le  j_1 < \cdots <  j_{n_0} \le p \}$, and set $\mu_{02, j} = \mu>0$ for all $j = j_1,\ldots, j_{n_0}$ and $\mu_{02,j}=0$ for the rest.
Thus, $n_0$ and $\mu$ are the number and magnitude of signals in the alternative, respectively.
Here, signals mean nonzero elements in $\mu_{02}- \mu_{01} \in \bbR^p$.
In our simulation study, the following scenarios for alternatives are considered:
\begin{enumerate}
	\item ($H_{1R}$: Rare signals) 
	To demonstrate a situation where only a few signals exist, we set $n_0 = 5$ and consider various magnitudes of signals \begin{equation*}
		\mu \in \{ 0.2, 0.25, 0.3, 0.35, 0.4, 0.5, 0.6, 0.8, 1.0, 1.5 \}.	
	\end{equation*}
	
	\item ($H_{1M}$: Many signals) 
	To demonstrate a situation where a lot of signals exist, 
	we set $n_0 = p/2$ and consider various magnitudes of signals \begin{equation*}
		\mu \in \{ 0.025, 0.05, 0.1, 0.15, 0.2, 0.25, 0.3, 0.4, 0.5, 0.6 \}.
	\end{equation*} 
	Note that relatively smaller signals are used compared to ``rare signals'' setting, due to the larger number of signals.
\end{enumerate}
Furthermore, we consider the following two settings for the true covariance matrix $\sg_0$: 
\begin{enumerate}
	\item (Sparse $\Omega_0 = \sg_0^{-1}$) 
	To demonstrate a situation where the true precision matrix is sparse, we randomly choose $1\%$ of entries in $\Omega_0 = (\omega_{0, ij})$ and set their value to $\omega_{0, ij} = 0.3$.
	The rest of entries in $\Omega_0$ are set to $0$. 
	When the resulting $\Omega_0$ is not positive definite, we make it positive definite by adding $\{-\lambda_{\min}(\Omega_0) + 0.1^3\} I_p$ to $\Omega_0$.
	Finally, we set $\sg_0 = \Omega_0^{-1}$.
	
	\item (Dense $\Omega_0 = \sg_0^{-1}$) 
	To demonstrate a situation where the true precision matrix is dense, we randomly choose $40\%$ of entries in $\Omega_0$ and set their value to $\omega_{0, ij} = 0.3$.
	The rest of the steps for constructing $\sg_0$ is the same as above.
\end{enumerate}

In each setting and hypothesis, 50 simulated data sets are generated.
For the proposed mxPBF-based two-sample mean test, the hyperparameter $\alpha$ is set to $\alpha = 2.01$ to satisfy condition \eqref{alpha_mean}.
As contenders, we consider the tests proposed by \cite{bai1996effect}, \cite{srivastava2008test} and \cite{cai2014two}, which will be simply denoted as BS, SD and CLX, respectively.
Here, CLX means the two-sample mean test based on the CLIME, while CLX.AT refers to the two-sample mean test based on the inverse of the adaptive thresholding estimator with the tuning parameter $\delta=2$ as a default choice.
Note that BS and SD are $L_2$-type tests, while mxPBF, CLX and CLX.AT are maximum-type tests.
It is expected that $L_2$-type tests perform better (worse) than maximum-type tests in ``many signals'' (``rare signals'') setting.
To illustrate performance of each test, receiver operating characteristic (ROC) curves are drawn.
Points of the curves are obtained by adjusting thresholds and significance levels for the mxPBF and frequentist tests, respectively.

Furthermore, we compare the performance of each test at a fixed threshold or significance level.
Note that we need to fix threshold and significance level in practice.
As default choices, threshold $C_{\rm th} =10$ and significance level $0.05$ are used.
Note that $C_{\rm th} =10$ corresponds to ``strong evidence'' for the alternative hypothesis based on the criteria suggested by \cite{jeffreys1998theory} and \cite{kass1995bayes}.

\begin{figure*}
	\centering
	\includegraphics[width=16.4cm, height=9.5cm]{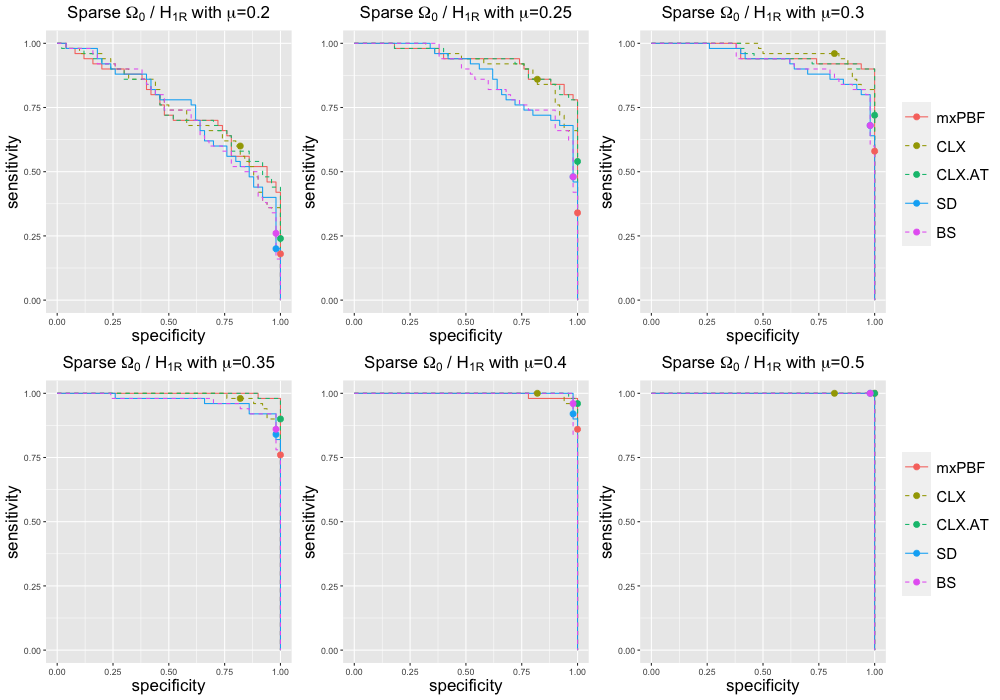}
	\includegraphics[width=16.4cm, height=9.5cm]{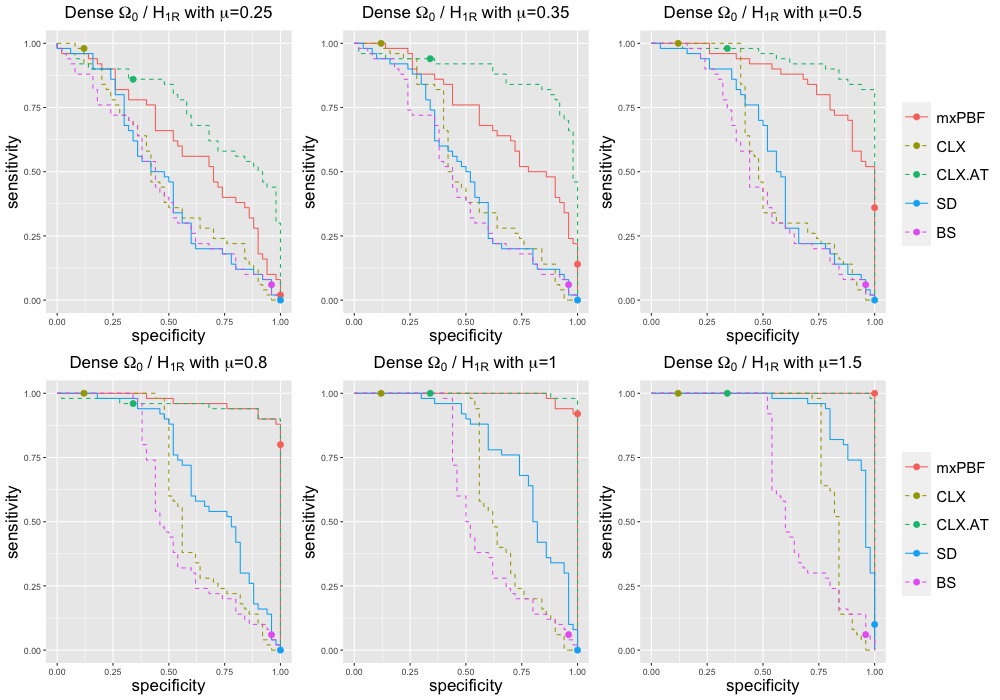}
	\vspace{-.7cm}
	\caption{
		ROC curves for the two-sample mean tests based on 50 simulated data sets for each hypothesis, $H_0$ and $H_{1R}$, with $p=100$.
		The mxPBF, SD and BS represent the test proposed in this paper, \cite{srivastava2008test} and \cite{bai1996effect}, respectively.
		The CLX and CLX.AT mean the tests proposed by \cite{cai2014two} based on the CLIME and the adaptive thresholding estimator, respectively.
		The dots show the results with $C_{\rm th}=1$ for the mxPBF and significance level $0.05$ for frequentist tests.
		%		``Sparse'' and ``Dense'' represent the settings where the true precision matrix $\Omega_0$ are sparse and dense, respectively. 
	}
	\label{fig:roc1}
\end{figure*}

\begin{figure*}
	\centering
	\includegraphics[width=16.4cm, height=9.5cm]{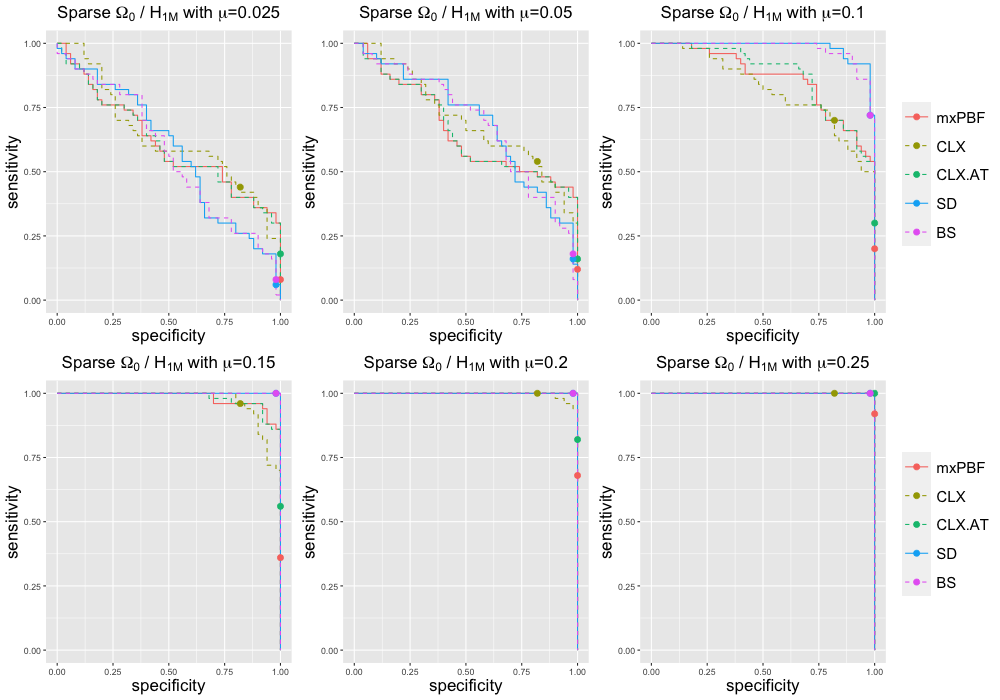}
	\includegraphics[width=16.4cm, height=9.5cm]{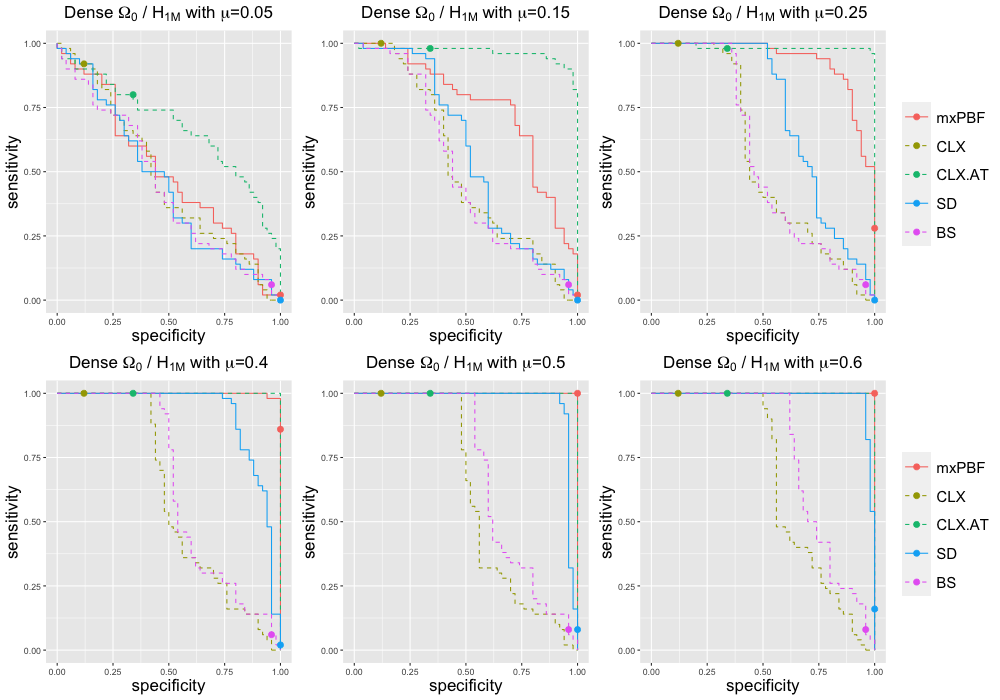}
	\vspace{-.6cm}
	\caption{
		ROC curves for the two-sample mean tests based on 50 simulated data sets for each hypothesis, $H_0$ and $H_{1M}$, with $p=100$.
	}
	\label{fig:roc2}
\end{figure*}

Figure \ref{fig:roc1} shows ROC curves based on 50 simulated data sets for each hypothesis, $H_0: \mu_{01}=\mu_{02}$ and $H_{1R}: \mu_{01}\neq \mu_{02}$, with $p=100$.
Here, $H_{1R}$ represents the ``rare signals'' scenario where $\mu_{02}  -\mu_{01} \in \bbR^p$ has only five nonzero elements with size $\mu$.
The dots in Figure \ref{fig:roc1} show the results with $C_{\rm th}=1$ for the mxPBF and significance level $0.05$ for frequentist tests.
When the true precision matrix $\Omega_0$ is sparse, the maximum-type tests overall slightly work better than the $L_2$-type tests as expected.
However, when the true precision matrix $\Omega_0$ is dense, we find that performance of CLX is not satisfactory.
We suspect this is because, as mentioned earlier, CLX relies on an estimated precision matrix by the CLIME.
In fact, we confirmed that the performance of the CLIME is worse in the dense $\Omega_0$ setting than in the sparse $\Omega_0$ setting, which supports our conjecture.
On the other hand, the mxPBF and CLX.AT outperform other tests in the  dense $\Omega_0$ setting.
When $\mu \le 0.5$, CLX.AT works better than the mxPBF in terms of the area under the curve (AUC), while when $\mu \ge 0.8$, the two tests produce quite similar ROC curves.
However, we find that CLX.AT with significance level $0.05$ tends to have low specificity in the dense $\Omega_0$ setting. 
On the other hand, the mxPBF-based Bayesian test with $C_{\rm th}=1$ performs reasonably well.
Especially in the dense $\Omega_0$ setting, when there are detectable signals $(\mu \ge 0.8)$, its specificity and sensitivity are close to $1$, while other tests suffer from low specificity or low sensitivity.
This clearly shows the relative advantage of the mxPBF-based two-sample mean test over the existing maximum-type tests, CLX and CLX.AT.

Figure \ref{fig:roc2} shows ROC curves based on 50 simulated data sets for each hypothesis, $H_0: \mu_{01}=\mu_{02}$ and $H_{1M}: \mu_{01}\neq \mu_{02}$, with $p=100$.
Here, $H_{1M}$ represents the ``many signals'' scenario where $\mu_{02}  -\mu_{01} \in \bbR^p$ has $p/2 = 50$ nonzero elements with size $\mu$.
When the true precision matrix $\Omega_0$ is sparse, overall, the $L_2$-type tests slightly work better than the maximum-type tests as expected.
However, when the true precision matrix $\Omega_0$ is dense, somewhat surprisingly, the mxPBF outperforms the $L_2$-type tests.
This observation can be partially explained by theoretical properties of the $L_2$-type tests: \cite{bai1996effect} and \cite{srivastava2008test} showed that powers of their tests decrease as the Frobenius norm of the true covariance and correlation matrices increase, respectively.
Indeed, in our simulations, we find that $\|\sg_0\|_F$ and $\|R_0\|_F$ are much larger in the dense $\Omega_0$ setting than in the sparse $\Omega_0$ setting.
We further confirmed that, when $H_{1M}$ is true, the $L_2$-type tests tend to fail to reject $H_0$ even when the size of signals is large.
Therefore, this observation suggests another advantage of the mxPBF that reasonable performance is maintained even when $\|\sg_0\|_F$ is large.

Again, the dots in Figure \ref{fig:roc2} show the results with $C_{\rm th}=1$ for the mxPBF and significance level $0.05$ for frequentist tests.
The performance of the mxPBF-based Bayesian test with $C_{\rm th}=1$ seems reasonable although it is a bit conservative in the sparse $\Omega_0$ setting.
In the dense $\Omega_0$ setting, however, the mxPBF clearly outperforms other tests when there are detectable signals $(\mu \ge 0.4)$.
Similar to $H_{1R}$ setting, the other tests suffer from low specificity or low sensitivity even when $\mu \ge 0.4$.

When $p=300$, similar phenomena are observed, thus we omit it here for reasons of space.
The results with $p=300$ including ROC curves and descriptions are deferred to the Supplementary material.

\subsection{Simulation study: two-sample covariance test}\label{subsec:sim_cov}

Now, we illustrate performance of the mxPBF for two-sample covariance test.
We generate the data as follows: $X_1,\ldots, X_{n} \overset{i.i.d.}{\sim} N_p( 0,  \sg_{01})$ and $Y_1,\ldots, Y_{n} \overset{i.i.d.}{\sim} N_p( 0 , \sg_{02})$ with $n = 100$ and $p \in \{100, 300\}$.
Under the null hypothesis, $H_0: \sg_{01} = \sg_{02}$, we set $\sg_{01} = \sg_{02} \equiv \sg_0\in \bbR^{p \times p}$.
Under the alternative hypothesis, $H_1: \sg_{01}\neq \sg_{02}$, we set $\sg_{01} \equiv \sg_0 $ and $\sg_{02} = \sg_{01} + U$ for some matrix $U \in \bbR^{p\times p}$ containing signals.
If $\sg_{01}$ or $\sg_{02}$ is not positive definite, we add a small diagonal matrix $\delta_1 I_p$ to them, where $\delta_1 = | \min  \{\lambda_{\min}(\sg_{01}), \lambda_{\min}(\sg_{02})\} | + 0.05$.
In our simulation study, the following two scenarios for alternatives are considered:
\begin{enumerate}
	\item ($H_{1R}$: Rare signals) 
	To demonstrate a situation where only few signals exist, 
	we randomly select five entries in the lower triangular part of $U$ and generate their values from ${\rm Unif}(0, \rho)$ with 
	\begin{equation*}
		\rho \in \{ 0.5, 0.8, 1.5, 3, 6, 15 \}.
	\end{equation*}
	
	\item ($H_{1M}$: Many signals) 
	To demonstrate a situation where a lot of signals exist, 
	we generate $u = (u_1,\ldots, u_p)^T$ from $u_j \overset{i.i.d.}{\sim} {\rm Unif}(0, \rho)$ for 
	\begin{equation*}
		\rho \in \{ 0.2, 0.3, 0.5, 0.7, 1, 1.5 \}.
	\end{equation*}
	Then, we set $U = u u^T$ that leads to $p(p+1)/2$ signals in $U$ (except upper triangular part).
	Note that relatively smaller signals are used compared to ``rare signals'' setting, due to the larger number of signals.
\end{enumerate}
Note that in the above, $\rho$ is the magnitude of signals.
Furthermore, we consider the following two settings for $\sg_{01}$: 
\begin{enumerate}
	\item (Sparse $\sg_{01}$) 
	To demonstrate a situation where $\sg_{01}$ is sparse, 
	we randomly choose $5\%$ of entries in $\Delta_1 = (\delta_{1,jk})$ and set their value to $\omega_{0, ij} = 0.5$.
	The rest of entries in $\Delta_1$ are set to $0$. 
	To make it positive definite, we set $\Delta = \Delta_1 + \delta I_p$, where $\delta = | \lambda_{\min}(\Delta_1)| + 0.05$.
	Finally, we set $\sg_{01} =  D^{1/2} \Delta D^{1/2}$, where $D= diag(d_j)$ and $d_j \overset{i.i.d.}{\sim} {\rm Unif}(0.5, 2.5)$.
	This setting corresponds to Model 3 in \cite{cai2013two}.
	
	\item (Dense $\sg_{01}$) 
	To demonstrate a situation where $\sg_{01}$ is dense, we set $\sg_{01}  = O \Delta O$, where $O = diag(\omega_j)$, $\omega_j \overset{i.i.d.}{\sim} {\rm Unif}(1,5)$, $\Delta = (\delta_{ij})$ and $\delta_{ij} = (-1)^{i+j} 0.4^{|i-j|^{1/10}}$.
	This setting corresponds to Model 4 in \cite{cai2013two}.
\end{enumerate}

In each setting and hypothesis, we generate 50 simulated data.
The mxPBF \eqref{B10_two_cov}  have hyperparameters $a_0, b_{0,ij} , b_{01,ij} , b_{02,ij}$ and $\alpha$.
We suggest using $a_0 =  b_{0,ij} = b_{01,ij} = b_{02,ij} = 0.01$ for all $1\le i \neq j \le p$.
Note that by the proof of Theorem \ref{thm:two_cov}, the leading terms affect the asymptotic behavior of the mxPBF \eqref{B10_two_cov} while including above hyperparameters are \eqref{leading1_hat} and \eqref{leading2_hat}.
Thus, it can be considered that the above choice leads to noninformative priors that have little effect on the mxPBF.
By Theorem \ref{thm:two_cov}, $\alpha > 6 C_2$ is required for consistency under the null.
This choice roughly means $\alpha$ is slightly larger than $12$, but we found this to overly conservative in practice.
Therefore, we suggest using $\alpha =2.01$ similar to the two-sample mean test.

For comparison, we consider the tests proposed by \cite{schott2007test}, \cite{li2012two} and \cite{cai2013two}, which will be denoted as Sch, LC and CLX, respectively.
Note that Sch and LC are $L_2$-type tests, while mxPBF and CLX are maximum-type tests.
Because the unbiased version of the test in \cite{li2012two} is computationally expensive, we use the biased version as suggested by \cite{li2012two}.
In our settings, we confirm that the biased version gives quite similar results to unbiased version.
Similar to the simulation study for two-sample mean test, ROC curves are drawn to demonstrate performance of each test.

\begin{figure*}
	\centering
	\includegraphics[width=16.4cm, height=9.5cm]{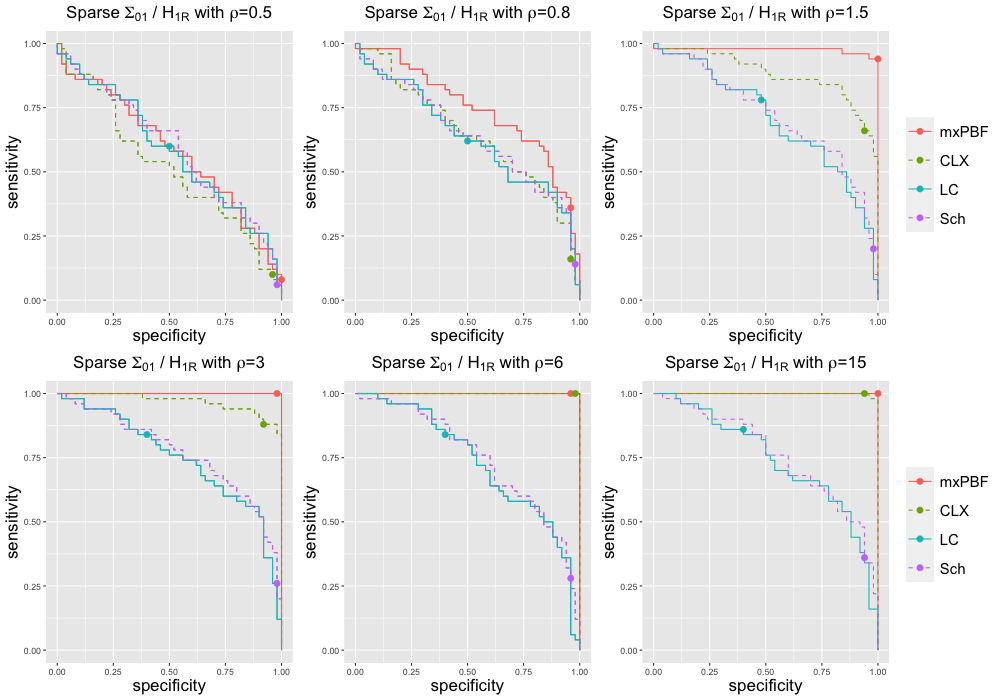}
	\includegraphics[width=16.4cm, height=9.5cm]{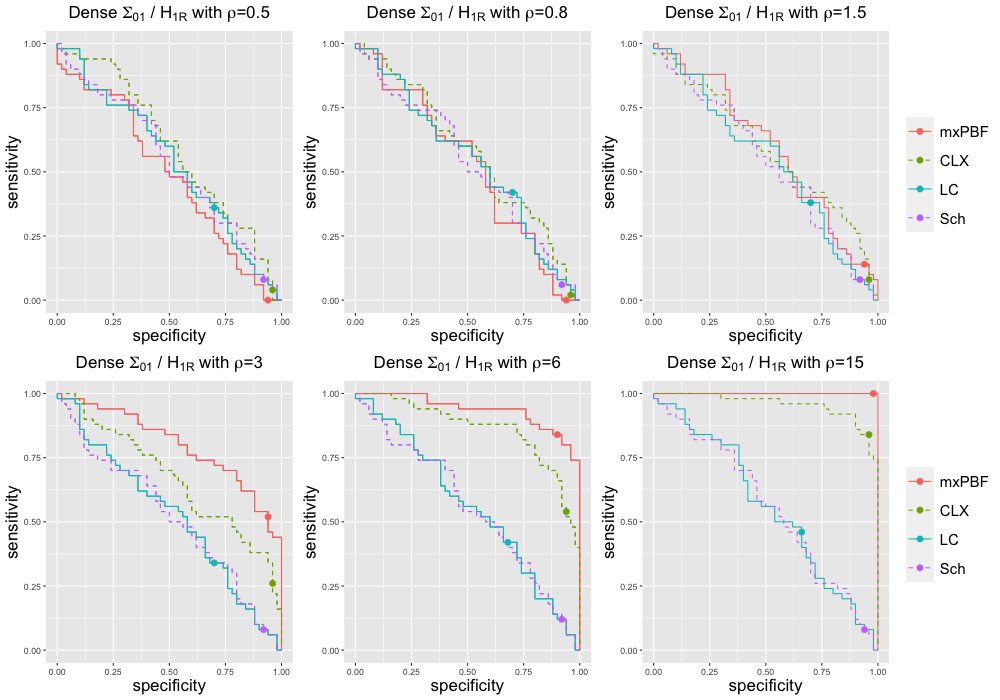}
	\vspace{-.7cm}
	\caption{
		ROC curves for the two-sample covariance tests based on 50 simulated data sets for each hypothesis, $H_0$ and $H_{1R}$, with $p=100$.
		The mxPBF, CLX, LC and Sch represent the test proposed in this paper, \cite{cai2013two}, \cite{li2012two} and \cite{schott2007test}, respectively.
		The dots show the results with $C_{\rm th}=1$ for the mxPBF and significance level $0.05$ for frequentist tests.
		%		``Sparse'' and ``Dense'' represent the settings where the true precision matrix $\Omega_0$ are sparse and dense, respectively. 
	}
	\label{fig:roc3}
\end{figure*}

\begin{figure*}
	\centering
	\includegraphics[width=16.4cm, height=9.5cm]{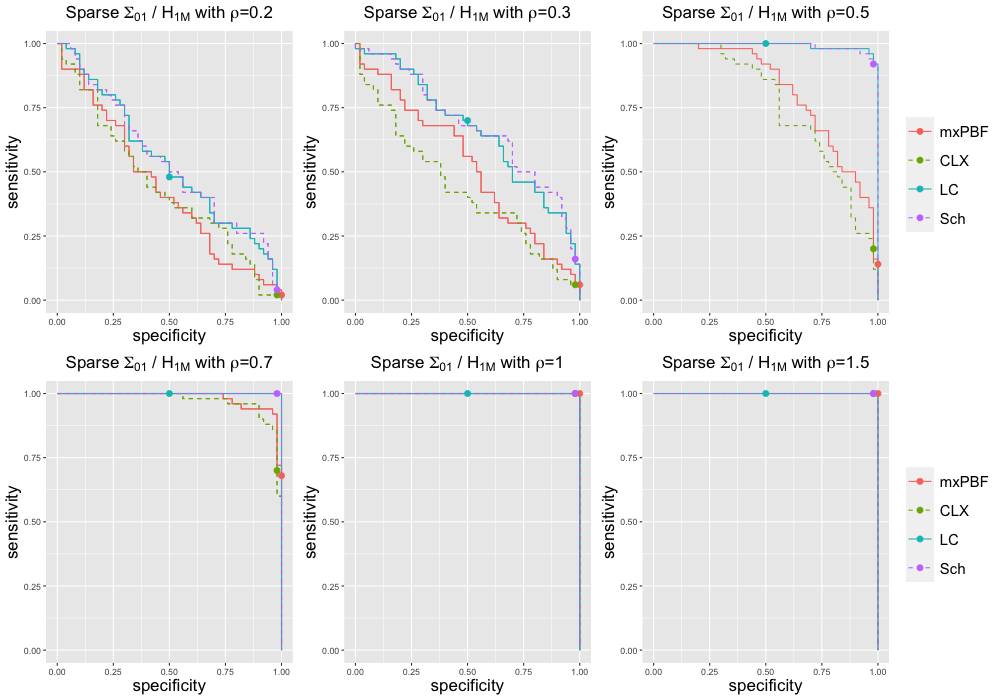}
	\includegraphics[width=16.4cm, height=9.5cm]{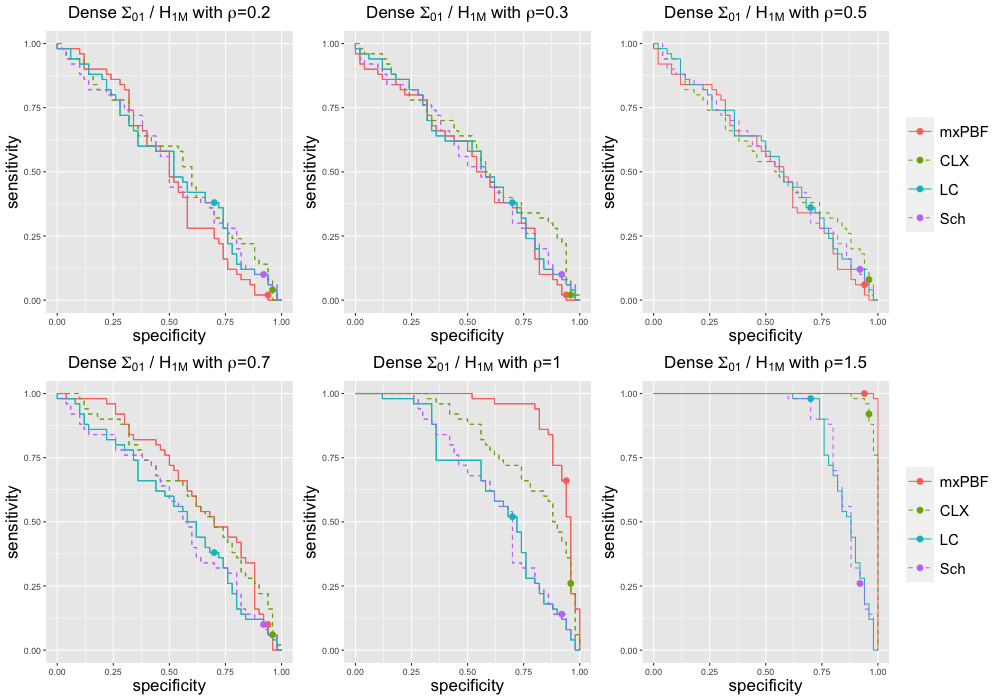}
	\vspace{-.6cm}
	\caption{
		ROC curves for the two-sample covariance tests based on 50 simulated data sets for each hypothesis, $H_0$ and $H_{1M}$, with $p=100$.
	}
	\label{fig:roc4}
\end{figure*}

Figure \ref{fig:roc3} shows ROC curves based on 50 simulated data sets for each hypothesis, $H_0: \sg_{01} = \sg_{02}$ and $H_{1R} : \sg_{01} \neq \sg_{02}$, with $p=100$.
When $\sg_{01}$ is sparse and signals are moderate $(\rho \ge 1.5)$, the maximum-type tests work better than the $L_2$-type tests as expected.
The performance of the $L_2$-type tests are slowly improved as $\rho$ gets larger.
Similar phenomena are observed in the dense $\sg_{01}$ setting, but in this case, the $L_2$-type tests do not work well even when there are large signals $(\rho =15)$.
Overall, we find that the mxPBF shows better performance than CLX.

Figure \ref{fig:roc4} shows ROC curves based on 50 simulated data sets for each hypothesis, $H_0: \sg_{01} = \sg_{02}$ and $H_{1M} : \sg_{01} \neq \sg_{02}$, with $p=100$.
As expected, the $L_2$-type tests slightly work better than the maximum-type tests when $\sg_{01}$ is sparse.
Note that the performance of the maximum-type tests are also rapidly improved as the signal $\rho$ gets larger.
Somewhat surprisingly, when $\sg_{01}$ is dense and signals are moderate $(\rho \ge 1)$, the maximum-type tests outperform the $L_2$-type tests.
We suspect that it is likely that the conditions for deriving the null distribution of Sch and LC are violated in the dense $\sg_{01}$ setting.
\cite{schott2007test} and \cite{li2012two} assumed that $\lim_{p\to\infty} {\rm tr}( \sg_0^i )/p =\gamma_i \in (0,\infty)$ for $i=1,\ldots, 8$ and ${\rm tr}(\sg_0^4) =  o \big\{  {\rm tr}(\sg_0^2)^2 \big\}$, respectively, to derive the null distribution.
In our settings, we find that ${\rm tr}( \sg_0^i )/p$ and ${\rm tr}(\sg_0^4)  / {\rm tr}(\sg_0^2)^2$ are much larger in the dense $\sg_{01}$ setting than in the sparse $\sg_{01}$ setting.
This partially supports our conjecture, although more rigorous investigation might be needed to determine the exact cause.

The dots in Figures \ref{fig:roc3} and \ref{fig:roc4} show the results with $C_{\rm th} =1$ for the mxPBF or significance level $0.05$ for frequentist tests.
The mxPBF and CLX with these default choices seem to work well if there is a reasonable amount of signals. 
On the other hand, the overall performances of LC and Sch with significance level $0.05$ are not satisfactory, especially in the sparse $\sg_{01}$ setting.

Lastly, we note that the experiment for $p=300$ showed similar phenomena whose results including ROC curves and descriptions are deferred to the Supplementary material due to lack of space.

\subsection{Real data analysis}\label{subsec:real}
In this section, we apply the proposed two-sample mean and covariance tests to two real datasets, small round blue cell tumors (SRBCT) dataset and prostate cancer dataset, respectively.
For both datasets, the sample sizes are quite small compared to the number of variables.
Thus, based on this numerical study, we would like to illustrate the practical performance of mxPBF-based tests in ``small $n$ large $p$'' situations.

We first apply two-sample mean tests to the SRBCT dataset.
The SRBCT dataset is available in the \verb|R| package \verb|plsgenomics|.
This is a gene expression data having 83 samples with 2308 genes $(p=2308)$ from the microarray experiments in \citep{khan2001classification}.
Among 83 samples, we focus on 11 cases of Burkitt lymphoma (BL) $(n_1=11)$ and 18 cases of neuroblastoma (NB) $(n_2=18)$.
Our main interest is to test equality of mean vectors of the gene expressions between BL and NB tumors.
We apply the mxPBF, CLX.AT, SD and BS to test equality of mean vectors.
Note that CLX.AT is used because the lack of prior information about the sparsity of the covariance matrix.
For this dataset, the value of the mxPBF is greater than $10^{8}$, and $p$-values of CLX.AT, SD and BS are less than $10^{-15}$.
Therefore, all the tests reject the null hypothesis, $H_0: \mu_1 =\mu_2$, if we use the default choices, threshold $C_{\rm th} =10$ an significance level $0.05$.

% SRBCT dataset:
% mxPBF: 10^{8.332577}
% CLX.AT: < 2.2e-16
% SD: 3.575e-16
% BS: < 2.2e-16

The prostate cancer dataset is available in the \verb|R| package \verb|SIS|.
This dataset contains $12600$ gene expressions from $52$ patients with prostate tumors $(n_1=52)$ and $50$ patients with normal prostate $(n_2=50)$.
As suggested by \cite{cai2013two}, $5000$ genes $(p=5000)$ with the largest absolute values of the $t$ statistics are selected.
Data were centered prior to analysis.
In this dataset, we would like to test equality of covariance matrices of the gene expressions between tumor and normal samples.
We apply the mxPBF, CLX, LC and Sch to test equality of covariance matrices.
For this dataset, the value of the mxPBF is greater than $10^{32}$, and $p$-values of CLX, LC and Sch are less than $0.0058$, $10^{-15}$ and $10^{-15}$, respectively.
Therefore, all the tests reject the null hypothesis, $H_0: \sg_1 =\sg_2$, if we use the default choices, threshold $C_{\rm th} =10$ an significance level $0.05$.

%prostate cancer dataset: 
% mxPBF: 10^{32.11797}
% CLX: 0.005769
% LC: < 2.2e-16
% Sch: < 2.2e-16

%%%%%%%%%%%%%%%%%%%%%%%%%%%%%%%%%%%%%%%%%%%%%
\section{Discussion}\label{sec:disc}

In this paper, we propose a Bayesian two-sample mean test and a Bayesian two-sample covariance test in high-dimensional settings based on the idea of the maximum pairwise Bayes factor \citep{lee2018maximum}.
These tests are not only computationally scalable but also enjoy Bayes factor consistency under relatively weak or similar conditions compared to existing tests.
The proposed methods can be applied to change point detection for mean vectors or covariance matrices, which is indeed one of our ongoing works.
Note that from the first data point, using only a subset of data within a certain window, a two-sample test can be sequentially conducted to detect change points.
Due to consistency of the proposed mxPBF-based two-sample tests, it is expected that the resulting change point detection procedures can consistently detect and estimate change points.

\newpage

\appendix

\section*{Supplementary Material}

\section{Additional numerical results when $p=300$}\label{sec:additional_sim}

\subsection{Two-sample mean test}

In this section, we provide additional numerical results for the two-sample mean tests when $p=300$.
Figures \ref{fig:add_roc1} and \ref{fig:add_roc2} show ROC curves based on 50 simulated data sets for each hypothesis where the alternative is $H_{1R}: \mu_{01} \neq \mu_{02}$ and $H_{1M}: \mu_{01} \neq \mu_{02}$, respectively.

Overall, in the rare signals setting $H_{1R}$, the mxPBF outperforms the other tests.
Note that CLX.AT produces similar ROC curves when $\Omega_0$ is sparse, but it does not seem to work well when $\Omega_0$ is dense.
When there are many signals $(H_{1M})$ and $\Omega_0$ is sparse, the $L_2$-type tests slightly work better than the maximum-type tests.
However, when $\Omega_0$ is dense, the mxPBF and CLX.AT are better than the $L_2$-type tests in terms of ROC curves.
This is consistent with what we observed when $p=100$.

\begin{figure*}
	\centering
	\includegraphics[width=16.cm,height=9.cm]{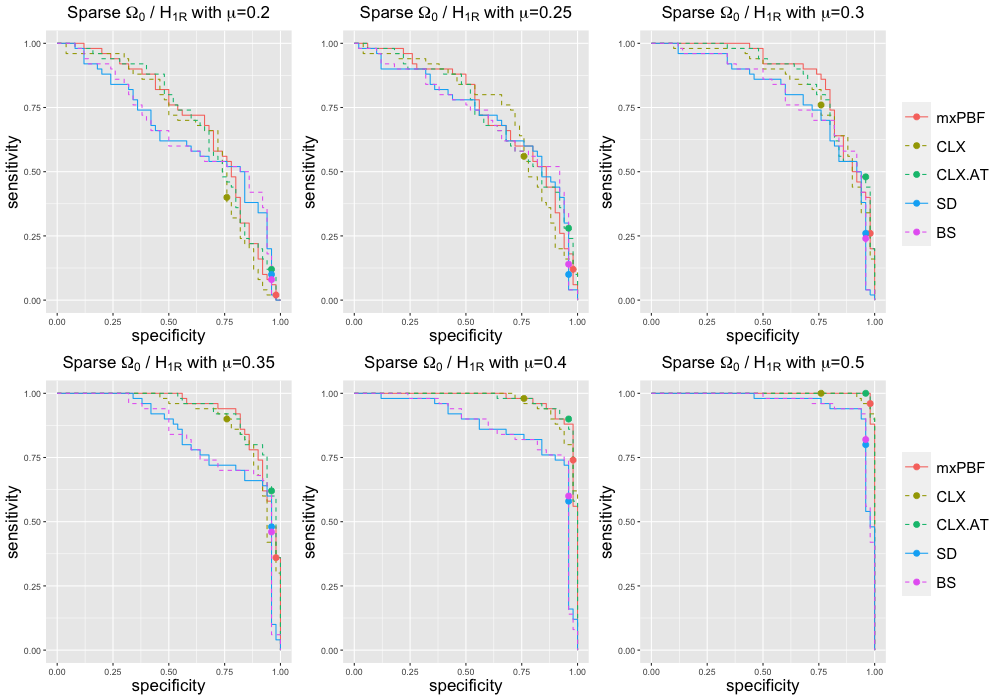}
	\includegraphics[width=16.cm,height=9.cm]{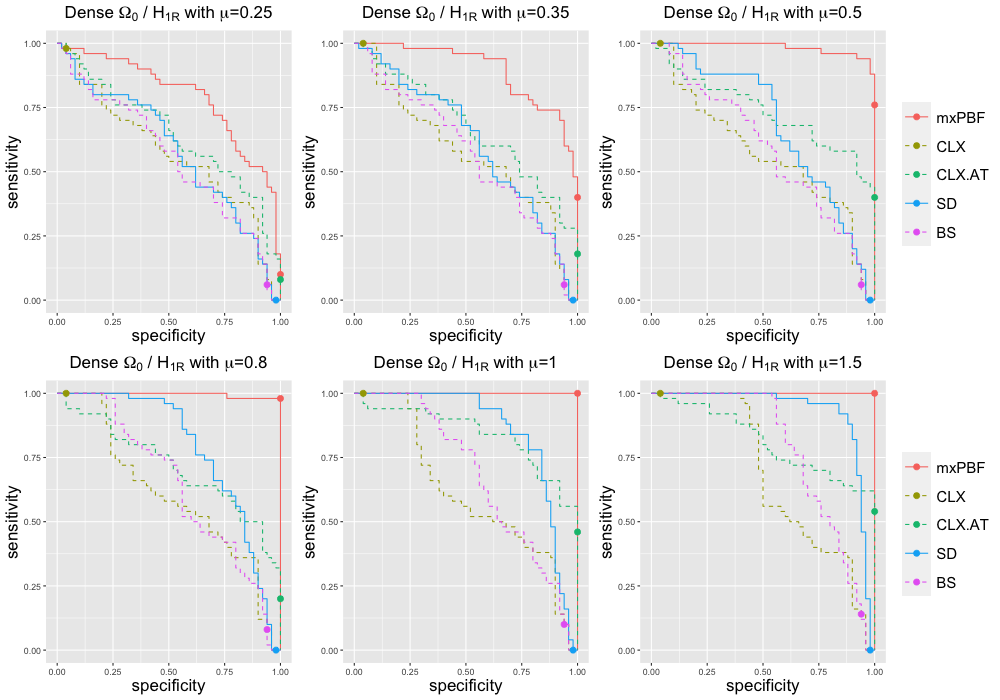}
%	\vspace{-.3cm}
	\caption{
		ROC curves for the two-sample mean tests based on 50 simulated data sets for each hypothesis, $H_0$ and $H_{1R}$, with $p=300$.
	}
	\label{fig:add_roc1}
\end{figure*}
\begin{figure*}
	\centering
	\includegraphics[width=16.cm,height=9.cm]{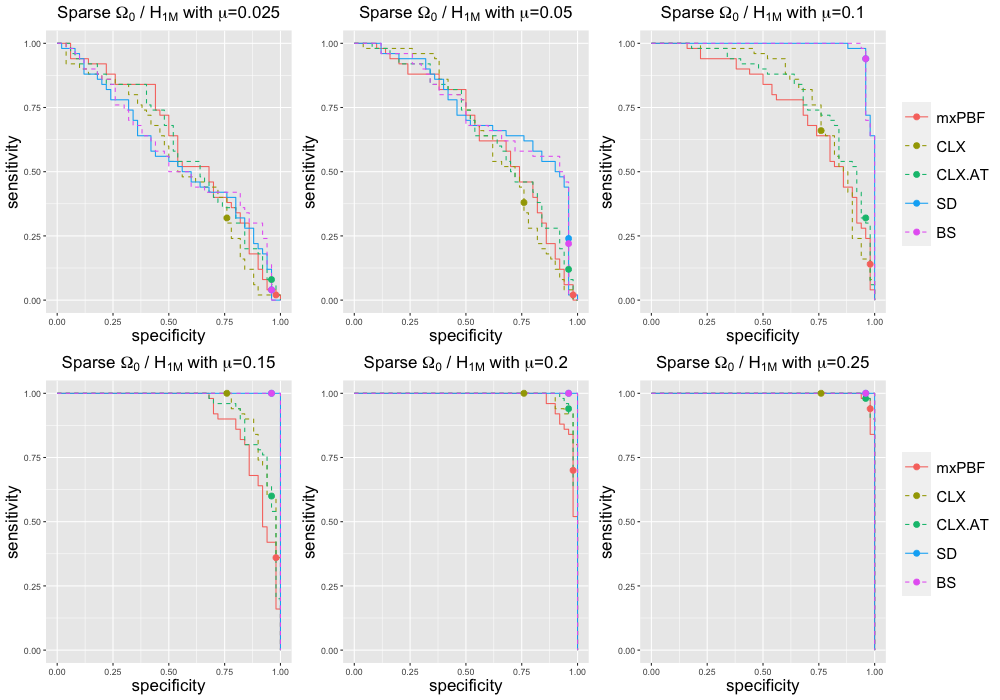}
	\includegraphics[width=16.cm,height=9.cm]{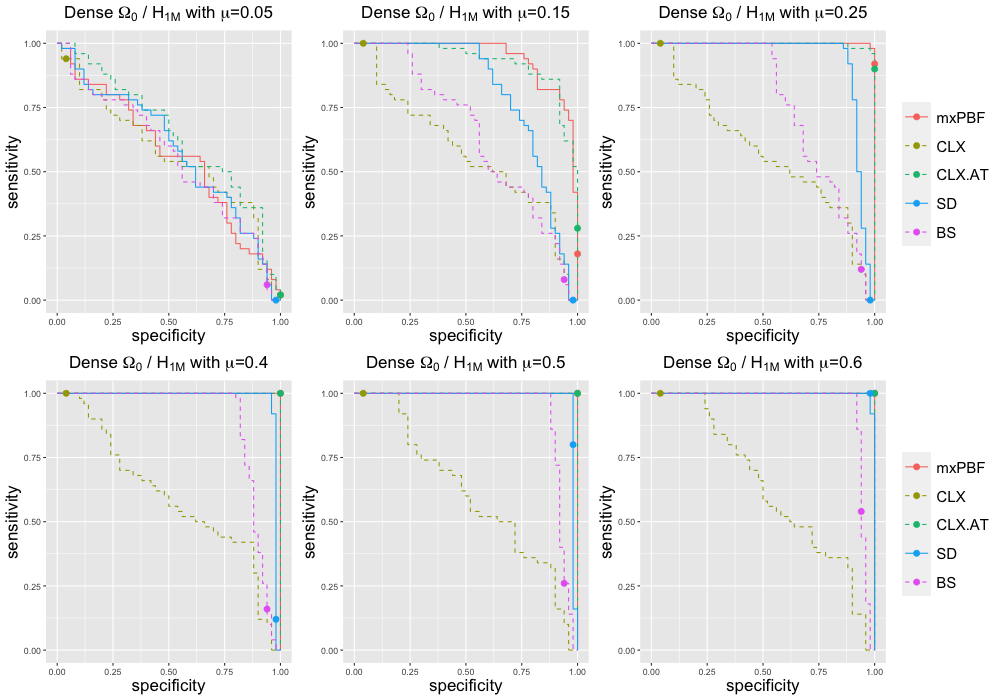}
%	\vspace{-.3cm}
	\caption{
		ROC curves for the two-sample mean tests based on 50 simulated data sets for each hypothesis, $H_0$ and $H_{1M}$, with $p=300$.
	}
	\label{fig:add_roc2}
\end{figure*}

\subsection{Two-sample covariance test}

In this section, we provide additional numerical results for the two-sample covariance tests when $p=300$.
Figures \ref{fig:add_roc3} and \ref{fig:add_roc4} show ROC curves based on 50 simulated data sets for each hypothesis where the alternative is $H_{1R}: \sg_{01} \neq \sg_{02}$ and $H_{1M}: \sg_{01} \neq \sg_{02}$, respectively.

In the rare signals setting $H_{1R}$, the maximum-type tests outperform the $L_2$-type tests as expected.
Especially, the mxPBF seems to work better than CLX in terms of ROC curves.
In the many signals setting $H_{1M}$, the $L_2$-type tests outperform the maximum-type tests.
However, when $\sg_{01}$ is dense and signals are moderate $(\rho \ge 0.7)$, the mxPBF and CLX outperforms the other tests.
This is consistent with our observation when $p=100$.

\begin{figure*}
	\centering
	\includegraphics[width=16.5cm,height=9.5cm]{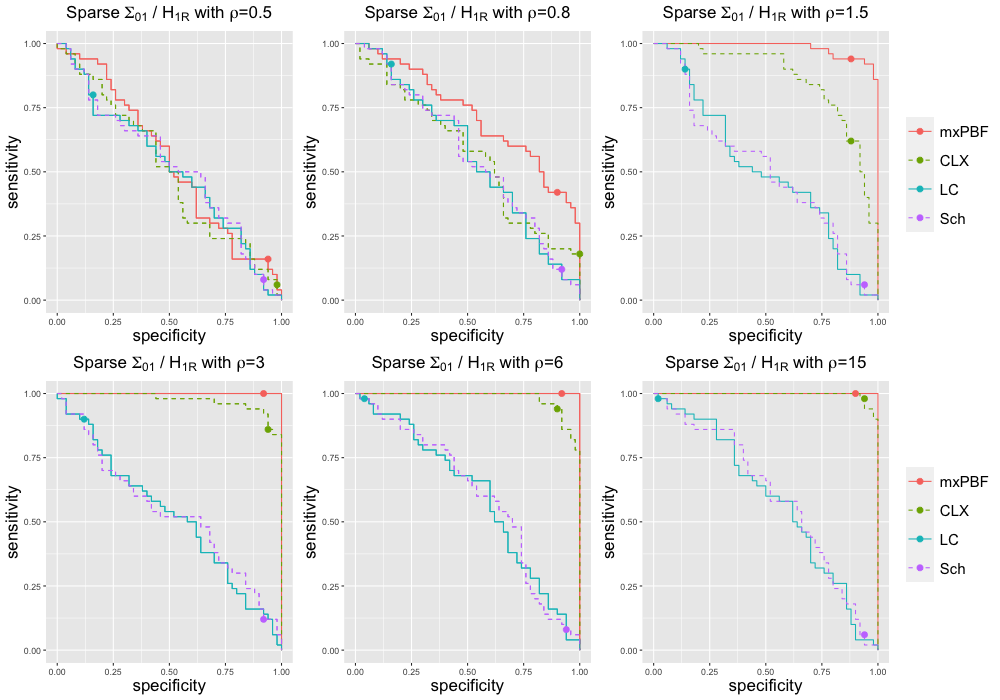}
	\includegraphics[width=16.5cm,height=9.5cm]{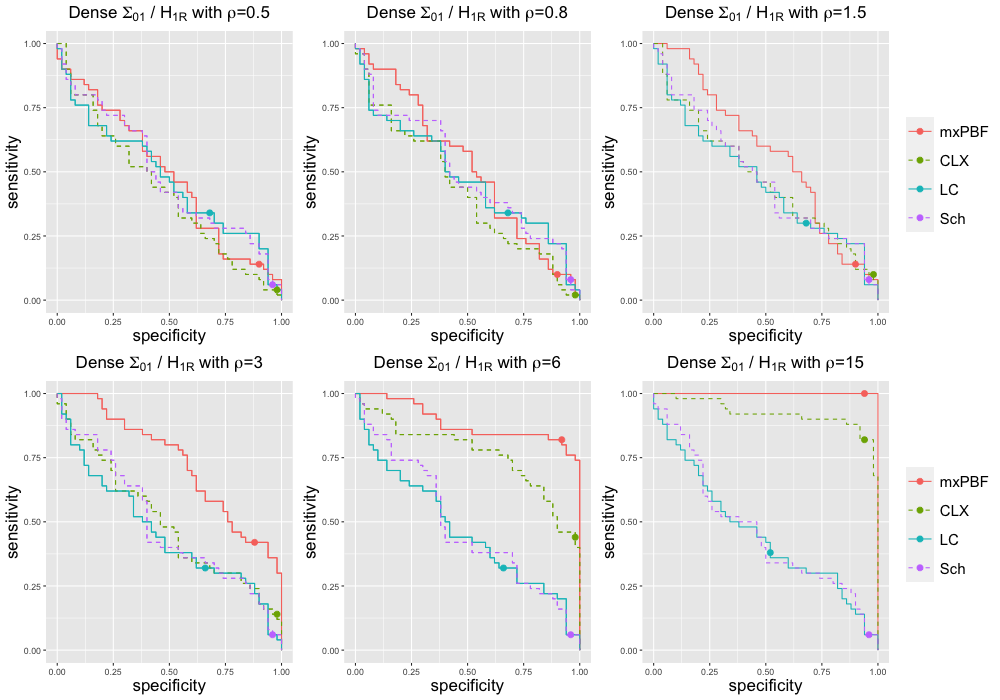}
%	\vspace{-.3cm}
	\caption{
		ROC curves for the two-sample covariance tests based on 50 simulated data sets for each hypothesis, $H_0$ and $H_{1R}$, with $p=300$.
	}
	\label{fig:add_roc3}
\end{figure*}
\begin{figure*}
	\centering
	\includegraphics[width=16.5cm,height=9.5cm]{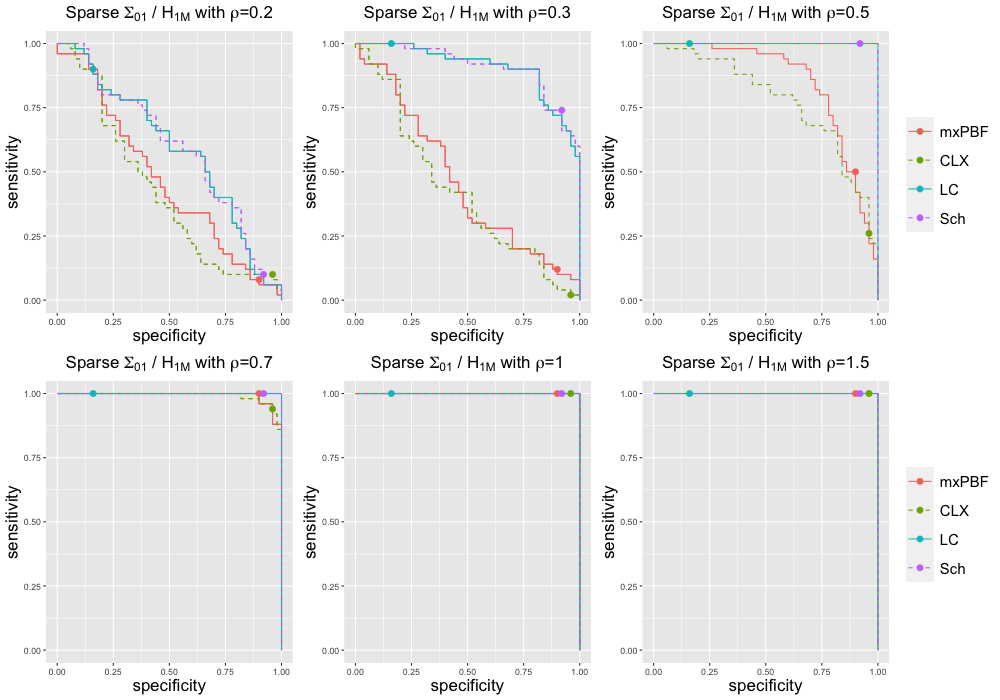}
	\includegraphics[width=16.5cm,height=9.5cm]{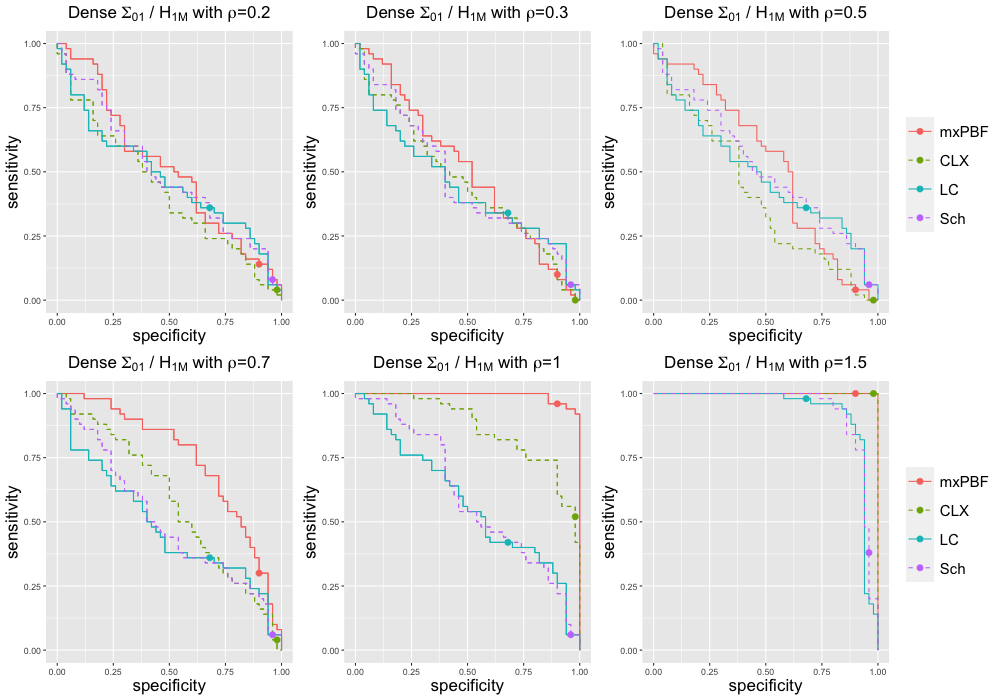}
%	\vspace{-.3cm}
	\caption{
		ROC curves for the two-sample covariance tests based on 50 simulated data sets for each hypothesis, $H_0$ and $H_{1M}$, with $p=300$.
	}
	\label{fig:add_roc4}
\end{figure*}

\section{Proofs}\label{sec:proof}
Throughout Section \ref{sec:proof}, we denote $\bfZ_n = (\bfX_{n_1}, \bfY_{n_2})$ for simplicity.
%For notational convenience, we denote $\what{\tau}_{1,ij} = \what{\tau}_{1,ij,0}$ and $\what{\tau}_{2,ij} = \what{\tau}_{2,ij,0}$.

\subsection{Proof of Theorem \ref{thm:two_mean}}

\begin{lemma}[Lemmas 6 and 8 in \cite{kolar2012marginal}]\label{lem:chi_bounds}
	Let  $\chi_k^2$ be a chi-squared random variable with degree of freedom $k$, and $\chi_k^2(\lambda)$ be a non-central chi-squared random variable with degree of freedom $k$ and non-centrality $\lambda$.
	Then, for any $x>0$,
	\bea
	\bbP \Big( \chi_k^2 \le k+ 2\sqrt{k x} + 2x \Big) &\ge& 1 -\exp(-x), \\
	\bbP \Big( \chi_k^2  \ge k - 2\sqrt{k x} \Big) &\ge& 1 -\exp(-x), \\
	\bbP \Big( \chi_k^2(\lambda) \le k + \lambda + 2\sqrt{(k+2\lambda) x} + 2x \Big) &\ge& 1 -\exp(-x),  \\
	\bbP \Big( \chi_k^2(\lambda)  \ge k + \lambda- 2\sqrt{(k+2\lambda) x} \Big) &\ge& 1 -\exp(-x) .
	\eea
\end{lemma}

\begin{proof}[Proof of Theorem \ref{thm:two_mean}]
	
	We first assume the true null $H_{0}: \mu_{1}=\mu_{2}=\mu_{0}$. 
	For a given $1\le j\le p$, the PBF is
	\bean
	\log B_{10} (\tilde{X}_j, \tilde{Y}_j)
	&=& \frac{1}{2}\log \Big(\frac{\gamma}{1+\gamma} \Big) \nonumber\\
	&+& \frac{n}{2}  \log \Bigg( \frac{ n\what{\sigma}^2_{Z_j} }{ n_1 \what{\sigma}^2_{X_j} + n_2 \what{\sigma}^2_{Y_j} }  \Bigg). \label{tmeanBF_main}
	\eean
	Since $\log(\gamma/(1+\gamma)) \le -\alpha \log (n\vee p)$, we focus on an upper bound for \eqref{tmeanBF_main}.
	We have
	\bea
	\log \Bigg( \frac{ n\what{\sigma}^2_{Z_j} }{ n_1 \what{\sigma}^2_{X_j} + n_2 \what{\sigma}^2_{Y_j} }  \Bigg)
	&\le& \frac{ n\what{\sigma}^2_{Z_j} - n_1 \what{\sigma}^2_{X_j} - n_2 \what{\sigma}^2_{Y_j} }{  n_1 \what{\sigma}^2_{X_j} + n_2 \what{\sigma}^2_{Y_j}  }
	\eea
	because $\log(1+x) \le x$ for any $x\in \bbR$.
	It is easy to show that
%	\bea
%	&& n(\what{\sigma}^2_{Z_j,\gamma} - \what{\sigma}^2_{Z_j}) - n_1 (\what{\sigma}^2_{X_j,\gamma}-\what{\sigma}^2_{X_j}) - n_2 (\what{\sigma}^2_{Y_j,\gamma}-\what{\sigma}^2_{Y_j}) \\
%	&\le& \frac{\gamma}{1+\gamma} \frac{1}{n}\Big( \sum_{i=1}^n Z_{ij} - n\mu_{0j}  \Big)^2  \,\,\overset{d}{=}\,\, \frac{\gamma}{1+\gamma} \sigma_{0,jj} \cdot \chi_1^2 
%	\eea
%	and
	\bea
	n_1 \what{\sigma}^2_{X_j} + n_2 \what{\sigma}^2_{Y_j}
	&\overset{d}{=}& \sigma_{0,jj}\cdot \chi_{n-2}^2,
	\eea
	under $H_{0j}: \mu_{1j}=\mu_{2j}=\mu_{0j}$.
	Furthermore,
	\bea
	&& n\what{\sigma}^2_{Z_j} - n_1 \what{\sigma}^2_{X_j} - n_2 \what{\sigma}^2_{Y_j} \\
	&=& \frac{1}{n} \Bigg( \sqrt{\frac{n_2}{n_1}}\sum_{i=1}^{n_1}X_{ij} - \sqrt{\frac{n_1}{n_2}}\sum_{i=1}^{n_2}Y_{ij} \Bigg)^2 \\
	&=& \frac{1}{n} \Bigg( \sqrt{\frac{n_2}{n_1}}\Big(\sum_{i=1}^{n_1}X_{ij}-n_1\mu_{0j} \Big) - \sqrt{\frac{n_1}{n_2}}\Big(\sum_{i=1}^{n_2}Y_{ij}-n_2\mu_{0j} \Big) \Bigg)^2 \,\,\overset{d}{=}\,\, \sigma_{0,jj} \cdot \chi_1^2 .
	\eea
	Define the events
	\bea
	A_{1j} &=& \Big\{\bfZ_n: \sigma_{0,jj}^{-1} \big( n_1 \what{\sigma}^2_{X_j} + n_2 \what{\sigma}^2_{Y_j} \big) \ge n-2 - 2 \sqrt{(n-2)C\log (n\vee p)}  \Big\},\\
	A_{2j} &=& \Big\{\bfZ_n: \sigma_{0,jj}^{-1} \big(n\what{\sigma}^2_{Z_j} - n_1 \what{\sigma}^2_{X_j} - n_2 \what{\sigma}^2_{Y_j} \big) \le 1 + 2\sqrt{C\log(n\vee p)} + 2C\log(n\vee p)   \Big\}
	\eea
	for any constant $1<C <C_1$.
	By Lemma \ref{lem:chi_bounds}, $\sum_{j=1}^p\sum_{l=1}^2\bbP_0(A_{lj}^c) \le 2 (n\vee p)^{-C+1} = o(1)$.
	On the event $A_{1j}\cap A_{2j}\cap A_{3j}$, \eqref{tmeanBF_main} is bounded above by
	\bea
	&& \frac{n}{2} \frac{1 + 2\sqrt{C\log(n\vee p)} + 2C\log(n\vee p)}{n-2 - 2 \sqrt{(n-2)C\log(n\vee p)}}  \\
	&=& \frac{1}{2}  \frac{ \{C\log(n\vee p)\}^{-1} +  2(\sqrt{C\log(n\vee p)})^{-1} + 2  }{1 - 2 n^{-1} - 2 \sqrt{(1-2 n^{-1})C n^{-1}\log(n\vee p)}} \cdot C\log(n\vee p) \\
	&=& \frac{ \{2 C\log(n\vee p)\}^{-1} +  (\sqrt{C\log(n\vee p)})^{-1} + 1  }{1 - 2 n^{-1} - 2 \sqrt{(1-2 n^{-1})C n^{-1}\log(n\vee p)}} \cdot C\log(n\vee p) \\
	&\le& \frac{1 + \epsilon_0}{1 - 3 \sqrt{C_1 \epsilon_0}} \cdot C\log(n\vee p) \\
	&=:& C_0 C \log(n\vee p) 
	\eea
	for all large $n_1 \wedge n_2$.
	Since $\alpha > 2C_0$ by condition \eqref{alpha_mean}, by taking $C$ arbitrarily close to 1, we have $C_0 C - \alpha/2 <0$.
	Thus,
	\bea
	&& \bbP_0 \Big( \log B^{\mu}_{\max,10}(\bfX_{n_1}, \bfY_{n_2}) \le (C_0 C - \frac{\alpha}{2} )\log (n\vee p) \Big) \\
	&=& 1 - \bbP_0 \Big( \max_{1\le j \le p}\log B_{10}(\tilde{X}_j, \tilde{Y}_j) > (C_0 C - \frac{\alpha}{2} )\log (n\vee p) \Big)  \\
	&\ge& 1 - \sum_{1\le j \le p} \bbP_0 \Big( \log B_{10}(\tilde{X}_j, \tilde{Y}_j) > (C_0 C - \frac{\alpha}{2} )\log (n\vee p) \Big)  \\
	&\ge& 1 - 2 (n\vee p)^{-C+1}
	\eea
	for some constant $1<C < C_1$. 
	It completes the proof under $H_0$.

	Now, for a given $1\le j\le p$, we assume $H_{1j}$ is true, i.e., $\mu_{0,1j} \neq \mu_{0,2j}$, and condition \eqref{betamin_mean} holds.
	Note that
	\bea
	\log B_{10} (\tilde{X}_j, \tilde{Y}_j)
	&\ge& \frac{1}{2}\log \Big(\frac{\gamma}{1+\gamma} \Big) +  \frac{n}{2} \,\frac{ n\what{\sigma}^2_{Z_j} - n_1 \what{\sigma}^2_{X_j} - n_2 \what{\sigma}^2_{Y_j} }{  n \what{\sigma}^2_{Z_j}  }  
	\eea
	because $-\log(1-x) \le x/(1-x)$ for any $x<1$, and
	\bea
	n\what{\sigma}^2_{Z_j}
	&=& n_1 \what{\sigma}^2_{X_j} + n_2 \what{\sigma}^2_{Y_j} + \frac{1}{n} \Big( \sqrt{\frac{n_2}{n_1}}\sum_{i=1}^{n_1}X_{ij} - \sqrt{\frac{n_1}{n_2}}\sum_{i=1}^{n_2}Y_{ij} \Big)^2.
	\eea
	Under $H_{1j}$, 
	\bea
	\frac{1}{n\sigma_{0,jj}}\Big( \sqrt{\frac{n_2}{n_1}}\sum_{i=1}^{n_1}X_{ij} - \sqrt{\frac{n_1}{n_2}}\sum_{i=1}^{n_2}Y_{ij} \Big)^2 &\overset{d}{=}& \chi_1^2\Big(\frac{n_1n_2}{n\sigma_{0,jj}} (\mu_{0,1j}-\mu_{0,2j})^2 \Big),\\
	n_1 \what{\sigma}^2_{X_j} + n_2 \what{\sigma}^2_{Y_j} &\overset{d}{=}& \sigma_{0,jj} \chi_{n-2}^2 .
	\eea
	Define the sets
	\bea
	A_{3j} &=& \Big\{ \bfZ_n : (n\sigma_{0,jj})^{-1}\Big( \sqrt{\frac{n_2}{n_1}}\sum_{i=1}^{n_1}X_{ij} - \sqrt{\frac{n_1}{n_2}}\sum_{i=1}^{n_2}Y_{ij} \Big)^2 \ge 1 + \frac{n_1n_2}{n\sigma_{0,jj}} (\mu_{0,1j}-\mu_{0,2j})^2 \\
	&& \quad\quad\quad\quad\quad\quad - 2 \sqrt{\Big(1+2\frac{n_1n_2}{n\sigma_{0,jj}} (\mu_{0,1j}-\mu_{0,2j})^2 \Big) C\log(n\vee p) }\,\, \Big\} , \\
	A_{4j} &=& \Big\{ \bfZ_n:  \sigma_{0,jj}^{-1}(n_1 \what{\sigma}^2_{X_j} + n_2 \what{\sigma}^2_{Y_j}) \le n-2 + 2\sqrt{(n-2)C\log(n\vee p)} + 2C\log(n\vee p)  \Big\} , 
	\eea
	then $\sum_{l=3}^4 \bbP_0(A_{lj}^c) \le 2(n\vee p)^{-C} = o(1)$ by Lemma \ref{lem:chi_bounds} for any constant $C$ such that $0<C<C_1$.
	On the event $A_{3j}\cap A_{4j} $, 
	\bea
	&& n\what{\sigma}^2_{Z_j} - n_1 \what{\sigma}^2_{X_j} - n_2 \what{\sigma}^2_{Y_j}  \\
	&\ge& \sigma_{0,jj}\Bigg\{ 1+ \frac{n_1n_2}{n\sigma_{0,jj}} (\mu_{0,1j}-\mu_{0,2j})^2  - 2 \sqrt{\Big(1+2\frac{n_1n_2}{n\sigma_{0,jj}} (\mu_{0,1j}-\mu_{0,2j})^2\Big) C\log(n\vee p) } \Bigg\} .
	\eea
	Let $\lambda_n := n_1 n_2 (\mu_{0,1j} - \mu_{0,2j})^2/(n\sigma_{0,jj})$.
	Since $\log B_{10}(\tilde{X}_j, \tilde{Y}_j) \to \infty$ as $n_1\wedge n_2\to\infty$ on event $A_{3j}\cap A_{4j} $ when $\lambda_n \ge \epsilon_0 n$, we only need to consider the case $\lambda_n \le \epsilon_0 n$.
	Thus, on the event $A_{3j}\cap A_{4j}$, 
	\bea
	\log B_{10}(\tilde{X}_j, \tilde{Y}_j) 
	&\ge& \frac{1}{2}\log \Big(\frac{\gamma}{1+\gamma}\Big) 
	+ \frac{\lambda_n \Big(1+ \lambda_n^{-1} - 2\lambda_n^{-1}\sqrt{(1+2\lambda_n)C\log(n\vee p)}  \Big) }{2\{1+ 2 \sqrt{C \epsilon_0} + (1+2C + 2 \sqrt{3C})\epsilon_0 \} } \\
	&=& -\frac{\alpha}{2}\log(n\vee p) + \frac{1}{2} \log(1+\gamma) 
	+ \frac{\lambda_n \Big(1+ \lambda_n^{-1} - 2\lambda_n^{-1}\sqrt{(1+2\lambda_n)C\log(n\vee p)}  \Big) }{2\{1+ 2 \sqrt{C \epsilon_0} + (1+2C + 2 \sqrt{3C})\epsilon_0 \} } 
	\eea
	for all large $n_1 \wedge n_2$.
	Let 
	\bea
	K &=& \Big( \sqrt{2C_1} + \sqrt{2C_1 + \alpha C_1 \{ 1+ (1+ 8 C_1)\epsilon_0 \} } \Big)^2
	\eea
	then $\lambda_n \ge K \log(n\vee p)$ and
	\bea
	\lambda_n^{-1}\sqrt{(1+2\lambda_n)C\log(n\vee p)} 
	&=& \sqrt{(2 + \lambda_n^{-1})C \lambda_n^{-1} \log(n\vee p)}  \\
	&\le& \sqrt{(2 + K^{-1} (\log (n\vee p))^{-1} )CK^{-1} } \\
	&\le& \sqrt{2 C_1 K^{-1}}
	\eea
	for some constants $0<C < C_1$ and all large $n_1 \wedge n_2$, and
	\bea
	\frac{\lambda_n \Big(1+ \lambda_n^{-1} - 2\lambda_n^{-1}\sqrt{(1+2\lambda_n)C\log(n\vee p)}  \Big) }{2\{1+ 2 \sqrt{C \epsilon_0} + (1+2C + 2 \sqrt{3C})\epsilon_0 \} } 
	&\ge& \frac{K(1 - 2 \sqrt{2C_1 K^{-1}}) \log (n\vee p)}{2\{1+ 2 \sqrt{C \epsilon_0} + (1+2C + 2 \sqrt{3C})\epsilon_0 \}}  \\
	&\ge& \frac{\alpha C_1}{2} \log (n\vee p)
	\eea
	because 
	\bea
	K &\ge& \Big( \sqrt{2C_1} + \sqrt{2C_1 + \alpha C_1 \{ 1+ 2 \sqrt{C \epsilon_0} + (1+2C + 2 \sqrt{3C})\epsilon_0 \} } \Big)^2 .
	\eea
	It implies 
	\bea
	&& \bbP_0 \Big( \log B^{\mu}_{\max,10}(\bfX_{n_1}, \bfY_{n_2}) \ge  \{\alpha(C_1-1)/2\} \, \log(n\vee p)  \Big)  \\
	&\ge& \max_{1\le j\le p} \bbP_0 \Big( \log B_{10}(\tilde{X}_j, \tilde{Y}_j) \ge  \{\alpha(C_1-1)/2\} \, \log(n\vee p)  \Big)  \\
	&\ge& 1 - 2(n\vee p)^{-C}
	\eea
	for some constants $0<C < C_1$.
	
\end{proof}

\subsection{Proof of Theorem \ref{thm:two_cov}}

\begin{lemma}\label{lem:Gammaf_ineq}
	For any integer $n_1$ and $n_2$ such that $n_1 \wedge n_2 > 4$ and constant $a_0>0$, 
	\bea
	&& \sum_{l=1}^2\log \Gamma \Big(\frac{n_l}{2}+a_0 \Big) - \log \Gamma \Big(\frac{n}{2}+a_0 \Big) \\
	&\le& 2\big(1 - \log 2\big) - a_0 + a_0 \log \Big( \frac{n}{2} + a_0\Big)
	+ \sum_{l=1}^2 \Big(\frac{n_l -1}{2}+a_0\Big)\log \Big(\frac{n_l +2a_0}{n+2a_0} \Big) , \\
	&& \sum_{l=1}^2\log \Gamma \Big(\frac{n_l}{2}+a_0 \Big) - \log \Gamma \Big(\frac{n}{2}+a_0 \Big) \\
	&\ge& 2\big(1 - \log 2\big) - a_0  + \Big( \frac{1}{2}+ a_0\Big)  \log \Big( \frac{n}{2} + a_0\Big) 	  \\
	&+& \sum_{l=1}^2 \left\{ \Big(\frac{n_l}{2}+a_0\Big)\log \Big(\frac{n_l +2a_0}{n+2a_0} \Big) - \frac{1}{2}\log\Big(\frac{n_l}{2}+ a_0 \Big) \right\},
	\eea
	where $n= n_1+n_2$.
\end{lemma}
\begin{proof}
	By Theorem 1 of \cite{kevckic1971some}, 
	\bea
	(x-1)\log x +y - (y-1)\log y -x &\le& \log \Gamma(x) - \log \Gamma(y) \\
	&\le& (x- \frac{1}{2})\log x +y - (y-\frac{1}{2})\log y  - x
	\eea
	for any $x\ge y>1$.
	Thus, 
	\bea
	\log \Gamma\Big(\frac{n_l}{2} + a_0 \Big) - \log \Gamma(2)
	&\le& \Big(\frac{n_l-1}{2} + a_0 \Big) \log \Big(\frac{n_l}{2} + a_0\Big) + 2 - \frac{3}{2}\log 2 - \Big(\frac{n_l}{2} + a_0\Big), \\
	- \log \Gamma\Big(\frac{n}{2}+a_0\Big) + \log \Gamma(2) 
	&\le& - \Big(\frac{n}{2}+a_0 \Big) \log \Big(\frac{n}{2}+a_0 \Big) -2 + \log 2 + \Big(\frac{n}{2}+a_0 \Big)
	\eea
	for $l=1,2$.
	This completes the proof for the first inequality.
	By Theorem 1 of \cite{kevckic1971some}, we also obtain
	\bea
	- \log \Gamma \Big(\frac{n}{2}+a_0 \Big) +\log \Gamma \Big(\frac{n_1}{2}+a_0 \Big)
	&\ge& - \Big(\frac{n}{2}+a_0 -\frac{1}{2} \Big) \log \Big(\frac{n}{2}+a_0 \Big) - \Big(\frac{n_1}{2}+a_0\Big) \\
	&+& \Big(\frac{n_1}{2}+a_0-\frac{1}{2}\Big) \log \Big(\frac{n_1}{2}+a_0 \Big) + \Big(\frac{n}{2}+a_0\Big), \\
	\log \Gamma \Big(\frac{n_1}{2}+a_0 \Big) &\ge& \Big(\frac{n_2}{2}+a_0-1\Big) \log \Big(\frac{n_2}{2}+a_0\Big) + 2 - \log 2 - \Big(\frac{n_2}{2}+a_0 \Big),
	\eea
	which completes the proof for the second inequality.
\end{proof}

\begin{lemma}\label{lem:setAij0}
	Suppose $\sg_1= \sg_2 = \sg_0 = (\sigma_{0,ij})$ in model \eqref{two_cov}.
	Define 
	\bea
	A_{0ij} &=& \Big\{ \bfZ_n: 1 - 2\sqrt{C\log(n\vee p)} \le D_{ij} \le 1 + 2 \sqrt{C\log(n\vee p)} + 2C\log(n\vee p)  \Big\} ,
	\eea
	where 
	\bea
	D_{ij} &=& \frac{1}{\tau_{0,ij} \|\tilde{Z}_j\|_2^2} \Big( \|\tilde{Y}_j\|_2 \frac{\tilde{X}_i^T\tilde{X}_j }{\|\tilde{X}_j\|_2 } - \|\tilde{X}_j\|_2 \frac{\tilde{Y}_i^T\tilde{Y}_j }{\|\tilde{Y}_j\|_2 }  \Big)^2 .
	\eea
	Then, 
	$$\sum_{i=1}^p\sum_{j=1}^p \bbP_0 (A_{0ij}^c) \le 2 p^2 (n\vee p)^{-C}$$ 
	for any constant $C>0$.
\end{lemma}
\begin{proof}
	By \eqref{cond_dist_XY}, we have
	\bea
	\tilde{X}_i^T \tilde{X}_j \mid \tilde{X}_j &\sim& N \big( a_{0,ij}\|\tilde{X}_j\|_2^2, \tau_{0,ij} \|\tilde{X}_j\|_2^2 \big), \\
	\tilde{Y}_i^T \tilde{Y}_j \mid \tilde{Y}_j &\sim& N\big( a_{0,ij}\|\tilde{Y}_j\|_2^2, \tau_{0,ij} \|\tilde{Y}_j\|_2^2 \big).
	\eea
	Then, it is easy to check that
	\bea
	\|\tilde{Y}_j\|_2 \frac{\tilde{X}_i^T\tilde{X}_j }{\|\tilde{X}_j\|_2 } - \|\tilde{X}_j\|_2 \frac{\tilde{Y}_i^T\tilde{Y}_j }{\|\tilde{Y}_j\|_2 } \mid \tilde{Z}_j
	&\sim& N\Big( 0, \tau_{0,ij} \|\tilde{Z}_j\|_2^2  \Big),
	\eea
	which implies $D_{ij} \sim \chi_1^2$.
	Thus, by Lemma \ref{lem:chi_bounds}, 
	\bea
	\sum_{i=1}^p\sum_{j=1}^p \bbP_0 (A_{0ij}^c) 
	&\le& \sum_{i=1}^p\sum_{j=1}^p  2 \exp \Big(- C \log (n\vee p) \Big) \\
	&\le& 2 p^2(n \vee p)^{-C}.
	\eea
\end{proof}

\begin{lemma}\label{lem:setAij1-8}
	Consider model \eqref{two_cov} with true covariances $\sg_{01}=(\sigma_{01,ij})$ and $\sg_{02}=(\sigma_{02,ij})$.
	Define the sets
	\bea
	A_{1ij} &=& \Big\{ \bfZ_n :  \Big(1-2 \sqrt{\frac{C \epsilon_{01}}{1-n_1^{-1}} } \, \Big) \le  \frac{\what{\tau}_{1,ij}}{\tau_{01,ij}(1- n_1^{-1})} \le \Big( 1 + 2 \sqrt{\frac{C \epsilon_{01}}{1-n_1^{-1}}}+  \frac{2C\epsilon_{01}}{1-n_1^{-1}} \Big) \,\, \Big\} , 
	\eea
	where $\tau_{0k,ij} = \sigma_{0k,ii}(1- \rho_{0k,ij}^2)$ and $\rho_{0k,ij}= \sigma_{0k, ij}/(\sigma_{0k,ii} \sigma_{0k,jj} )^{1/2}$ for $k=1,2$.
	Similarly, define the sets
	\bea
	A_{2ij} &=& \Big\{ \bfZ_n :  \Big(1-2 \sqrt{\frac{C \epsilon_{02}}{1-n_2^{-1}} } \, \Big) \le  \frac{\what{\tau}_{2,ij}}{\tau_{02,ij}(1- n_2^{-1})} \le \Big( 1 + 2 \sqrt{\frac{C \epsilon_{02}}{1-n_2^{-1}}}+  \frac{2C\epsilon_{02}}{1-n_2^{-1}} \Big) \,\, \Big\} ,
	\eea
	Then, 
	$$\sum_{l=1}^2 \bbP_0 (A_{lij}^c) \le 4(n\vee p)^{-C}$$ 
	for any constant $C>0$.
\end{lemma}
\begin{proof}
	Since $\tilde{X}_i\mid \tilde{X}_j \sim N_{n_1} (a_{01,ij}\tilde{X}_j , \tau_{01,ij}I_{n_1} )$, we have 
	\bea
	n_1\what{\tau}_{1,ij}/\tau_{01,ij} &=& \tilde{X}_i^T (I_{n_1} - H_{\tilde{X}_j})\tilde{X}_i /\tau_{01,ij}  \,\sim\, \chi_{n_1-1}^2 .
	\eea
	Then,
	\bea
	\frac{\what{\tau}_{1,ij}}{\tau_{01,ij}} &\le& \frac{1}{n_1} \Big\{ n_1-1 + 2 \sqrt{(n_1-1)C \log(n\vee p)} + 2C \log(n\vee p) \Big\} \\
	&\le&\big(1 - n_1^{-1} \big)\Big\{ 1 + 2 \sqrt{\frac{C \log(n\vee p)}{n_1-1}} + 2C \frac{\log(n\vee p)}{n_1-1} \Big\} \\
	&=& \big(1 - n_1^{-1} \big)\Big\{ 1 + 2 \sqrt{C\epsilon_{01}/(1-n_1^{-1})} + 2 C \epsilon_{01}/(1-n_1^{-1}) \Big\}
	\eea
	with probability at least $1 - (n\vee p)^{-C}$ and 
	\bea
	\frac{\what{\tau}_{1,ij}}{\tau_{01,ij}} &\ge& \frac{1}{n_1}\Big\{ n_1-1 -2\sqrt{(n_1-1)C\log (n\vee p)}  \Big\} \\
	&=& \big(1 - n_1^{-1} \big) \big(1 - 2\sqrt{C\epsilon_{01}/(1-n_1^{-1})} \big)
	\eea
	with probability at least $1 - (n\vee p)^{-C}$, by Lemma \ref{lem:chi_bounds}.
	Thus, 
	\bea
	\bbP_0 (A_{1ij}^c) 
	&\le& 2 (n\vee p)^{-C}.
	\eea
	By similar arguments, it is easy to see that
	\bea
	\bbP_0 (A_{2 ij}^c) 
	&\le& 2 (n\vee p)^{-C},
	\eea
	which completes the proof.
	
\end{proof}

\bigskip

\begin{proof}[Proof of Theorem \ref{thm:two_cov}]
	We prove the consistency under the null and alternative in turn.
	\paragraph{Consistency under $H_{0}: \sg_1 = \sg_2$.}
	We first assume the true null $H_{0}: \sg_1 = \sg_2$, i.e.,  $\sg_0  \equiv \sg_{01}= \sg_{02}$, where $\sg_0 = (\sigma_{0,ij})$.
	For a given pair $(i,j), i\neq j$, the PBF is
	\bea
	&& \log B_{10}(\tilde{X}_i, \tilde{Y}_i, \tilde{X}_j,\tilde{Y}_j) \\
	&=& \frac{1}{2}\log \Big(\frac{\gamma}{1+\gamma} \Big) + \log \Gamma\Big(\frac{n_1}{2}+a_0 \Big) + \log \Gamma\Big(\frac{n_2}{2}+a_0 \Big) - \log \Gamma\Big(\frac{n}{2}+a_0 \Big) + \log \Big(\frac{b_{01,ij}^{a_0} b_{02,ij}^{a_0}}{b_{0,ij}^{a_0}\Gamma(a_0)} \Big) \\
	&-& \Big(\frac{n_1}{2}+a_0 \Big) \log \big(b_{01,ij}+ \frac{n_1}{2}\what{\tau}_{1,ij} \big) - \Big(\frac{n_2}{2}+a_0 \Big) \log \big(b_{02,ij}+ \frac{n_2}{2}\what{\tau}_{2,ij} \big) \\
	&+& \Big(\frac{n}{2}+a_0 \Big) \log \big(b_{0,ij}+ \frac{n}{2}\what{\tau}_{ij} \big).
	\eea
	By Lemma \ref{lem:Gammaf_ineq},
	\bea
	&& \log B_{10}(\tilde{X}_i, \tilde{Y}_i, \tilde{X}_j,\tilde{Y}_j) \\
	&\le& \frac{1}{2}\log \Big(\frac{\gamma}{1+\gamma} \Big) + 2(1-\log 2) - a_0 - \frac{1}{2}\left\{\log\Big(\frac{n_1+2a_0}{n+2a_0}\Big)  + \log\Big(\frac{n_2+2a_0}{n+2a_0}\Big) \right\} \\
	&-& a_0 \log \Big(\frac{n\what{\tau}_{ij}+2 b_{0,ij} }{n + 2a_0} \Big) + a_0 \log \Big(\frac{n\what{\tau}_{ij}+2 b_{0,ij} }{n_1\what{\tau}_{1,ij}+2 b_{01,ij}  } \Big) + a_0\log \Big(\frac{n\what{\tau}_{ij}+2 b_{0,ij} }{n_2\what{\tau}_{2,ij}+2 b_{02,ij}  } \Big) \\
	&+& a_0 \log\Big(\frac{n_1+2a_0}{n+2a_0} \Big) + a_0 \log\Big(\frac{n_2+2a_0}{n+2a_0} \Big) + \log \Big(\frac{b_{01,ij}^{a_0} b_{02,ij}^{a_0}}{b_{0,ij}^{a_0}\Gamma(a_0)} \Big) \\
	&+& \frac{n_1}{2} \log \Big(\frac{n\what{\tau}_{ij}+2 b_{0,ij} }{n_1\what{\tau}_{1,ij}+ 2 b_{01,ij}  } \Big) +  \frac{n_2}{2}  \log \Big(\frac{n\what{\tau}_{ij}+2 b_{0,ij} }{n_2\what{\tau}_{2,ij}+2 b_{02,ij}  } \Big) \\
	&+&  \frac{n_1}{2}  \log\Big(\frac{n_1+2a_0}{n+2a_0} \Big) +  \frac{n_2}{2}  \log\Big(\frac{n_2+2a_0}{n+2a_0} \Big) .
	\eea
	We first focus on the last two lines.
	Note that
	\bea
	n\what{\tau}_{ij} &=& n_1 \what{\tau}_{1,ij} + n_2 \what{\tau}_{2,ij} + \frac{1}{\| \tilde{Z}_j\|_2^2 } \Big( \|\tilde{Y}_j\|_2 \frac{\tilde{X}_i^T\tilde{X}_j }{\|\tilde{X}_j\|_2 } - \|\tilde{X}_j\|_2 \frac{\tilde{Y}_i^T\tilde{Y}_j }{\|\tilde{Y}_j\|_2 }  \Big)^2 \\
	&\equiv& n_1 \what{\tau}_{1,ij} + n_2 \what{\tau}_{2,ij} + \tau_{0,ij} D_{ij},
	\eea
	where $D_{ij} \sim \chi_1^2$ under $H_0$.
	It implies 
	\bean
	&&  \frac{n_1}{2} \log \Big(\frac{n\what{\tau}_{ij}+2 b_{0,ij} }{n_1\what{\tau}_{1,ij}+2 b_{01,ij}  } \Big) +  \frac{n_1}{2} \log\Big(\frac{n_1+2a_0}{n+2a_0} \Big) +  \frac{n_2}{2}  \log \Big(\frac{n\what{\tau}_{ij}+2 b_{0,ij} }{n_2\what{\tau}_{2,ij}+2 b_{02,ij}  } \Big) +  \frac{n_2}{2}  \log\Big(\frac{n_2+2a_0}{n+2a_0} \Big) \nonumber  \\
	&\le&  \frac{n_1}{2}  \left[ \log \Big( \frac{ n^{-1}n_1\what{\tau}_{1,ij} + n^{-1}n_2 \what{\tau}_{2,ij} + n^{-1}\tau_{0,ij}D_{ij} + 2 b_{0,ij} n^{-1}  }{ \what{\tau}_{1,ij} + 2 b_{01,ij} n_1^{-1} } \Big) + \log \Big(\frac{n(n_1+2a_0)}{n_1(n+2a_0)} \Big)  \right] \nonumber \\
	&&+\,\, \frac{n_2}{2} \left[ \log \Big( \frac{ n^{-1}n_1\what{\tau}_{1,ij} + n^{-1}n_2 \what{\tau}_{2,ij} + n^{-1}\tau_{0,ij}D_{ij} + 2  b_{0,ij} n^{-1}  }{ \what{\tau}_{2,ij} + 2 b_{02,ij} n_2^{-1} } \Big) + \log \Big(\frac{n(n_2+2a_0)}{n_2(n+2a_0)} \Big)  \right] \nonumber\\
	&\le& \frac{n}{2}\left[ \log \Big( \frac{n_1}{n}\what{\tau}_{1,ij} + \frac{n_2}{n}\what{\tau}_{2,ij} + \frac{\tau_{0,ij}}{n}D_{ij} + \frac{2 b_{0,ij}}{n} \Big) - \frac{n_1}{n}\log (\what{\tau}_{1,ij} ) - \frac{n_2}{n}\log (\what{\tau}_{2,ij} )    \right] \nonumber\\
	&&+\,\, \frac{n_1}{2} \log \Big(\frac{n(n_1+2a_0)}{n_1(n+2a_0)} \Big) + \frac{n_2}{2} \log \Big(\frac{n(n_2+2a_0)}{n_2(n+2a_0)} \Big)  \nonumber\\
	&=&  \frac{n}{2}\left[ \log \Big( \frac{n_1}{n}\frac{\what{\tau}_{1,ij}}{\what{\tau}_{2,ij}} + \frac{n_2}{n} + \frac{\tau_{0,ij}}{n \what{\tau}_{2,ij}}D_{ij} + \frac{2 b_{0,ij}}{n \what{\tau}_{2,ij}} \Big) - \frac{n_1}{n}\log \Big(\frac{\what{\tau}_{1,ij}}{\what{\tau}_{2,ij}} \Big)  \right] \label{part1_BF10_2} \\
	&& +\,\,  \frac{n_1}{2} \log \Big(\frac{n(n_1+2a_0)}{n_1(n+2a_0)} \Big) + \frac{n_2}{2} \log \Big(\frac{n(n_2+2a_0)}{n_2(n+2a_0)} \Big). \label{part1_BF10}
	\eean
	It is easy to see that \eqref{part1_BF10} is bounded by a constant because
	\bea
	\frac{n_1}{2} \log \Big(\frac{n(n_1+2a_0)}{n_1(n+2a_0)} \Big) + \frac{n_2}{2} \log \Big(\frac{n(n_2+2a_0)}{n_2(n+2a_0)} \Big) &\le& a_0,
	\eea
	so we only need to focus on \eqref{part1_BF10_2}.

%	Define the sets
%	\bea
%	A_{0ij} &=& \Big\{ \bfZ_n: 1 - 2\sqrt{C\log(n\vee p)} \le D_{ij} \le 1 + 2 \sqrt{C\log(n\vee p)} + 2C\log(n\vee p)  \Big\}   , \\
%	A_{1ij} &=& \Big\{ \bfZ_n :  \big(1- n_1^{-1}\big)\Big(1-2 \sqrt{\frac{C \epsilon_{01}}{1-n_1^{-1}} } \, \Big) \le  \tau_{0,ij}^{-1} \what{\tau}_{1,ij} \le 1 + 2 \sqrt{C \epsilon_{01}}+ 2 C\epsilon_{01} \,\, \Big\} , \\
%	A_{2ij} &=& \Big\{ \bfZ_n :    \tau_{0,ij}^{-1}\tilde{X}_i^T H_{\tilde{X}_j }\tilde{X}_i \le 1 + \frac{a_{0,ij}^2}{\tau_{0,ij}}\|\tilde{X}_j\|_2^2 + 2 \sqrt{( 1+ 2 a_{0,ij}^2 \|\tilde{X}_j\|_2^2/\tau_{0,ij} ) C\log(n\vee p)} + 2C \log(n\vee p)  \Big\} , \\
%	A_{3ij} &=& \Big\{ \bfZ_n: \tau_{0,ij}^{-1}\tilde{X}_i^T H_{\tilde{X}_j }\tilde{X}_i \ge  1 + \frac{a_{0,ij}^2}{\tau_{0,ij}}\|\tilde{X}_j\|_2^2 - 2 \sqrt{( 1+ 2 a_{0,ij}^2 \|\tilde{X}_j\|_2^2/\tau_{0,ij} ) C\log(n\vee p)} \Big\}, \\
%	A_{4ij} &=& \Big\{ \bfZ_n: n_1\big(1-2\sqrt{C\epsilon_{01}} \big)  \le  \sigma_{0,jj}^{-1}\|\tilde{X}_j\|_2^2 \le n_1 \big( 1 + 2 \sqrt{C\epsilon_{01}} + 2C\epsilon_{01} \big)  \Big\}.
%	\eea
%	We also define the sets $\{ A_{lij}: l=5,\ldots,8\}$ using $n_2, \epsilon_{02}$ and $\bfY_{n_2}$ instead of $n_1, \epsilon_{01}$ and $\bfX_{n_1}$ in $\{A_{lij}: l=1,\ldots,4\}$, respectively.
%	Then, $\sum_{i=1}^p\sum_{j=1}^p \sum_{l=0}^8 \bbP_0 (A_{lij}^c) \le 12(n\vee p)^{-C+2} = o(1)$ by Lemma \ref{lem:chi_bounds} for any constant $C>2$.
	
	Let $\what{\tau}_{1,ij}/\what{\tau}_{2,ij} \equiv 1 + R_{ij}$ and 
	\bea
	R_{ij}' &=& R_{ij} + \frac{\tau_{0,ij}}{n_1 \what{\tau}_{2,ij} } D_{ij} + \frac{2 b_{0,ij}}{n_1 \what{\tau}_{2,ij}} .
	\eea
	Then, we can rewrite \eqref{part1_BF10_2} as
	\bean
	&& \frac{n}{2}\Big[ \log \big( 1 + \frac{n_1}{n} R_{ij}' \big) - \frac{n_1}{n} \log (1+ R_{ij}) \Big] \nonumber \\
	&=& \frac{n}{2}\Big[ - \Big\{ \frac{n_1}{n} R_{ij}'  - \log \big(1+ \frac{n_1}{n}R_{ij}' \big) \Big\} + \frac{n_1}{n}\Big\{ R_{ij} - \log (1+R_{ij}) \Big\} + \frac{n_1}{n} \big(R_{ij}' - R_{ij} \big) \Big] . \label{log_expan}
	\eean
	Define a set $A_{ij}:= \cap_{l=0}^2 A_{lij}$ with some constant $C>2$, where $A_{0ij},A_{1ij}$ and $A_{2ij}$ are defined at Lemmas \ref{lem:setAij0} and \ref{lem:setAij1-8}.
	By Lemmas \ref{lem:setAij0} and \ref{lem:setAij1-8}, it suffices to focus on the event $\cap_{i=1}^p \cap_{j=1}^p A_{ij}$ to prove Theorem \ref{thm:two_cov} because 
	\bea
	\sum_{i=1}^p \sum_{j=1}^p \bbP_0(A_{ij}^c) &\le& 6 (n\vee p)^{-C+2} \,\,=\,\, o(1)
	\eea
	for any constant $C>2$.
	On the event $A_{ij}$, we have the following bounds:
	\bea
	\frac{\what{\tau}_{1,ij}}{\what{\tau}_{2,ij}} 
	&\le& \frac{1 + 2 \sqrt{C \epsilon_{01}}+ 2 C\epsilon_{01}}{(1- n_2^{-1}) \big(1- 2 \sqrt{C \epsilon_{02} /(1-n_2^{-1}) \,} \big) }  \\
	&\le& (1-n_2^{-1})^{-1} \Big( 1 + 2 \sqrt{C \epsilon_{01}}+ 2 \sqrt{C\epsilon_{02}/(1-n_2^{-1})} + O(\epsilon_{01} \vee \epsilon_{02})\Big) \\
	&\le&  ( 1+ 3 n_2^{-1} ) \Big( 1 + 2 \sqrt{C\epsilon_{01}} + 2 \sqrt{C\epsilon_{02}(1+ 3 n_2^{-1})} + O(\epsilon_{01} \vee \epsilon_{02})\Big)   \\
	&\le&  ( 1+ 3 n_2^{-1} ) \Big( 1 + 2 \sqrt{C\epsilon_{01}}+ 2 \sqrt{C\epsilon_{02}}(1+2 n_2^{-1}) + O(\epsilon_{01} \vee \epsilon_{02})\Big) \\
	&\le&  1 + 4 \sqrt{C (\epsilon_{01}\vee \epsilon_{02})} + O(\epsilon_{01} \vee \epsilon_{02}) \\
	\frac{\what{\tau}_{1,ij}}{\what{\tau}_{2,ij}} &\ge& \frac{(1-n_1^{-1}) \big( 1- 2 \sqrt{C\epsilon_{01} /(1-n_1^{-1}) \,} \big) }{1+ 2 \sqrt{C \epsilon_{01}}+ 2 C\epsilon_{01}}  \\
	&\ge& (1-n_1^{-1}) \Big( 1 - 2 \sqrt{C\epsilon_{01} /(1-n_1^{-1}) \,}  - 2 \sqrt{C \epsilon_{02}} - 2 C \epsilon_{01}  \Big) \\
	&\ge&  1 - 4 \sqrt{C(\epsilon_{01}\vee \epsilon_{02})}   - O(\epsilon_{01} \vee \epsilon_{02})
	\eea
	for all large $n_1\wedge n_2$.
	Thus, on the event $A_{ij}$, 
	\bea
	-4\sqrt{C(\epsilon_{01} \vee \epsilon_{02})}  - O(\epsilon_{01} \vee \epsilon_{02}) \le  R_{ij} \le 4 \sqrt{C (\epsilon_{01}\vee \epsilon_{02})} + O(\epsilon_{01} \vee \epsilon_{02}) 
	\eea
	and
	\bea
	R_{ij}' - R_{ij} &=& \frac{\tau_{0,ij}}{n_1 \what{\tau}_{2,ij} } D_{ij} + \frac{2 b_{0,ij}}{n_1 \what{\tau}_{2,ij} } \\ &\le& \frac{1 + 2\sqrt{C \log(n\vee p)}+ 2C \log(n\vee p)}{n_1(1-n_2^{-1})\big(1 - 2\sqrt{C \epsilon_{02}(1-n_2^{-1})^{-1}}\big) } \\
	&&\,\,+\,\, \frac{2 b_{0,ij}}{\tau_{0,ij}n_1(1-n_2^{-1})\big(1 - 2\sqrt{C \epsilon_{02}(1-n_2^{-1})^{-1}}\big)} \\
	&\le& 2 C  \epsilon_{01} + o(\epsilon_{01}\vee \epsilon_{02})
	\eea
	for all sufficiently large $n_1\wedge n_2$ because  $\tau_{0,ij} \gg (\log (n\vee p))^{-1}$ by condition (A2).
	Thus $|R_{ij}|$ and $|R_{ij}'|$ can be regarded as small values close to $0$ on $A_{ij}$.
	
	Note that $x-\log(1+x) \le x^2/2 + |x|^3/(3(1+x))$ and $x-\log (1+x)\ge x^2/2 - |x|^3/3$ for any small $|x|>0$.
	Then, on the event $A_{ij}$, the upper bound of \eqref{log_expan} can be derived as
	\bea
	&& \frac{n}{2}\Big[ -\Big\{ \frac{1}{2}\Big(\frac{n_1}{n} R_{ij}' \Big)^2 - \frac{1}{3} \Big(\frac{n_1}{n} |R_{ij}'|\Big)^3 \Big\} + \frac{n_1}{2n} R_{ij}^2 + \frac{n_1}{3n}\frac{|R_{ij}^3|}{1 + R_{ij}}  + \frac{n_1}{n} \big(R_{ij}' - R_{ij} \big) \Big] \\
	&=& \frac{n_1}{4} R_{ij}^2 \Big[ 1 - \frac{n_1}{n}\Big(\frac{R_{ij}'}{R_{ij}} \Big)^2  + \frac{2}{3} \Big(\frac{n_1}{n} \Big)^2 |R_{ij}|\Big(\frac{R_{ij}'}{R_{ij}} \Big)^2 + \frac{2|R_{ij}|}{3(1 + R_{ij})}  \Big] + \frac{n_1}{2} \big(R_{ij}' - R_{ij} \big) \\
	&=& \frac{n_1}{4} R_{ij}^2 \Big[ 1 - \frac{n_1}{n}  + \frac{2}{3} \Big(\frac{n_1}{n} \Big)^2 |R_{ij}| + \frac{2}{3}\frac{|R_{ij}|}{1 + R_{ij}}  \Big] + \frac{n_1}{2} \big(R_{ij}' - R_{ij} \big)     \\
	&& -\,\, \frac{n_1^3}{6n^2} |R_{ij}|^3   \Big\{ 1- \Big(\frac{R_{ij}'}{R_{ij}} \Big)^2  \Big\} +   \frac{n_1^2}{4n} R_{ij}^2 \Big\{ 1- \Big(\frac{R_{ij}'}{R_{ij}} \Big)^2  \Big\}  \\
	&\le& 4n_1 C_2(\epsilon_{01}\vee \epsilon_{02}) \left[1 - \frac{n_1}{n}  \right] + \frac{n_1}{2} \, 2C_2\epsilon_{01} + o \big( n_1(\epsilon_{01} \vee \epsilon_{02}) \big) \\
	&\le& \frac{n_1}{2}(\epsilon_{01}\vee \epsilon_{02}) \big( 8 C_2 (1- \frac{n_1}{n}) + 2C_2 \big) + o \big( n_1(\epsilon_{01} \vee \epsilon_{02}) \big) \\
	&=& \frac{n_1}{2}(\epsilon_{01}\vee \epsilon_{02}) \big( 6C_2 + o(1) \big) 
	\eea
	for some constant $C_2$ such that $C_2>C>2$.
	Also note that on the event $A_{ij}$, the rest part of $\log B_{10}(\tilde{X}_i, \tilde{Y}_i, \tilde{X}_j, \tilde{Y}_j)$, except $\log (\gamma/(1+\gamma)) /2$, is bounded above by $\log (\log (n\vee p))$ due to the assumption $\tau_{0,ij} \gg ( \log(n\vee p) )^{-1}$.
	Since $\log (\gamma/(1+\gamma)) /2 \le  - \alpha \log(n\vee p)/2$, we have
	\bea
	\log B_{10}(\tilde{X}_i, \tilde{Y}_i, \tilde{X}_j, \tilde{Y}_j) &\le& - \frac{1}{2}\Big(\alpha - 6 C_2  + o(1)\Big) \log(n\vee p)
	\eea
	on the event $A_{ij}$.
	Thus,
	\bea
	&&\bbP_0 \left( \log B^{\sg}_{\max,10}(\bfX_{n_1}, \bfY_{n_2}) < - \frac{1}{2}\big(\alpha - 6 C_2  + o(1)\big) \log(n\vee p)  \right)  \\
	&=& 1 -\bbP_0 \left( \log B^{\sg}_{\max,10}(\bfX_{n_1}, \bfY_{n_2}) > - \frac{1}{2}\big(\alpha - 6 C_2  + o(1)\big) \log(n\vee p)  \right)  \\
	&\ge& 1- \sum_{(i,j): i\neq j} \bbP_0 \left( \log B_{10}(\tilde{X}_i, \tilde{Y}_i, \tilde{X}_j, \tilde{Y}_j) > - \frac{1}{2}\big(\alpha - 6 C_2 + o(1)\big) \log(n\vee p)  \right)  \\
	&\ge& 1- \sum_{(i,j): i\neq j} \bbP_0 (A_{ij}^c)  \\
	&\ge& 1 - 6 (n \vee p)^{-C+2} \,\, = \,\, 1 - o(1).
	\eea
	It completes the proof under the true null $H_0: \sg_1=\sg_2$ because we assume $\alpha > 6C_2 $.
	
	\paragraph{Consistency under $H_{1}: \sg_1 \neq \sg_2$.}
	Specifically, assume that $H_{1,ij}: a_{1,ij}\neq a_{2,ij}$ or $\tau_{1,ij}\neq \tau_{2,ij}$ is true for some pair $(i,j)$.
	First, we focus on the case $\tau_{01,ij}\neq \tau_{02,ij}$ and $a_{01,ij}= a_{02,ij} \equiv a_{0,ij}$, and suppose condition (A3) holds.
	By Lemma \ref{lem:Gammaf_ineq},
	\bean
	\log B_{10}(\tilde{X}_i, \tilde{Y}_i, \tilde{X}_j, \tilde{Y}_j) 
	&\ge& \frac{1}{2}\log \Big( \frac{\gamma}{1+\gamma}\Big) + 2 - \log 2 - a_0  + \log \Big(\frac{b_{01,ij}^{a_0} b_{02,ij}^{a_0}}{b_{0,ij}^{a_0}\Gamma(a_0)} \Big) \nonumber\\
	&-& \frac{1}{2}\log \Big(\frac{n_1}{2}+a_0 \Big) -  \log \Big(\frac{n_2}{2}+a_0 \Big) + \frac{1}{2}\log \Big(\frac{n}{2}+a_0 \Big) \nonumber\\
	&+& \Big\{  -a_0 \Big[\log \Big( \frac{n}{2}\what{\tau}_{ij}+ 	b_{0,ij}  \Big) - \log \Big(\frac{n}{2}+a_0 \Big) \Big] \label{BF10_H1_0} \\
	&&+\,\, \Big(\frac{n_1}{2}+ a_0\Big) \Big[  \log \Big( \frac{n \what{\tau}_{ij} + 2 b_{0,ij} }{n_1 \what{\tau}_{1,ij} + 2 b_{01,ij}} \Big) + \log \Big(\frac{n_1+2a_0}{n+2a_0}\Big) \Big] \label{BF10_H1_1} \\
	&&+\,\, \Big(\frac{n_2}{2}+ a_0\Big) \Big[  \log \Big( \frac{n \what{\tau}_{ij} + 2 b_{0,ij} }{n_2 \what{\tau}_{2,ij} + 2 b_{02,ij}} \Big) + \log \Big(\frac{n_2+2a_0}{n+2a_0}\Big) \Big]  \Big\}. \label{BF10_H1_2}
	\eean
	Since the sum of three terms \eqref{BF10_H1_0}, \eqref{BF10_H1_1} and \eqref{BF10_H1_2} is increasing in $\what{\tau}_{ij}$ and 
	\bea
	n\what{\tau}_{ij} &=& n_1 \what{\tau}_{1,ij} + n_2 \what{\tau}_{2,ij} + \frac{1}{\|\tilde{Z}_j\|_2^2}\Big(\|\tilde{Y}_j\|_2 \frac{\tilde{X}_i^T\tilde{X}_j}{\|\tilde{X}_j\|_2} - \|\tilde{X}_j\|_2 \frac{\tilde{Y}_i^T\tilde{Y}_j}{\|\tilde{Y}_j\|_2}\Big)^2\\
	&\ge& n_1 \what{\tau}_{1,ij} + n_2 \what{\tau}_{2,ij} ,
	\eea
	a lower bound for the sum of three terms \eqref{BF10_H1_0}, \eqref{BF10_H1_1} and \eqref{BF10_H1_2} is given by
	\bean
	&& 
%	-a_0 \Big[ \log \Big(\frac{1}{2}(n_1 \what{\tau}_{1,ij,\gamma} + n_2 \what{\tau}_{2,ij,\gamma}) +2b_0  \Big) - \log \Big(\frac{n}{2}+a_0 \Big)\Big] \nonumber\\
%	&+& \frac{n_1}{2} \Big[  \log \Big( \frac{n \what{\tau}_{ij,\gamma} + 2b_0 }{n_1 \what{\tau}_{1,ij,\gamma} + 2b_0} \Big) + \log \Big(\frac{n_1+2a_0}{n+2a_0}\Big) \Big] \nonumber \\
%	&+& \frac{n_2}{2} \Big[  \log \Big( \frac{n \what{\tau}_{ij,\gamma} + 2b_0 }{n_2 \what{\tau}_{2,ij,\gamma} + 2b_0} \Big) + \log \Big(\frac{n_2+2a_0}{n+2a_0}\Big) \Big] \nonumber \\
%	&+& a_0 \log \Big(\frac{n_1+2a_0}{n+2a_0}\Big) + a_0 \log \Big(\frac{n_2+2a_0}{n+2a_0}\Big)  \nonumber \\
%	&\ge&  
%	-a_0 \Big[ \log \Big(\frac{1}{2}(n_1 \what{\tau}_{1,ij,\gamma} + n_2 \what{\tau}_{2,ij,\gamma}) +2b_0  \Big) - \log \Big(\frac{n}{2}+a_0 \Big)\Big]  \nonumber\\
%	&+& \frac{n_1}{2} \Big[ \log \Big( \frac{\frac{n_1}{n} \what{\tau}_{1,ij,\gamma} + \frac{n_2}{n} \what{\tau}_{2,ij,\gamma} + \frac{2b_0}{n}}{ \what{\tau}_{1,ij,\gamma} + \frac{2b_0}{n_1}} \Big) + \log \Big(\frac{n(n_1+2a_0)}{n_1(n+2a_0)}\Big)  \Big] \nonumber \\
%	&+& \frac{n_2}{2} \Big[ \log \Big( \frac{\frac{n_1}{n} \what{\tau}_{1,ij,\gamma} + \frac{n_2}{n} \what{\tau}_{2,ij,\gamma} + \frac{2b_0}{n}}{ \what{\tau}_{2,ij,\gamma} + \frac{2b_0}{n_2}} \Big) + \log \Big(\frac{n(n_2+2a_0)}{n_2(n+2a_0)}\Big)  \Big] \nonumber \\
%	&+& a_0 \log \Big(\frac{n_1+2a_0}{n+2a_0}\Big) + a_0 \log \Big(\frac{n_2+2a_0}{n+2a_0}\Big)  \nonumber \\
%	&=& 
	-a_0 \Big[ \log \Big(\frac{1}{2}(n_1 \what{\tau}_{1,ij} + n_2 \what{\tau}_{2,ij}) + b_{0,ij}  \Big) - \log \Big(\frac{n}{2}+a_0 \Big)\Big] \label{BF10_H1_00}\\
	&+& \frac{n}{2}\Big[ \log\Big( \frac{n_1}{n} \what{\tau}_{1,ij} + \frac{n_2}{n} \what{\tau}_{2,ij} + \frac{2 b_{0,ij}}{n}\Big) - \frac{n_1}{n}\log\Big( \what{\tau}_{1,ij} + \frac{2 b_{01,ij}}{n_1}\Big) - \frac{n_2}{n}\log\Big( \what{\tau}_{2,ij} + \frac{2 b_{02,ij}}{n_2}\Big)  \Big]  \label{BF10_H1_3} \\
	&+& a_0 \log \Big(\frac{n_1+2a_0}{n+2a_0}\Big) + a_0 \log \Big(\frac{n_2+2a_0}{n+2a_0}\Big).   \label{BF10_H1_4}
	\eean
	Term \eqref{BF10_H1_4} is bounded below by a constant because we assume $n_1/n \to 1/2$ as $n_1 \wedge n_2\to\infty$.
	
	We will first calculate a lower bound for \eqref{BF10_H1_3}, and then calculate a lower bound for \eqref{BF10_H1_00}.
	Since we are considering the case $\tau_{01,ij} \neq \tau_{02,ij}$, without loss of generality, we assume that $\tau_{01,ij} > \tau_{02,ij}$.
	Note that \eqref{BF10_H1_3} is bounded below by
	\bean
	&&\frac{n}{2}\Big[ \log\Big( \frac{n_1}{n} \frac{\what{\tau}_{1,ij}}{\what{\tau}_{2,ij}} + \frac{n_2}{n}  + \frac{2 b_{0,ij}}{n\what{\tau}_{2,ij}}\Big) - \frac{n_1}{n}\log\Big( \frac{\what{\tau}_{1,ij}}{\what{\tau}_{2,ij}} + \frac{2 b_{01,ij}}{n_1 \what{\tau}_{2,ij}}\Big) - \frac{n_2}{n}\log\Big( 1+ \frac{2 b_{02,ij}}{n_2 \what{\tau}_{2,ij}}\Big)  \Big] \nonumber \\
	&\ge& \frac{n}{2} \Big[ \log\Big( \frac{n_1}{n} \frac{\what{\tau}_{1,ij}}{\what{\tau}_{2,ij}} + \frac{n_2}{n}  \Big) - \frac{n_1}{n}\log\Big( \frac{\what{\tau}_{1,ij}}{\what{\tau}_{2,ij}} \Big) - \frac{n_1}{n}\log\Big( 1+ \frac{2 b_{01,ij}}{n_1 \what{\tau}_{1,ij}}\Big) - \frac{n_2}{n}\log\Big( 1+ \frac{2 b_{02,ij}}{n_2 \what{\tau}_{2,ij}}\Big)  \Big] \nonumber \\
	&\ge& \frac{n}{2} \Big[ \log\Big( \frac{n_1}{n} \frac{\what{\tau}_{1,ij}}{\what{\tau}_{2,ij}} + \frac{n_2}{n}  \Big) - \frac{n_1}{n}\log\Big( \frac{\what{\tau}_{1,ij}}{\what{\tau}_{2,ij}} \Big) - \frac{2 (b_{01,ij} \vee b_{02,ij})}{n} \Big( \what{\tau}_{1,ij}^{-1} + \what{\tau}_{2,ij}^{-1}  \Big) \Big]  \nonumber \\
	&\equiv& \frac{n}{2} \Big[ \log\Big( 1+  \frac{n_1}{n} R_{ij}   \Big) - \frac{n_1}{n}\log\Big( 1 + R_{ij} \Big)   - \frac{2 (b_{01,ij} \vee b_{02,ij})}{n} \Big( \what{\tau}_{1,ij}^{-1} + \what{\tau}_{2,ij}^{-1}  \Big) \Big]  ,  \label{BF10_H1_5}
	\eean
	by the definition of $R_{ij}$.
	Consider the  sets $\{A_{1ij}, A_{2ij} \}$, which are defined at Lemma \ref{lem:setAij1-8}, with some constant $0<C < C_1$.
	Then by the similar arguments, on the event $A_{-0,ij}= \cap_{l=1}^2 A_{lij}$, 
	\bea
%	\frac{\what{\tau}_{1,ij,\gamma}}{\what{\tau}_{2,ij,\gamma}} 
%	&\le& \frac{\tau_{01,ij}}{\tau_{02,ij}} (1-n_2^{-1})  \Big( 1 + 2 \sqrt{C_2\epsilon_{01}}(1+ 3\sqrt{C_2\epsilon_{01}}) + 2 \sqrt{C_2\epsilon_{02}/(1-n_2^{-1})} + C' C_2 \sqrt{(\epsilon_{01}\vee \epsilon_{02}) \epsilon_{02}}\Big)    \\
%	&\le&  \frac{\tau_{01,ij}}{\tau_{02,ij}} (1-n_2^{-1}) \Big( 1+ 4 \sqrt{C_2' (\epsilon_{01}\vee \epsilon_{02})} \Big)        ,\\
	1 + R_{ij} \,\,=\,\, \frac{\what{\tau}_{1,ij}}{\what{\tau}_{2,ij}} 
	&\ge&  \frac{\tau_{01,ij}}{\tau_{02,ij}} \big( 1 -  4\sqrt{C_1 (\epsilon_{01} \vee \epsilon_{02})}  - O(\epsilon_{01}\vee \epsilon_{02}) \big) \\
	&\ge& 1 + C_{\rm bm} \sqrt{\epsilon_{01}} + O(\epsilon_{01}\vee \epsilon_{02}) \,\,=:\,\, 1 + R_{\min}
	\eea
	for all sufficiently large $n_1\wedge n_2$.
	The last inequality follows from condition (A3).
	It implies that $R_{ij}>0$ on the event $A_{-0,ij}$.
	Note that $\log(1+x n_1/n) - n_1 \log(1+x)/n$ is increasing in $x>0$.
	Thus, \eqref{BF10_H1_5} is bounded below by
	\bea
	&& \frac{n}{2}\Big[ \log \Big(1+ \frac{n_1}{n}R_{\min}\Big) - \frac{n_1}{n}\log\Big(1+R_{\min}\Big)   - \frac{2 (b_{01,ij}\vee b_{02,ij})}{n} \Big( \what{\tau}_{1,ij}^{-1} + \what{\tau}_{2,ij}^{-1}  \Big)  \Big] \\
	&\ge& \frac{n}{2} \Big[ \frac{ R_{\min}^2 n_1}{2n} - \frac{ R_{\min}^3 n_1}{3n} - \frac{R_{\min}^2 n_1^2}{2n^2} - \frac{n_1^3}{3n^3}\frac{R_{\min}^3}{1+ \frac{n_1}{n}R_{\min}}  - \frac{2 (b_{01,ij}\vee b_{02,ij})}{n} \Big( \what{\tau}_{1,ij}^{-1} + \what{\tau}_{2,ij}^{-1}  \Big) \Big] \\
	&\ge&  \frac{n_1 R_{\min}^2 }{4} \Big[ 1 - \frac{n_1}{n} - \frac{2}{3}R_{\min} - \frac{2}{3}\Big(\frac{n_1}{n}\Big)^2 R_{\min}  \Big] - (b_{01,ij}\vee b_{02,ij}) \Big( \what{\tau}_{1,ij}^{-1} + \what{\tau}_{2,ij}^{-1}  \Big) \\
	&\ge&  C_{\rm bm}^2 \frac{\log(n\vee p)}{4}\Big[1 -\frac{n_1}{n} - \frac{2}{3}\big(1 + \big(\frac{n_1}{n}\big)^2 \big)R_{\min}  \Big]  + o\big( \log(n\vee p) \big) \\
	&\ge&  C_{\rm bm}^2    \frac{\log(n\vee p)}{8} + o\big( \log(n\vee p) \big)  
	\eea
	for all sufficiently large $n_1 \wedge n_2$ because we assume that $\tau_{01,ij}\wedge \tau_{02,ij} \gg ( \log(n\vee p) )^{-1}$ by condition (A3).
	Note that, on the event $A_{-0,ij}$, \eqref{BF10_H1_00} is negligible compared with $\log (n\vee p)$ because we assume that $\tau_{01,ij}\vee \tau_{02,ij} \ll (n\vee p)$ by condition (A3).
	Thus, the leading terms in the lower bound of $\log B_{10}(\tilde{X}_i, \tilde{Y}_i, \tilde{X}_j, \tilde{Y}_j) $ is
	\bea
	&&\frac{C_{\rm bm}^2}{8} \log(n\vee p) + \frac{1}{2}\log \Big( \frac{\gamma}{1+\gamma}\Big) - \frac{1}{2}\log \Big(\frac{n_1}{2}+a_0 \Big) -  \log \Big(\frac{n_2}{2}+a_0 \Big) + \frac{1}{2}\log \Big(\frac{n}{2}+a_0 \Big) \\
	&\ge& \Big(  \frac{C_{\rm bm}^2}{8} - \alpha - 1 + o(1) \Big) \log(n\vee p) 
	\eea
	on the event $A_{-0,ij}$, which implies 
	\bea
	&&\bbP_0 \left( \log B^{\sg}_{\max,10}(\bfX_{n_1}, \bfY_{n_2}) \ge \big(  C_{{\rm bm}}^2/8 - \alpha - 1 + o(1)\big) \log(n\vee p) \right)  \\
	&\ge& \max_{(i,j): i\neq j} \bbP_0 \left( \log B_{10}(\tilde{X}_i, \tilde{Y}_i, \tilde{X}_j, \tilde{Y}_j) \ge \big(  C_{{\rm bm}}^2/8 - \alpha - 1 + o(1) \big) \log(n\vee p)  \right)  \\
	&\ge& \max_{(i,j): i\neq j} \bbP_0 (A_{-0,ij}) \\
	&\ge& 1 - 4 (n \vee p)^{-C} \,\, = \,\, 1- o(1)
	\eea
	for some constant $0<C<C_1$.
	It completes the proof because we assume that $C_{\rm bm}^2 > 8 (\alpha+1)$.
	If $\tau_{01,ij} < \tau_{02,ij}$, the same arguments hold by condition (A3).

	Now, we consider the case $a_{01,ij} \neq a_{02,ij}$, and suppose condition (A3$^\star$) holds.
%	We again focus on lower bounds for \eqref{BF10_H1_0}, \eqref{BF10_H1_1} and \eqref{BF10_H1_2}.
	Define the sets
	\bea
	A_{3ij} &=& \Big\{ \bfZ_n: \|\tilde{X}_i\|_2^2 \le \sigma_{01,ii} n_1(1 + 2\sqrt{C \epsilon_{01}} + 2 C\epsilon_{01} )   \Big\} , \\
	A_{4ij} &=& \Big\{ \bfZ_n: \|\tilde{Y}_i\|_2^2 \le \sigma_{02,ii} n_2(1 + 2\sqrt{C \epsilon_{02}} + 2 C\epsilon_{02} )   \Big\}  , \\
	A_{5ij} &=& \Big\{ \bfZ_n: n_1(1 - 2 \sqrt{C \epsilon_{01}})  \le \sigma_{01,jj}^{-1}\|\tilde{X}_j\|_2^2 \le n_1(1 + 2\sqrt{C \epsilon_{01}} + 2 C\epsilon_{01} )   \Big\} , \\
	A_{6ij} &=& \Big\{ \bfZ_n: n_2(1 - 2 \sqrt{C \epsilon_{02}})  \le\sigma_{02,jj}^{-1} \|\tilde{Y}_j\|_2^2 \le  n_2(1 + 2\sqrt{C \epsilon_{02}} + 2 C\epsilon_{02} )   \Big\} 
	\eea
	for some constant $0<C <C_1$, and let $A_{ij}' = \cap_{l=1}^{6} A_{lij}$.
	Then, $\bbP_0((A_{ij}')^c ) = o(1)$ by Lemma \ref{lem:chi_bounds}, so we can focus on the set $A_{ij}'$.
	On the set $A_{ij}'$, \eqref{BF10_H1_0} is bounded below by
	\bea
	&& -a_0 \Big[ \log \Big( \|\tilde{X}_i\|_2^2 + \|\tilde{Y}_i\|_2^2 + 2 b_{0,ij}  \Big)  - \log \Big(n+ 2a_0 \Big)  \Big] \nonumber \\
	&\ge& - a_0 \Big[ \log \Big( \sigma_{01,ii} n_1 \big\{ 1 + 2\sqrt{C \epsilon_{01}} + 2 C\epsilon_{01} \big\} + \sigma_{02,ii} n_2 \big\{ 1 + 2\sqrt{C \epsilon_{02}} + 2 C\epsilon_{02}\big\} + 4 b_{0,ij}  \Big)  - \log \Big(n+ 2a_0 \Big) \Big] \nonumber \\
	&\ge& - a_0 \Big[ \log \Big( \sigma_{01,ii} n_1(1 + 4\sqrt{C \epsilon_{01}} ) + \sigma_{02,ii} n_2(1 + 4\sqrt{C \epsilon_{02}} ) + 4 b_{0,ij}  \Big)  - \log \Big(n+ 2a_0 \Big) \Big] \nonumber \\
	&\ge& - a_0   \log \Big( \sigma_{01,ii} (1 + 4\sqrt{C \epsilon_{01}} ) + \sigma_{02,ii} (1 + 4\sqrt{C \epsilon_{02}} ) + 4 b_{0,ij} n^{-1}  \Big) , \nonumber
	\eea
	which is smaller than $- a_0 \log (n\vee p)$ because $\sigma_{01,ii} \vee \sigma_{02,ii} \ll (n\vee p)$ by condition (A3$^\star$).

	Note that \eqref{BF10_H1_1} is bounded below by
	\bea
	&&\Big( \frac{n_1}{2}+a_0\Big) \log \Big( 1 - \frac{ \frac{n_1\what{\tau}_{1,ij}+2 b_{01,ij} }{n_1+2a_0} - \frac{n\what{\tau}_{ij}+2 b_{0,ij} }{n+2a_0} }{ \frac{n_1\what{\tau}_{1,ij}+2 b_{01,ij} }{n_1+2a_0} }  \Big) \\
	&\ge& \Big( \frac{n_1}{2}+a_0\Big) \,\frac{ \frac{n\what{\tau}_{ij}+2 b_{0,ij} }{n+2a_0} - \frac{n_1\what{\tau}_{1,ij}+2 b_{01,ij} }{n_1+2a_0}  }{ \frac{n\what{\tau}_{ij}+2 b_{0,ij} }{n+2a_0} }  \\
	&=& \Big( \frac{n_1}{2}+a_0\Big) \Bigg\{ \frac{\what{\tau}_{ij} -\what{\tau}_{1,ij}  }{\what{\tau}_{ij} + \frac{2 b_{0,ij}}{n} } - \frac{2}{\what{\tau}_{ij}+ \frac{2 b_{0,ij}}{n} } \Big(\frac{ b_{01,ij}(1 + \frac{2a_0}{n}) - b_{0,ij} (\frac{n_1}{n} + \frac{2a_0}{n})  }{n_1+2a_0}  - \what{\tau}_{1,ij}\frac{a_0(1-\frac{n_1}{n})}{n_1+2a_0} \Big)   \Bigg\}
	\eea
	where the inequality follows from $\log(1-x) \ge -x/(1-x)$ for any $x<1$.
	A lower bound for \eqref{BF10_H1_2} can be derived similarly.
	Thus, the sum of \eqref{BF10_H1_1} and \eqref{BF10_H1_2} is bounded below by
	\bean
	&& \frac{1}{2} \frac{n\what{\tau}_{ij} - n_1\what{\tau}_{1,ij} - n_2\what{\tau}_{2,ij} }{\what{\tau}_{ij} + \frac{2 b_{0,ij}}{n} }  \nonumber   \\
	&+&   \frac{1}{\what{\tau}_{ij} + \frac{2 b_{0,ij}}{n}}\left[  a_0\Big( 2\what{\tau}_{ij} - \what{\tau}_{1,ij} - \what{\tau}_{2,ij} \Big) -  2 b_{0,ij}  + \what{\tau}_{1,ij} a_0 \Big(1- \frac{n_1}{n}\Big) + \what{\tau}_{2,ij} a_0 \Big(1- \frac{n_2}{n}\Big)  \right]  \nonumber \\
	&\ge& \frac{1}{\what{\tau}_{ij} + \frac{2 b_{0,ij}}{n}} \left[  \Big(\frac{1}{2} + \frac{a_0}{n}  \Big)\big( n\what{\tau}_{ij} - n_1\what{\tau}_{1,ij} - n_2\what{\tau}_{2,ij}\big)  - 2 b_{0,ij}    \right]  \label{delta1} 
	\eean
	for all large $n_1 \wedge n_2$.
%	Note that 
%	\bea
%	n_1(\what{\tau}_{1,ij,\gamma}-\what{\tau}_{1,ij})  
%	&=& \frac{\gamma}{1+\gamma}  \tilde{X}_i^T H_{\tilde{X}_j}\tilde{X}_i \\
%	&\le& \gamma \tau_{01,ij} \cdot \tau_{01,ij}^{-1} \tilde{X}_i^T H_{\tilde{X}_j}\tilde{X}_i \\
%	&\le& \tau_{01,ij} C \epsilon_{01}
%	\eea
% 	on $A_{2ij} \cap A_{4ij}$.
%	Similarly, we have
%	\bea
%	n_1(\what{\tau}_{1,ij,\gamma}-\what{\tau}_{1,ij})  
%	&\le& \tau_{02,ij} C \epsilon_{02}
%	\eea
%	on $A_{6ij}\cap A_{8ij}$.
	Define
	\bea
	\lambda &:=& \frac{(a_{01,ij}- a_{02,ij})^2 \|\tilde{X}_j\|_2^2 \|\tilde{Y}_j\|_2^2 }{\tau_{01,ij}\|\tilde{Y}_j\|_2^2 + \tau_{02,ij}\|\tilde{X}_j\|_2^2 }.
	\eea
	It is easy to see that
	\bea
	n\what{\tau}_{ij} - n_1\what{\tau}_{1,ij} - n_2\what{\tau}_{2,ij} &=&
	\frac{1}{\|\tilde{Z}_j\|_2^2}\Big(\|\tilde{Y}_j\|_2 \frac{\tilde{X}_i^T\tilde{X}_j}{\|\tilde{X}_j\|_2} - \|\tilde{X}_j\|_2 \frac{\tilde{Y}_i^T\tilde{Y}_j}{\|\tilde{Y}_j\|_2}\Big)^2 \\
	&\overset{d}{\equiv}& \frac{ \tau_{01,ij}\|\tilde{Y}_j\|_2^2 + \tau_{02,ij}\|\tilde{X}_j\|_2^2}{\|\tilde{Z}_j\|_2^2} \cdot \chi_1^2 \left( \lambda \right)
	\eea
	given $\tilde{Z}_j$ because  
	\bea
	\frac{1}{\|\tilde{Z}_j  \|_2^2} \Big(\|\tilde{Y}_j\|_2 \frac{\tilde{X}_i^T\tilde{X}_j}{\|\tilde{X}_j\|_2} - \|\tilde{X}_j\|_2 \frac{\tilde{Y}_i^T\tilde{Y}_j}{\|\tilde{Y}_j\|_2}\Big) \,\,|\, \tilde{Z}_j &\sim&  N\Big( \frac{(a_{01,ij}-a_{02,ij})\|\tilde{X}_j\|_2 \|\tilde{Y}_j\|_2 }{\|\tilde{Z}_j  \|_2}  ,\, \frac{\tau_{01,ij}\|\tilde{Y}_j\|_2^2 + \tau_{02,ij}\|\tilde{X}_j\|_2^2 }{\|\tilde{Z}_j\|_2^2} \Big).
	\eea
	Then, on the set $A_{ij}'$,
	\bean
	\lambda &=& \frac{(a_{01,ij}- a_{02,ij})^2}{\tau_{01,ij}\|\tilde{X}_j\|_2^{-2} + \tau_{02,ij}\|\tilde{Y}_j\|_2^{-2} } \nonumber\\
	&\ge& \frac{(a_{01,ij}- a_{02,ij})^2}{\tau_{01,ij}(\sigma_{01,jj}n_1(1-2\sqrt{C\epsilon_{01}}) )^{-2} + \tau_{02,ij}(\sigma_{02,jj}n_2(1-2\sqrt{C\epsilon_{02}}))^{-2} }\nonumber \\
	&\ge& \frac{25}{2} C_1 \log (n\vee p)  \label{lam_lb}
	\eean
	by condition (A3$^\star$).
	By Lemma \ref{lem:chi_bounds}, with probability at least $1-(n\vee p)^{-C}$ for the constant $0<C<C_1$ used to define the set $A_{ij}'$,
	\bea
	n\what{\tau}_{ij} - n_1\what{\tau}_{1,ij} - n_2\what{\tau}_{2,ij}
	&\ge& \frac{ \tau_{01,ij}\|\tilde{Y}_j\|_2^2 + \tau_{02,ij}\|\tilde{X}_j\|_2^2}{\|\tilde{Z}_j\|_2^2} \Big\{  1 + \lambda - 2 \sqrt{(1+2\lambda)C\log(n\vee p)}  \Big\}\\
	&\ge& \frac{ \tau_{01,ij}\|\tilde{Y}_j\|_2^2 + \tau_{02,ij}\|\tilde{X}_j\|_2^2}{\|\tilde{Z}_j\|_2^2} \Big\{ 1 + \frac{\lambda}{5}  \Big\} \\
	&\ge& \frac{(a_{01,ij}- a_{02,ij})^2 }{5} \frac{\|\tilde{X}_j\|_2^2 \|\tilde{Y}_j\|_2^2 }{\|\tilde{Z}_j\|_2^2}\\
	&\ge& \frac{(a_{01,ij}- a_{02,ij})^2 }{5} \Big\{ \sum_{k=1}^2\big(\sigma_{0k,jj}n_k(1-2\sqrt{C_1 \epsilon_{0k}}) \big)^{-1}  \Big\}^{-1}
	\eea
	where the second inequality holds by \eqref{lam_lb}.
	Furthermore, on the set $A_{ij}'$,
	\bea
	\what{\tau}_{ij,\gamma} 
	&\le& \sigma_{01,ii} (1+ 4\sqrt{C\epsilon_{01}}) + \sigma_{02,ii} (1+ 4\sqrt{C\epsilon_{02}}).
	\eea
	Then, \eqref{delta1} is bounded below by
	\bea
	&& \frac{1}{\sigma_{01,ii} (1+ 4\sqrt{C\epsilon_{01}}) + \sigma_{02,ii} (1+ 4\sqrt{C\epsilon_{02}}) + \frac{2 b_{0,ij}}{n} } \\
	&\times& \Big\{ \Big(\frac{1}{2}+ \frac{a_0}{n}\Big) \frac{(a_{01,ij}- a_{02,ij})^2 }{5} \Big\{ \sum_{k=1}^2\big(\sigma_{0k,jj}n_k(1-2\sqrt{C_1 \epsilon_{0k}}) \big)^{-1}  \Big\}^{-1} - 2 b_{0,ij}   \Big\} \\
	&\ge& C_{{\rm bm}, a} \log (n\vee p)
	\eea
	by condition (A3$^\star$).
	Thus, the leading terms in the lower bound of $\log B_{10}(\tilde{X}_i, \tilde{Y}_i, \tilde{X}_j, \tilde{Y}_j) $ is
	\bea
	&& \big( C_{{\rm bm}, a} - a_0\big) \log (n \vee p) + \frac{1}{2}\log \Big( \frac{\gamma}{1+\gamma}\Big) - \frac{1}{2}\log \Big(\frac{n_1}{2}+a_0 \Big) -  \log \Big(\frac{n_2}{2}+a_0 \Big) + \frac{1}{2}\log \Big(\frac{n}{2}+a_0 \Big) \\
	&\ge& \Big(  C_{{\rm bm}, a} - a_0- \alpha - 1 \Big) \log(n\vee p) 
	\eea
	on the set $A_{ij}'$, which implies
	\bea
	&&\bbP_0 \left( \log B^{\sg}_{\max,10}(\bfX_{n_1}, \bfY_{n_2}) \ge \big(  C_{{\rm bm}, a} - a_0- \alpha - 1 \big) \log(n\vee p) \right)  \\
	&\ge& \max_{(i,j): i\neq j} \bbP_0 \left( \log B_{10}(\tilde{X}_i, \tilde{Y}_i, \tilde{X}_j, \tilde{Y}_j) \ge \big(  C_{{\rm bm}, a} - a_0- \alpha - 1 \big) \log(n\vee p)  \right)  \\
	&\ge& \max_{(i,j): i\neq j} \bbP_0 (A_{ij}')  \\
	&\ge& 1 - 14 (n \vee p)^{-C} \,\, = \,\, 1- o(1)
	\eea
	for some constant $0<C<C_1$.
	It completes the proof because we assume that $C_{{\rm bm}, a} > a_0+\alpha +1$.
	
\end{proof}

\subsection{Proof of Theorem \ref{thm:two_cov_lowerbound}}
\begin{proof}
	Let 
	\bea
	\calU(C) &:=& \Big\{ (\sg_1,\sg_2): \max_{1\le i\le j\le p}\frac{ (\sigma_{1,ij} - \sigma_{1,ij})^2 }{ \theta_{1,ij}/n_1 + \theta_{2,ij}/n_2 } \ge C  \, \log p   \Big\} ,
	\eea 
	where $\theta_{1,ij} = \sigma_{1,ii}\sigma_{1,jj} + \sigma_{1,ij}^2$ and $\theta_{2,ij}=\sigma_{2,ii}\sigma_{2,jj} + \sigma_{2,ij}^2$, then $\widetilde{H}_1(C_\star,c_0) \subset \calU(C)$ for some $C_\star, c_0$ and $C>0$.
	By Theorem 3, the proof of Theorem 4 in \cite{cai2013two} and arguments in \cite{baraud2002non} (p. 595), we have
	$$\inf_{(\sg_{01},\sg_{02}) \in \calU(C)} \sup_{\phi_{\alpha_0} \in \calT_{\alpha_0}} \bbE_{\sg_{01},\sg_{02}} \phi_{\alpha_0} \le \alpha + o(1)$$ 
	for some constant $C>0$, where $\calT_{\alpha_0}$ is the set of $\alpha_0$-level tests over the multivariate normal distributions with $\alpha_0>0$, that is, $\bbE_0 \phi_{\alpha_0} \le \alpha_0$.
	In the proof of Theorem 4 in \cite{cai2013two}, the infimum of $(\sg_{01},\sg_{02})$ is essentially taken over a subset of $\widetilde{H}_1(C_\star,c_0)$ for some small $C_\star>0$ and any $c_0>0$.
	Hence, because $p \ge n^c$ for some constant $c>0$,  we have
	\bea
	\inf_{(\sg_{01},\sg_{02}) \in \widetilde{H}_1(C_\star,c_0)} \sup_{\phi_{\alpha_0} \in \calT_{\alpha_0}} \bbE_{\sg_{01},\sg_{02}} \phi_{\alpha_0} &\le& {\alpha_0} + o(1)
	\eea
	for some small $C_\star>0$.
	Note that $\widetilde{H}_1(C_\star ,c_0) \subset {H}_1(C_{\rm bm}, C_{{\rm bm},a})$ for some small constants $C_{\rm bm}$ and $C_{{\rm bm},a} >0$.
	Furthermore, we have $\calT \subset \calT_{\alpha_0}$ for any $\alpha_0 >0$.
	Therefore, for any $\alpha_0 >0$,
	\bea
	\inf_{(\sg_{01},\sg_{02}) \in {H}_1(C_{\rm bm}, C_{{\rm bm},a})} \sup_{\phi \in \calT} \bbE_{\sg_{01},\sg_{02}} \phi &\le& {\alpha_0} + o(1).
	\eea
\end{proof}

\newpage

\section{Auxiliary results}\label{sec:auxil}
\begin{lemma}\label{lem:testable_two_cov}
	Consider two different $p\times p$ covariances $\sg_{01} = (\sigma_{01,ij})$ and $\sg_{02} = (\sigma_{02,ij})$ such that 
	\bea
	\begin{split}
		\max_{1\le k \le 2}\max_{1\le i\neq j \le p}\rho_{0k,ij}^2 &\le 1 - c_0 \\
		\{\log(n\vee p) \}^{-1} \ll \min_{1\le k \le 2}\min_{1\le i \le p} \sigma_{0k,ii} &\le \max_{1\le k \le 2}\max_{1\le i \le p} \sigma_{0k,ii} \ll (n\vee p)
	\end{split}
	\eea
	for some small constant $c_0>0$.
	If 
	\bea
	\max_{1\le i\le j\le p}\frac{ (\sigma_{01,ij} - \sigma_{01,ij})^2 }{ \sigma_{01,ii}\sigma_{01,jj} + \sigma_{02,ii}\sigma_{02,jj} } &\ge& C_\star  \, \frac{\log (n\vee p)}{n}
	\eea
	for some large constant $C_\star>0$, then condition (A3) (or (A3$^\star$)) holds with $\alpha>1$ and $n_1\asymp n_2$.
\end{lemma}

\begin{proof}
	Throughout the proof, suppose $\alpha>1$ and $n_1\asymp n_2$.
	Note that if $\frac{\sigma_{01,jj}}{\sigma_{02,jj}} \vee \frac{\sigma_{02,jj}}{\sigma_{01,jj}} \lra \infty$ for some $j$ as $n\to\infty$, condition (A3) is met for some pair of indices because $\max_{1\le k \le 2}\max_{1\le i\neq j \le p}\rho_{0k,ij}^2 \le 1- c_0$ for some constant $c_0>0$.
	Since it trivially completes the proof, suppose $\frac{\sigma_{01,jj}}{\sigma_{02,jj}} \vee \frac{\sigma_{02,jj}}{\sigma_{01,jj}}  < C$ for all $j$ and some constant $C>0$.
	Let
	\bean\label{cond_i0j0}
	\max_{1\le i\le j\le p}\frac{ (\sigma_{01,ij} - \sigma_{01,ij})^2 }{ \sigma_{01,ii}\sigma_{01,jj} + \sigma_{02,ii}\sigma_{02,jj} } \equiv \frac{ (\sigma_{01,i_0 j_0} - \sigma_{01,i_0 j_0})^2 }{ \sigma_{01,i_0 i_0}\sigma_{01,j_0 j_0} + \sigma_{02,i_0 i_0}\sigma_{02,j_0 j_0} }
	\ge C_\star  \, \frac{\log (n\vee p)}{n}
	\eean
	for some $(i_0,j_0)$.

	First, suppose that $i_0 = j_0$.
	Without loss of generality, assume that $\sigma_{01,i_0 i_0} > \sigma_{02,i_0 i_0}$.
	If condition (A3$^\star$) is satisfied for some $(i_0,j)$ or $(j,i_0)$, the proof is completed.
	Thus, now we assume that condition (A3$^\star$) does not hold for $(i_0,j)$ and $(j,i_0)$ for all $j\neq i_0$.
	Since $\sigma_{01,i_0 i_0} \ll (n\vee p)$ always holds, at least one of inequalities \eqref{diff_aij1} and \eqref{diff_aij2} for $(i_0,j)$ and $(j,i_0)$ does not hold.
	Note that 
	\bea
	\frac{\tau_{01,i_0 j}}{\tau_{02,i_0 j}}
	&=& \frac{\sigma_{01,i_0 i_0}}{\sigma_{02,i_0 i_0}} \frac{1 - a_{01,i_0 j} a_{01,j i_0}}{1 - a_{02,i_0 j} a_{02,j i_0}} \\
	&\equiv& \frac{\sigma_{01,i_0 i_0}}{\sigma_{02,i_0 i_0}} \frac{1 - (a_{02,i_0 j}+ \Delta_1) (a_{02,j i_0} + \Delta_2)}{1 - a_{02,i_0 j} a_{02,j i_0}} \\
	&=& \frac{\sigma_{01,i_0 i_0}}{\sigma_{02,i_0 i_0}}  \Big\{ 1 - \frac{\Delta_1 a_{02,j i_0} + \Delta_2 a_{02,i_0 j} + \Delta_1 \Delta_2 }{1 - a_{02,i_0 j} a_{02,j i_0}} \Big\} .
	\eea
	Since we assume that at least one of inequalities \eqref{diff_aij1} and \eqref{diff_aij2} for $(i_0,j)$ and $(j,i_0)$ does not hold,
	\bea
	|\Delta_1| &\lesssim& \left\{ \Big( \frac{\sigma_{01,i_0 i_0}}{\sigma_{01,jj}} + \frac{\sigma_{02,i_0 i_0}}{\sigma_{02,jj}} \Big)\frac{\log(n\vee p)}{n} \right\}^{\frac{1}{2}} \vee \left\{ \Big( \frac{1}{\sigma_{01,jj}} + \frac{1}{\sigma_{02,jj}} \Big)\big(\sigma_{01,i_0 i_0} + \sigma_{02,i_0 i_0} \big)\frac{\log(n\vee p)}{n} \right\}^{\frac{1}{2}}, \\
	|\Delta_2| &\lesssim& \left\{ \Big( \frac{\sigma_{01,jj}}{\sigma_{01,i_0 i_0}} + \frac{\sigma_{02,jj}}{\sigma_{02,i_0 i_0}} \Big)\frac{\log(n\vee p)}{n} \right\}^{\frac{1}{2}} \vee \left\{ \Big( \frac{1}{\sigma_{01,i_0 i_0}} + \frac{1}{\sigma_{02,i_0 i_0}} \Big)\big(\sigma_{01,jj} + \sigma_{02,jj} \big)\frac{\log(n\vee p)}{n} \right\}^{\frac{1}{2}}.
	\eea
	Because we assume that $\frac{\sigma_{01,jj}}{\sigma_{02,jj}} \vee \frac{\sigma_{02,jj}}{\sigma_{01,jj}}  < C$ for all $j$, we have $|\Delta_1 \Delta_2| \lesssim \log(n\vee p)/n$.
	Furthermore,	
	\bea
	&& |\Delta_1 a_{02, ji_0}| \\
	&\le&  \frac{|\sigma_{02,j i_0}|}{\sigma_{02,i_0 i_0}} |\Delta_1|  \\
	&\le& \Big( \frac{\sigma_{02,j j}}{\sigma_{02,i_0 i_0}} \Big)^{\frac{1}{2}} |\Delta_1|  \\
	&\lesssim& \Big( \frac{\sigma_{02,j j}}{\sigma_{02,i_0 i_0}} \Big)^{\frac{1}{2}} \left\{  \Big[ \Big( \frac{\tau_{01,i_0 j}}{\sigma_{01,jj}} + \frac{\tau_{02,i_0 j}}{\sigma_{02,jj}} \Big) \frac{\log(n\vee p)}{n} \Big]^{\frac{1}{2}} \vee \Big[ \Big( \frac{1}{\sigma_{01,jj}} +\frac{1}{\sigma_{02,jj}} \Big)(\sigma_{01,i_0 i_0} + \sigma_{02,i_0 i_0})\frac{\log(n\vee p)}{n} \Big]^{\frac{1}{2}} \right\} \\
	&\le&  \Big[ \Big( \frac{\sigma_{02,j j}}{\sigma_{02,i_0 i_0}} \frac{\sigma_{01,i_0 i_0}}{\sigma_{01,jj}} + 1 \Big) \frac{\log(n\vee p)}{n} \Big]^{\frac{1}{2}} \vee \Big[ \Big( \frac{\sigma_{02,jj}}{\sigma_{01,jj}} +1\Big)\Big(\frac{\sigma_{01,i_0 i_0}}{\sigma_{02,i_0 i_0}} + 1\Big)\frac{\log(n\vee p)}{n} \Big]^{\frac{1}{2}} \\
	&\lesssim& \sqrt{\frac{\log(n\vee p)}{n}}.
	\eea
	Similarly, we have $|\Delta_2 a_{02, i_0 j}| \lesssim \sqrt{\log(n\vee p)/n}$. 	
	Then,
	\bea
	\Big|\frac{\Delta_1 a_{02,j i_0} + \Delta_2 a_{02,i_0 j} + \Delta_1 \Delta_2 }{1 - a_{02,i_0 j} a_{02,j i_0}} \Big|
	&\lesssim& \sqrt{\frac{\log(n\vee p)}{n}} + |\Delta_1 \Delta_2| \,\,\lesssim\,\, \sqrt{\frac{\log(n\vee p)}{n}} .
	\eea	
	Since \eqref{cond_i0j0} implies 
	\bea
	\frac{\sigma_{01,i_0 i_0}}{\sigma_{02,i_0 i_0}} 
	&\ge& 1 + \sqrt{C_{\star}}\frac{\sigma_{01,i_0 i_0}}{\sigma_{02,i_0 i_0}}  \sqrt{\frac{\log(n\vee p)}{n}} \,>\, 1 + \sqrt{C_{\star}}  \sqrt{\frac{\log(n\vee p)}{n}},
	\eea
	it implies that condition (A3) holds for $(i_0,j)$ for some large $C_\star>0$.
	
	Now, suppose that $i_0 \neq j_0$.
%	Without loss of generality, assume that $\sigma_{01,i_0 j_0} > \sigma_{02,i_0 j_0}$.
	If $\sigma_{01,i_0 i_0} = \sigma_{02,i_0 i_0}$, we have
	\bea
	(a_{01,j_0 i_0} - a_{02,j_0 i_0})^2 &=& \frac{1}{\sigma_{01,i_0 i_0}^2} ( \sigma_{01, i_0 j_0} - \sigma_{02, i_0 j_0} )^2 \\
	&\ge& \frac{\sigma_{01,i_0 i_0}\sigma_{01,j_0j_0} + \sigma_{01,i_0 i_0}\sigma_{02,j_0j_0} }{\sigma_{01,i_0 i_0}^2} \, C_\star  \frac{\log(n\vee p)}{n} \\
	&=& \frac{\sigma_{01,j_0j_0} + \sigma_{02,j_0j_0} }{\sigma_{01,i_0 i_0}} \, C_\star  \frac{\log(n\vee p)}{n} .
	\eea
	Then, condition (A3$^\star$) holds for $(j_0, i_0)$ and some large $C_\star>0$.
	Similarly, if $\sigma_{01,j_0 j_0} = \sigma_{02,j_0 j_0}$, then condition (A3$^\star$) holds for $(i_0, j_0)$ and some large $C_\star>0$.
	Thus, we only need to consider the case $\sigma_{01,i_0 i_0} \neq \sigma_{02,i_0 i_0}$ and $\sigma_{01,j_0 j_0} \neq \sigma_{02,j_0 j_0}$.
%	Suppose condition (A3$^\star$) does not hold for $(i_0,j_0)$ and $(j_0,i_0)$ because, otherwise, the proof is completed.
	Without loss of generality, suppose $\sigma_{01,i_0 i_0} > \sigma_{02,i_0 i_0}$.
	Note that if $\sigma_{01,i_0 i_0}/\sigma_{02,i_0 i_0} > 1 + C \sqrt{\log(n\vee p)/n}$ for some large constant $C>0$,  condition (A3) is met as shown in the previous paragraph.
	If $1 < \sigma_{01,i_0 i_0}/\sigma_{02,i_0 i_0} \le 1 + C \sqrt{\log(n\vee p)/n}$ and $\sigma_{01,i_0 i_0}/\sigma_{02,i_0 i_0} \equiv 1 + \Delta_3$ with $\Delta_3>0$, we have
	\bea
	( a_{01,j_0 i_0} - a_{02,j_0 i_0} )^2
	&=& \Big( \frac{\sigma_{01,i_0 j_0}}{\sigma_{01,i_0 i_0}} - \frac{\sigma_{02,i_0 j_0}}{\sigma_{02,i_0 i_0}}  \Big)^2 \\
	&=& \Big( \frac{\sigma_{01,i_0 j_0}}{\sigma_{01,i_0 i_0}} - \frac{\sigma_{02,i_0 j_0}}{\sigma_{01,i_0 i_0}}\frac{\sigma_{01,i_0 i_0}}{\sigma_{02,i_0 i_0}}  \Big)^2 \\
	&=& \frac{1}{\sigma_{01,i_0 i_0}^2}\Big( \sigma_{01,i_0 j_0} - \sigma_{02,i_0 j_0} - \sigma_{02,i_0 j_0} \Delta_3   \Big)^2 \\
	&\ge& \frac{1}{\sigma_{01,i_0 i_0}^2}\Big( |\sigma_{01,i_0 j_0} - \sigma_{02,i_0 j_0}| - |\sigma_{02,i_0 j_0} \Delta_3|   \Big)^2 \\
	&\ge& \frac{1}{\sigma_{01,i_0 i_0}^2}\Big( |\sigma_{01,i_0 j_0} - \sigma_{02,i_0 j_0}| - |\sqrt{\sigma_{02,i_0 i_0} \sigma_{02,j_0 j_0}} \Delta_3|   \Big)^2 \\
	&\gtrsim& C_\star \big(\frac{\sigma_{01,j_0j_0}}{\sigma_{01,i_0 i_0}} + \frac{\sigma_{02,i_0 i_0}\sigma_{02,j_0j_0}}{\sigma_{01,i_0 i_0}^2} \big) \frac{\log(n\vee p)}{n}  \\
	&\ge& C_\star \big(\frac{\sigma_{01,j_0j_0}}{\sigma_{01,i_0 i_0}} + C'\frac{\sigma_{02,j_0j_0}}{\sigma_{02,i_0 i_0}} \big) \frac{\log(n\vee p)}{n}
	\eea
	for some large constant $C_\star>0$ and constant $C'>0$.
	Thus, condition (A3$^\star$) holds for $(j_0,i_0)$ and some large constant $C_\star>0$, and it completes the proof.
	
\end{proof}

\bibliographystyle{dcu}
\bibliography{two-sample-tests}

\end{document}